\documentclass[pra,aps,floatfix,amsmath,superscriptaddress,twocolumn,longbibliography,nofootinbib]{revtex4-2}
\usepackage{wrapfig, blindtext}
\usepackage[most]{tcolorbox}
\tcbset{colback=yellow!10!white, colframe=red!50!black, 
  highlight math style= {enhanced, 
    colframe=red,colback=red!10!white,boxsep=0pt}
}
\usepackage{amsmath}
\usepackage{tikz-cd}
\usepackage{empheq}
\usepackage{graphicx}

\usepackage{anyfontsize}
\usepackage{hyperref}
\usepackage[capitalise]{cleveref}
\usepackage{nameref}
\usepackage{xpatch}
\makeatletter
\xpatchcmd{\@ssect@ltx}{\@xsect}{\protected@edef\@currentlabelname{#8}\@xsect}{}{}
\xpatchcmd{\@sect@ltx}{\@xsect}{\protected@edef\@currentlabelname{#8}\@xsect}{}{}
\makeatother
\usepackage{youngtab}

\usepackage{tikz}
\usetikzlibrary{decorations.pathreplacing}

\newcommand{\U}{\operatorname{U}}
\newcommand{\SU}{\operatorname{SU}}

\newcommand\End{\operatorname{L}}
\newcommand{\e}{\operatorname{e}}
\renewcommand{\d}{\mathrm{d}}
\renewcommand{\i}{\mathrm{i}}
\newcommand{\pb}{\mathbf{p}}
\renewcommand{\P}{\mathbf{P}}
\renewcommand{\H}{\mathcal{H}}
\newcommand{\fermi}{\mathrm{fermi}}
\newcommand{\comp}{{\mathrm{comp}}}

\newcommand{\sgn}{\operatorname{sgn}}

\usepackage{makecell}
\newcommand\TopRule{\Xhline{0.08em}}

\newcommand\MidRule{\Xhline{0.03em}}
\newcommand\BotRule{\Xhline{0.08em}}

\makeatletter 
\renewcommand\onecolumngrid{%
  \do@columngrid{one}{\@ne}%
  \def\set@footnotewidth{\onecolumngrid}%
  \def\footnoterule{\kern-6pt\hrule width 1.5in\kern6pt}%
}
\makeatother 

\newlength\dlf 

\usepackage{amsmath,amssymb}
\usepackage{amssymb,enumerate}
\usepackage{graphicx}
\usepackage{graphics}
\usepackage{amsmath}
\usepackage{amsthm,bbm}
\usepackage{color}
\usepackage{dsfont}
\usepackage{bbm}

\usepackage{lipsum}
\usepackage{lmodern}
\usepackage{tcolorbox}

\usepackage{amsfonts}
\usepackage{graphicx,graphics,epsfig,times,bm,bbm,amssymb,amsmath,amsfonts,mathrsfs}
\usepackage[normalem]{ulem}
\usepackage{subfigure}
\usepackage{dsfont}
\usepackage{braket}
\usepackage{upgreek }
\usepackage{tikz}
\usepackage{natbib}
\usepackage{chngcntr}

\newtheorem{theorem}{Theorem}

\newtheorem{corollary}{Corollary}

\newtheorem{lemma}{Lemma}
\newtheorem*{lemma*}{Lemma}

\newtheorem{proposition}{Proposition}

\theoremstyle{remark}
\newtheorem{remark}{Remark}

\newcommand{\bes} {\begin{subequations}}
\newcommand{\ees} {\end{subequations}}
\newcommand{\bea} {\begin{eqnarray}}
\newcommand{\eea} {\end{eqnarray}}
\newcommand{\be} {\begin{equation}}
\newcommand{\ee} {\end{equation}}

\def\>{\rangle}
\def\<{\langle}
\def\Tr{\operatorname{Tr}}

\newcommand{\ignore}[1]{}

\usepackage{amsmath,amssymb}

\newcommand{\real}{\mathbb{R}}
\newcommand{\complex}{\mathbb{C}}
\newcommand{\hilbert}[1][H]{\mathcal{#1}}
\newcommand{\identity}{\mathbb{I}}

\usepackage{mathtools}
\usepackage{microtype}

\DeclarePairedDelimiter{\p}{(}{)}
\DeclarePairedDelimiter{\sets}{\{}{\}}
\newcommand{\cset}[2]{\sets{#1 \, : \, #2}}

\DeclarePairedDelimiter{\nket}{\lvert}{\rangle}
\DeclarePairedDelimiter{\nbra}{\langle}{\rvert}
\newcommand{\qproj}[1]{\nket{#1} \nbra{#1}}

\usepackage{ytableau}
\ytableausetup{smalltableaux}

\crefname{section}{Sec.}{Secs.}
\crefname{claim}{Claim}{Claims}

\begin{document}
\title{Qudit circuits with SU(d) symmetry: Locality imposes additional conservation laws}

\author{Iman Marvian}
\affiliation{Departments of Physics, Duke University, Durham, NC 27708, USA}
\affiliation{Department of Electrical and Computer Engineering, Duke University, Durham, NC 27708, USA}
\author{Hanqing Liu}
\affiliation{Departments of Physics, Duke University, Durham, NC 27708, USA}
\author{Austin Hulse}
\affiliation{Departments of Physics, Duke University, Durham, NC 27708, USA}

\begin{abstract}
  Local symmetric quantum circuits  provide a simple framework to study the dynamics and phases of complex quantum systems with conserved charges. However, some of their basic properties have not yet been understood. Recently, it has been shown that such quantum circuits only generate a restricted subset of symmetric unitary transformations [I. Marvian, Nature Physics, 2022]. In this paper, we consider circuits with 2-local SU$(d)$-invariant unitaries acting on qudits, i.e., $d$-dimensional quantum systems. Our results reveal a significant distinction between the cases of $d=2$ and $d\geq 3$. For qubits with SU$(2)$ symmetry, arbitrary global rotationally-invariant unitaries can be generated with 2-local ones, up to relative phases between the subspaces corresponding to inequivalent irreducible representations (irreps) of the symmetry, i.e., sectors with different angular momenta. On the other hand, for $d\geq 3$, in addition to similar constraints on the relative phases between the irreps, locality also restricts the generated unitaries \emph{inside} these conserved subspaces.  These constraints impose conservation laws that hold for dynamics under 2-local {SU}$(d)$-invariant unitaries, but are violated under general  {SU}$(d)$-invariant unitaries. Based on this result, we show  that the distribution of unitaries generated by random 2-local SU$(d)$-invariant unitaries does not converge to the Haar measure over the group of all SU$(d)$-invariant unitaries, and in fact, for $d\ge 3$, is not even a 2-design for the Haar distribution.

\end{abstract}

\maketitle

\section{Introduction}

How do symmetries of a composite system with local interactions restrict the long-term evolution of the system? The standard conservation laws implied by Noether's theorem \cite{noether1918nachrichten, noether1971invariant}, such as conservation of the total angular momentum vector for systems with rotational symmetry, hold regardless of whether the Hamiltonian is local or not.
Suppose the total Hamiltonian of a composite system can be decomposed as a sum of $k$-local terms (also known as $k$-body terms), i.e., terms that only act non-trivially on at most $k$ subsystems for a fixed $k$. For such systems, do the symmetries of the Hamiltonian put any further constraints on the long-term evolution of the system? Perhaps surprisingly, it turns out that the answer is yes. One of us has recently shown that time evolutions under such local Hamiltonians can only generate a restricted family of symmetric unitary transformations \cite{marvian2022restrictions}. In particular, in the case of continuous symmetries such as $\U(1)$ and $\SU(2)$, the difference between the dimensions of the manifold of all symmetric unitaries and the submanifold of unitaries generated by symmetric local Hamiltonians constantly grows with the number of subsystems. This holds even for time-dependent symmetric Hamiltonians with arbitrarily long-range interactions, provided that each term in the Hamiltonian acts only a (fixed) finite number of subsystems.

In the language of quantum circuits, the result of \cite{marvian2022restrictions} means that general symmetric unitary transformations cannot be generated using local symmetric gates, even approximately. This is in sharp contrast with the well-known universality of 2-local unitary transformations, which holds in the absence of symmetries \cite{lloyd1995almost, divincenzo1995two, deutsch1995universality, brylinski2002universal}. The recent result of \cite{marvian2022restrictions} indicates that some basic properties of local symmetric quantum circuits have not been yet understood. Given the wide range of applications of this family of quantum circuits, from the classification of symmetry-protected topological phases \cite{chen2011classification} to quantum thermodynamics \cite{FundLimitsNature, brandao2013resource, janzing2000thermodynamic, guryanova2016thermodynamics, lostaglio2015quantumPRX, lostaglio2017thermodynamic, lostaglio2015description, halpern2016microcanonical, halpern2016beyond, narasimhachar2015low}, quantum reference frames \cite{QRF_BRS_07, gour2008resource, marvian2013theory, marvian2008building}, and quantum control, characterizing and understanding their properties is of crucial importance. Specifically, understanding the additional constraints imposed by the presence of both locality and symmetry can be useful in the context of quantum chaos \cite{maldacena2016bound, piroli2020random} and quantum complexity \cite{susskind2016computational} in systems with conserved charges \cite{khemani2018operator, rakovszky2018diffusive}.

In this paper we study local quantum circuits with $\SU(d)$ symmetry for systems of qudits, i.e., subsystems with Hilbert spaces of dimension $d\ge 3$. These symmetry groups appear in broad areas of physics, including nuclear physics, condensed matter physics and quantum gravity. Specifically, we consider quantum circuits formed from a sequence of 2-local $\SU(d)$-invariant unitaries, i.e., those that only couple pairs of qudits. 

Our results reveal an unexpected distinction between the case of $d=2$, i.e., systems of qubits with $\SU(2)$ symmetry, and $d>2$. It turns out that in the case of $d=2$, locality of interactions only restricts the relative phases between the subspaces corresponding to inequivalent irreducible representations (irreps) of the symmetry, i.e., sectors with different angular momenta \cite{MLH_2022, 
LHM_2022}. On the other hand, in the case of $d\ge 3$, in addition to similar constraints on the relative phases between sectors with inequivalent irreps of $\SU(d)$, locality also restricts the unitaries generated inside these sectors. More precisely, we find that even for states restricted to a sector which carries a single irrep (i.e., restricted to the isotypic subspace associated to that irrep), there are additional conservation laws that hold under transformations that can be realized using a sequence of 2-local $\SU(d)$-invariant unitaries, whereas they are violated under general $\SU(d)$-invariant unitaries (equivalently, these are conservation laws that hold for time evolutions generated by $\SU(d)$-invariant Hamiltonians that can be written as a sum of 2-local terms, whereas they are violated for general $\SU(d)$-invariant Hamiltonians). In \cref{Fig}, we discuss an interesting example of this phenomenon, which is further studied in \cref{sect:examp}. Furthermore, in addition to constraints on unitaries inside conserved sectors, it turns out that unitaries in different sectors also satisfy certain constraints among themselves.

\begin{figure}[t]
  \includegraphics[scale=.7]{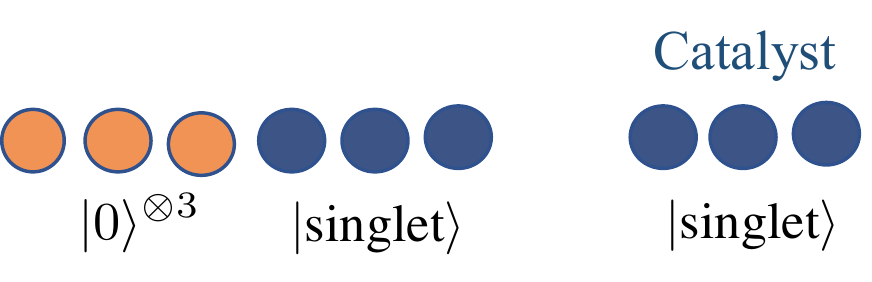}
  \caption{\emph{An illustrative example of restrictions imposed by the locality of interactions}: Consider 6 qutrits, i.e., systems with 3D Hilbert spaces, in the joint state $|\Psi\rangle=|\mathrm{singlet}\rangle|0\rangle^{\otimes 3}$, where $|\mathrm{singlet}\rangle=|0\rangle\wedge |1\rangle\wedge |2\rangle$ is the unique totally anti-symmetric state of 3 qutrits, and $\{|0\rangle, |1\rangle, |2\rangle\}$ is an orthonormal basis for a qutrit (see \cref{eq:wedge} for the definition of the wedge product). The state $|\Psi\rangle$ is restricted to a subspace corresponding to an irrep of $\SU(3)$. Furthermore, for any $\alpha, \beta\in\mathbb{C}$ the superposition $\alpha |\mathrm{singlet}\rangle|0\rangle^{\otimes 3}+\beta |0\rangle^{\otimes 3} |\mathrm{singlet}\rangle$ can be obtained from the initial state $|\Psi\rangle$ via an $\SU(3)$-invariant unitary. However, unless $\alpha=0$ or $\beta=0$, such superpositions cannot be obtained from the initial state $|\Psi\rangle$ via any sequence of 2-local $\SU(3)$-invariant unitaries, i.e., those that couple only pairs of qutrits. More precisely, such state transitions are forbidden by certain conservation laws that are satisfied if the Hamiltonian is $\SU(3)$-invariant and can be written as a sum of 2-local terms, whereas they are violated by general $\SU(3)$-invariant Hamiltonians (see \cref{eq:renyi}). Interestingly, this forbidden transition becomes possible with 2-local $\SU(3)$-invariant unitaries provided that the qutrits can interact with a catalyst, namely 3 additional qutrits in the singlet state (see \cref{sect:examp} for further discussion about this example).}
  \label{Fig}
\end{figure}
The additional constraints that appear for $d\ge 3$ can be broadly classified into two different types: the first type, discussed in \cref{Sec:Z2}, is based on a $\mathbb{Z}_2$ symmetry related to the notion of the parity of permutations and a corresponding conservation law that constrains the time evolution of systems with $n\le d^2$ qudits (see \cref{cons-fskew}). The second type of constraints, discussed in \cref{Sec:ferm}, is based on a new correspondence between the dynamics of the qudit system and the dynamics of a free fermionic system. This correspondence reveals additional conservation laws, which remain non-trivial for systems with arbitrarily large numbers of qudits (see \cref{eq:renyi}). 

In \cref{Sec:Random} we discuss implications of our results in the context of random symmetric circuits. In particular, we argue that for $d\ge 3$, even the second moment of the distribution of unitaries generated by 2-local unitaries does not converge to the uniform (Haar) measure over the group of all $\SU(d)$-invariant unitaries (in other words, 2-local symmetric unitaries do not form a 2-design for all symmetric unitaries---see \cref{design}). Further discussions about the full characterization of the unitaries generated by local $\SU(d)$-invariant unitaries and the statistical properties of random $\SU(d)$-invariant quantum circuits can be found in a follow-up paper \cite{HLM_2021}. \\

\paragraph*{Previous Related Works}
It turns out that the group of unitaries generated by 2-local SU($d$)-invariant unitaries is a compact connected Lie group with the Lie algebra generated by swap operators (See the next section).  A closely related Lie algebra has been previously studied in the mathematical literature. Namely, in \cite{marin2007algebre} Marin obtains the decomposition into simple factors of the Lie algebra generated by transpositions, thought of as a Lie subalgebra of the group algebra of the symmetric group.  As we explain throughout the paper, some results which are established using rather elementary techniques in this paper, can also be obtained using more advanced results of \cite{marin2007algebre} (specifically, see the discussions below \cref{lem9,Thm5}). In \cite{HLM_2021} we discuss more about the connection between the results presented here and the result of \cite{marin2007algebre}.  

Another set of related works are results regarding universal quantum computation with the Heisenberg exchange interaction (see e.g. \cite{DiVincenzo:2000kx, levy2002universal, kempe2001theory, bacon2001coherence, rudolph2005relational}). These works study the implementation of an \emph{encoded} version of a desired unitary using the exchange  interaction, which is rotationally-invariant and 2-local. In contrast, we are interested in characterizing \emph{all}  realizable unitaries. We also note that recently an application of Marin's work in this context has been studied in  \cite{van2021universality}.

\section{Preliminaries}

Consider a composite system formed from $n$ qudits, i.e., systems with $d$-dimensional local Hilbert spaces, with the total Hilbert space $(\mathbb{C}^d)^{\otimes n}$. For concreteness one can assume the qudits lie on a 1D open chain and are labeled as $i=1, \cdots, n$, corresponding to their positions in this chain. However, our results do not depend on this particular geometry. We say a unitary transformation $V$ on this system is $\SU(d)$-invariant if it commutes with $U^{\otimes n}$ for all $U\in \SU(d)$, i.e., 
\be\label{def}
\forall U\in \SU(d):\ \ [V, U^{\otimes n}]=0\ .
\ee
 To characterize the set of $\SU(d)$-invariant unitaries, it is convenient to consider the decomposition of the tensor product representation of $\SU(d)$ into irreps. Then, the total Hilbert space decomposes as
\be\label{decomp:SU}
(\mathbb{C}^d)^{\otimes n} \cong \bigoplus_\lambda \mathcal{H}_\lambda =\bigoplus_\lambda \hilbert[Q]_\lambda \otimes \hilbert[M]_\lambda\ ,
\ee
where $\lambda$ labels different \emph{charges}, i.e., inequivalent irreps of $\SU(d)$, $\mathcal{H}_\lambda$ is the sector corresponding to irrep $\lambda$ (also known as the isotypic component of $\lambda$), $\hilbert[Q]_\lambda$ is the irrep of $\SU(d)$ corresponding to $\lambda$, and $\hilbert[M]_\lambda$ is the (virtual) multiplicity subsystem \cite{QRF_BRS_07, zanardi2001virtual}, where $\SU(d)$ acts trivially. 
Using this decomposition together with Schur's Lemma one can easily classify all $\SU(d)$-invariant unitaries \cite{QRF_BRS_07}: they are all block-diagonal with respect to sectors with different charges. Furthermore, they act trivially on $\hilbert[Q]_\lambda$, and non-trivially on $\hilbert[M]_\lambda$ (see \cref{app:duality} for further details). In summary, this means $V$ is $\SU(d)$-invariant if and only if it can be written as $V\cong\bigoplus_\lambda \mathbb{I}_\lambda \otimes v_\lambda$, where $ \mathbb{I}_\lambda$ is the identity operator on $\hilbert[Q]_\lambda$ and $v_\lambda$ is an arbitrary unitary on the multiplicity subsystem $\hilbert[M]_\lambda$.

In the presence of symmetry, certain properties of state remain conserved under the time evolution of the system. In particular, if under an $\SU(d)$-invariant unitary $V$ the initial state $|\psi\rangle$ evolves to $|\psi'\rangle=V|\psi\rangle$, then from the definition in \cref{def}, we can easily see 
\be\label{consv}
\forall U\in \SU(d):\ \langle\psi|U^{\otimes n}|\psi\rangle=\langle\psi'|U^{\otimes n}|\psi'\rangle\ .
\ee 
We refer to these equations as Noether's conservation laws. Note that they can be equivalently understood as the conservation of the moments of the generators of symmetry. In the case of $\SU(2)$, for instance, they are equivalent to
\be\label{consv2}
 \langle\psi|J^l_{\hat{r}}|\psi\rangle=\langle\psi'|J^l_{\hat{r}}|\psi'\rangle\ ,
\ee
for all integers $l$ and all unit vectors $\hat r\in \mathbb{R}^3$, where $J_{\hat{r}}=r_x J_x+r_y J_y+r_z J_z$ is angular momentum in direction $\hat{r}$. 

It turn out that, in the case of pure states\footnote{Interestingly, this does not hold in the case of mixed states. For instance, for mixed states the conservation of the moments of angular momentum in different directions does not capture all the consequences of symmetry \cite{marvian2014extending}.}, the constraints imposed by the $\SU(d)$ symmetry are fully captured by \cref{consv}, or equivalently, by the conservation of angular momenta and their moments in the case of $\SU(2)$: this equation implies that there exists an $\SU(d)$-invariant unitary $V$, such that $V|\psi\rangle=|\psi'\rangle$\cite{marvian2014asymmetry}. In this work we ask, for the family of unitaries that can be implemented using \emph{local} symmetric Hamiltonians, or, equivalently, using local symmetric circuits, are there any additional conservation laws which are not captured by the standard conservation laws in \cref{consv}? 

\subsection*{Local $\SU(d)$-invariant quantum circuits}
We first formulate this question in the framework of local symmetric circuits. Consider the family of quantum circuits formed from $k$-local $\SU(d)$-invariant unitaries, i.e., unitaries that act non-trivially on \emph{at most} $k$ qudits. Let $\mathcal{V}_{k}$ be the group generated by composing finitely many $k$-local $\SU(d)$-invariant unitaries, i.e., any $V \in \mathcal{V}_{k}$ can be written as $V=\prod_{i=1}^{m} V_{i}$ for some finite $m$, where each $V_i$ is a $k$-local $\SU(d)$-invariant unitary. In particular, $\mathcal{V}_{n}$ is the group of all $\SU(d)$-invariant unitaries. As we saw before, the decomposition in \cref{decomp:SU} provides a simple characterization of this group. Schur-Weyl duality \cite{harrow2005applications, goodman2009symmetry} implies that $\mathcal{V}_{n}$ can also be characterized as the set of unitaries in the linear span of permutations of $n$ qudits (see \cref{app:duality}).

A special case of interest is that of circuits formed from 2-local unitaries. It turns out that, up to a global phase, any 2-local 
$\SU(d)$-invariant unitary 
  acting on qudits $a$ and $b$ can be written as
\be\label{swap}
\e^{\i \theta \P_{ab}}=\cos \theta\ \mathbb{I}+\i \sin\theta\ \P_{ab}\ , 
\ee
where $\P_{ab}$ is the swap operator, also known as the transposition, i.e., the operator that exchanges the state of $a$ and $b$, and leaves the rest of the qudits unchanged, such that $\P_{ab} |\psi\>_a|\phi\>_b=|\phi\>_a|\psi\>_b$, $\forall |\phi\>, |\psi\>\in\mathbb{C}^d$ (we will often drop the tensor product with the identity on the rest of qudits). Therefore, up to global phases, $\mathcal{V}_{2}$ is the group of unitaries generated by unitaries $\{\e^{\i \theta \P_{ab}} : \theta\in[0,2\pi), \ 1\le a<b\le n\}$. In other words,
\be
\mathcal{V}_{2}=\langle\e^{\i\theta} \mathbb{I}\ ,\ \e^{\i \theta \P_{ab}} : \theta\in[0,2\pi), \ 1\le a<b\le n \rangle \ .
\ee

Clearly, any unitary generated by $2$-local $\SU(d)$-invariant  unitaries is itself $\SU(d)$-invariant. That is $\mathcal{V}_{2}\subset \mathcal{V}_{n}$. The question is whether there are $\SU(d)$-invariant  unitaries that cannot be implemented using $2$-local $\SU(d)$-invariant  unitaries, and if so, how they can be characterized.

This question can be equivalently stated in terms of Hamiltonians. Consider an $\SU(d)$-invariant qudit Hamiltonian $H(t)$ that can be written as a sum of 2-local terms. Then, up to a constant shift, it can be written as
\be\label{Hamiltonian}
H(t)=\sum_{i<j} h_{ij}(t)\ \P_{ij}\quad :\ t\ge 0\ ,
\ee
where in general, $h_{ij}$ is an arbitrary real function of time $t$. Note that this Hamiltonian can include long-range interactions between arbitrarily far qudits. Also, note that in the case of qubits, up to a 
constant shift, the Hamiltonian $H(t)$ can be written as $\frac{1}{2}\sum_{i<j} h_{ij}(t)\ \vec{\sigma}_i\cdot \vec{\sigma}_j$, describing a qubit chain with the Heisenberg exchange interactions. 

Consider the class of unitaries generated by these Hamiltonians under the Schr\"{o}dinger equation 
\be
\frac{\d V(t)}{\d t}=-\i H(t) V(t) \ , 
\ee
where $V(0)=\mathbb{I}$ is the identity operator. It turns out that, up to a global phase, any such unitary can be implemented \emph{exactly} by a finite sequence of 2-local unitary transformations in the form of \cref{swap}, and hence is an element of $\mathcal{V}_{2}$. Conversely, up to a global phase, any unitary transformation in $\mathcal{V}_{2}$, can be realized by a time evolution generated under a Hamiltonian in the form of \cref{Hamiltonian} \cite{d2002uniform, marvian2022restrictions}. Therefore, by studying the family of quantum circuits that can be generated by 2-local unitaries and the corresponding group $\mathcal{V}_{2}$, we also characterize general constraints on the time evolutions generated by any $\SU(d)$-invariant Hamiltonian that can be written as a sum of 2-local terms. More generally, adding any (possibly non-local) interaction in the Lie algebra generated by $\SU(d)$-invariant 2-local terms to the Hamiltonian $H(t)$ in \cref{Hamiltonian} does not enlarge the set of realizable unitaries; the generated unitary will be inside $\mathcal{V}_{2}$, and therefore satisfies the constraints studied in this paper.

\subsection*{Systems of qubits with \texorpdfstring{$\SU(2)$}{SU(d)} symmetry}
In \cite{MLH_2022, LHM_2022} we study qubit circuits with $k$-local SU(2)-invariant gates and fully characterize the constraints on realizable unitaries. It turns out that in this case, the locality of interactions only restricts the relative phases between sectors with inequivalent irreps of SU(2), i.e., subspaces with different angular momenta (assuming $k\ge 2$). In particular, using a rather elementary argument we show that, up to these relatives phases, any SU(2)-invariant unitary can be implemented using the Heisenberg exchange interaction,  that is using Hamiltonians in the form of Eq.(\ref{Hamiltonian}) (We note that this fact can also be established by applying more advanced results of Marin in \cite{marin2007algebre}). In addition to fully characterizing the constraints on the relative phases, in  \cite{MLH_2022} we also present a physical interpretation of these constraints,  a method for experimentally observing violations of these constraints, and finally a scheme for circumventing these constraints and implementing general rotationally-invariant unitaries, using two ancilla qubits.

\subsection*{Systems of qudits with \texorpdfstring{$\SU(d)$}{SU(d)} symmetry for \texorpdfstring{$d\ge3$}{d>=3}}

Perhaps surprisingly, it turns out that
the above fact about qubits with the SU(2) symmetry, does not hold for systems of qudits with SU($d$) symmetry with $d\ge 3$: As we show in the rest of this paper, for such systems, in addition to the constraints on the relative phases between different charge sectors $\sets{\hilbert_\lambda}$, locality also restricts the realizable unitaries \emph{inside} certain charge sectors. Furthermore, there are some charge sectors in which the generated unitaries cannot be independent of each other; in fact, the unitary in one of these sectors completely determines the generated unitaries in the others. It follows that there are additional conservation laws that hold under 2-local $\SU(d)$-invariant unitaries whereas they are violated under general 3-local $\SU(d)$-invariant unitaries. These additional constraints can be classified into two different types, as discussed in the following sections. We also discuss the constraints on the relative phases between different charge sectors and find a simple characterization in terms of the quadratic Casimir operators in \cref{app:relative-phases}. 

\color{black}

\section{A \texorpdfstring{$\mathbb{Z}_2$}{Z2} symmetry from parity of permutations}\label{Sec:Z2}

In this section we discuss a constraint on the family of unitaries generated by 2-local $\SU(d)$-invariant Hamiltonians that is related to the notion of parity of permutations. For a system with $n$ qudits, the set of swaps on pairs of qudits generates a group of unitaries that is a representation of $\mathbb{S}_n$, the permutation group of $n$ objects, also known as the symmetric group. The permutation $\sigma\in\mathbb{S}_n $ is represented by the operator $\P(\sigma)$ that satisfies 
\be\label{eq:permrep}
\forall \sigma\in\mathbb{S}_n:\ \ \P(\sigma)\bigotimes_{i=1}^n |\psi_i\rangle=\bigotimes_{i=1}^n |\psi_{\sigma{(i)}}\rangle\ ,
\ee
for any set of states $\{|\psi_{i}\rangle\in\mathbb{C}^d\}$. In particular, the transposition $\sigma_{ab}\in \mathbb{S}_n$ is represented by a swap operator ${\P}(\sigma_{ab})\equiv {\P}_{ab}$. Recall that the parity of a permutation $\sigma\in\mathbb{S}_n$ is even/odd if the number of transpositions that are needed to generate $\sigma$ is even/odd. We denote these two cases by $\sgn(\sigma)=1$ and $\sgn(\sigma)=-1$, respectively. 
 
Consider the following simple example with 3 qudits: the unitary $\e^{\i \phi \P_{23}} \e^{\i \theta \P_{12}}$ for $\theta,\phi\in[0,2\pi)$ can be expanded as the term $\cos \theta \cos\phi\ \mathbb{I}- \sin \theta \sin\phi \ \P_{23}\P_{12}$ plus the term $\i\cos\theta \sin\phi \ \P_{23}+\i\sin\theta \cos\phi \ \P_{12}$. In this decomposition, the first term only contains even parity permutations with real coefficients, whereas the second term only contains odd parity permutations with purely imaginary coefficients. As one combines more unitaries in the form $ \e^{\i\theta \P_{ab}}$, this correlation between parity and reality always holds: even permutations appear with real coefficients, whereas odd permutations appear with purely imaginary coefficients. Since permutations themselves can be represented by real matrices (e.g., in the computational basis), this means that in this basis the real part of the generated unitary matrix is determined by the even permutations whereas the imaginary part is determined by the odd permutations. 

Does this restrict the family of unitaries that can be implemented with 2-local $\SU(d)$-invariant unitaries? The answer depends on $d$, the dimension of the local subsystems. Unless $d$ is sufficiently large, the set of permutations are not linearly independent, and in particular, the linear combinations of odd permutations contain even permutations and vice versa. For instance, for $d=2$, we have $\mathbb{I}+\P_{12}\P_{23}+\P_{23}\P_{12}=\P_{12}+\P_{23}+\P_{13}$, whereas this equality does not hold for $d\ge 3$. Therefore, for $d=2$ the left-hand side, which is a sum of even terms with real coefficients, can also be obtained as a linear combination of odd terms with real coefficients. In fact, it turns out that, for any number of qubits with $\SU(2)$ symmetry, the above constraint does not impose any additional restrictions on the time evolution of the system. In the following, we formulate this constraint in a more general setting and show that, for $d\ge 3$, it does impose non-trivial constraints on the generated unitaries. 

\subsection{Formulating the \texorpdfstring{$\mathbb{Z}_2$}{Z2} symmetry}

According to Schur-Weyl duality (see \cref{app:duality}), any general $\SU(d)$-invariant Hamiltonian $H(t)$ on qudits can be written as a linear combination of permutations, i.e., 
\be\label{try}
H(t)=\sum_{\sigma\in\mathbb{S}_n} h_{\sigma}(t)\ \P(\sigma)\ .
\ee
In the following discussion it is convenient to assume that in this decomposition the coefficient of the identity is zero, i.e., $h_e(t)=0$. This can always be achieved by changing the energy reference, which is equivalent to adding a proper multiple of the identity operator to the Hamiltonian.

We are interested in $\SU(d)$-invariant Hamiltonians that have a decomposition in this form with the additional property that for all permutations $\sigma\in \mathbb{S}_n$, $ h_{\sigma}(t)$ is real (purely imaginary) for odd (even) parity $\sigma$. Equivalently, this condition can be stated as 
\be\label{assump}
h_{\sigma}(t)=-\sgn(\sigma) h^\ast_{\sigma}(t)\ \ \ : \forall t\ge 0\ , \sigma\in\mathbb{S}_n\ ,
\ee
where $h^\ast_{\sigma}(t)$ is the complex conjugate of $h_{\sigma}(t)$. The transformation $ h_{\sigma}(t)\mapsto -\sgn(\sigma) h^\ast_{\sigma}(t)\ $ can be thought of as a transformation on the vector $h_{\sigma}(t): \sigma\in \mathbb{S}_n$. Applying this transformation twice is equivalent to the identity map, i.e., it generates a $\mathbb{Z}_2$ symmetry. 

If the $\SU(d)$-invariant Hamiltonian $H(t)$ has a decomposition satisfying \cref{assump}, then the unitary time evolution generated by $H(t)$ under the Schr\"{o}dinger equation $\frac{\d V(t)}{\d t}=-\i H(t){V}(t) $, with the initial condition $V(0)=\mathbb{I}$, has a decomposition as
\be
V(t)=\sum_{\sigma\in \mathbb{S}_n} v_\sigma(t)\ \P(\sigma)\ , 
\ee
where the coefficient $v_\sigma(t)$ is real (purely imaginary) for even (odd) parity $\sigma$, i.e.,
\be\label{thn}
v_\sigma(t)=\sgn(\sigma)v^\ast_\sigma(t)\ \ \ : \forall t\ge 0\ , \sigma\in\mathbb{S}_n\ .
\ee
Note that this equation expresses the pattern we observed in the simple example discussed at the beginning of this section. To see why this equation holds note that if $V_1$ and $V_2$ are two general $\SU(d)$-invariant operators with decompositions satisfying \cref{thn}, then the product $V_2V_1$ also satisfies this property. Combining this with the fact that $V(t)=\lim_{L\rightarrow \infty} \prod_{l=0}^L [\mathbb{I}-\frac{\i t}{L} H(\frac{l t}{L})]$, one can show that \cref{assump} implies \cref{thn}. Using a similar argument it can be seen that unitary transformations satisfying this property form a Lie group with a Lie algebra characterized by \cref{assump}. 

We note that a closely related Lie algebra has been previously studied in the mathematical literature. Namely, in \cite{marin2010group} Marin considers operators $\sum_{\sigma\in\mathbb{S}_n} h_\sigma \mathbf{R}(\sigma)$, where $\mathbf{R}$ is the regular representation, and $h_\sigma$ satisfies a similar (but related) condition to \cref{assump}.\footnote{In fact, in \cite{marin2010group} this is studied in the case of general finite groups, and the problem is phrased in the language of the group algebra associated to the finite group (rather than that of the regular representation). In \cref{app:genZ2}, we also consider the case of a general representation of an arbitrary finite or compact Lie group and generalize the results of this section. } Since the permutations $\cset{\mathbf{R}(\sigma)}{\sigma \in \mathbb{S}_n}$ in the regular representation are linearly independent, in this special case, there is a one-to-one correspondence between operators in the span of the regular representation and functions over the group (in fact, in this case, \cref{assump} can be understood as an anti-unitary symmetry on the space of operators).

In our case, however, the problem is more subtle: as we mentioned before, the unitaries $\cset{\P(\sigma)}{\sigma \in \mathbb{S}_n}$ are, in general, \textit{not} linearly independent. Hence, an $\SU(d)$-invariant Hamiltonian $H(t)$ that has a decomposition satisfying \cref{assump}, may also have other decompositions violating it. In the following lemma, we provide a simple criterion that determines whether decompositions satisfying \cref{assump} exist or not.

Consider the operator 
\be\label{Def:K}
K\equiv\frac{1}{n!} \sum_{\sigma\in\mathbb{S}_n} \sgn(\sigma)\ \P(\sigma)\otimes \P(\sigma)\
\ee
acting on $n$ pairs of qudits. It can be easily shown that $K$ is the Hermitian projector to the sign representation of the permutation group on these pairs (see \cref{App:Kf}), i.e., $K=K^\dag=K^2$ and
\be\label{Def:K1}
K[\P(\sigma)\otimes \P(\sigma)]=\sgn(\sigma) K\ .
\ee
With this operator $K$, we have the following lemma,
\begin{lemma}\label{lem13}
  An $\SU(d)$-invariant Hermitian operator $H$
  has a decomposition as $H=\sum_{\sigma\in\mathbb{S}_n} h_{\sigma}\ \P(\sigma)$, satisfying $h_\sigma=-\sgn(\sigma) h^\ast_\sigma$, if and only if 
  \be\label{bbc}
  K[H\otimes \mathbb{I}+\mathbb{I}\otimes H]=[H\otimes \mathbb{I}+\mathbb{I}\otimes H]K=0\ .
  \ee 
\end{lemma}
In \cref{app:genZ2}, we discuss a generalization of the notion of $\mathbb{Z}_2$ symmetry and present the proof of a generalized version of this lemma (see \cref{thm:Z2}). One direction of the lemma is proven below. 

\begin{proof}
  Here, we prove that if $h_\sigma=-\sgn(\sigma) h^\ast_\sigma~ \forall \sigma\in\mathbb{S}_n$, then $K[H\otimes \mathbb{I}+\mathbb{I}\otimes H]=0$. Consider
\begin{align}
  H^\dag&=\sum_{\sigma\in\mathbb{S}_n} h^\ast_{\sigma}\ \P^\dag(\sigma)=-\sum_{\sigma\in\mathbb{S}_n} \sgn(\sigma) h_{\sigma}\ \P(\sigma^{-1})\nonumber \\ &=-\sum_{\sigma\in\mathbb{S}_n} \sgn(\sigma) h_{\sigma^{-1}}\ \P(\sigma)\ ,
\end{align}
where to get the second equality we have used the assumption of the lemma and to get the third equality, we have used $\sgn(\sigma)=\sgn(\sigma^{-1})$. Since $H$ is Hermitian, using $H=(H+H^\dag)/2$, we conclude $H=\sum_{\sigma\in\mathbb{S}_n} \tilde{h}_{\sigma}\ \P(\sigma)$, where 
\begin{align}
  \tilde{h}_\sigma\equiv \frac{1}{2}\big[h_{\sigma}- \sgn(\sigma) h_{\sigma^{-1}}\big]\ ,
\end{align}
therefore satisfying 
\be\label{fdr}
\tilde{h}_{\sigma}= -\sgn(\sigma)\tilde{h}_{\sigma^{-1}}\ \ : \forall\sigma\in\mathbb{S}_n \ .
\ee
Then, 
\begin{align}\label{art11}
  &\quad~ K[H\otimes \mathbb{I}+\mathbb{I}\otimes H]\nonumber\\ &=K \sum_\sigma [\tilde{h}_{\sigma}\ \P(\sigma)\otimes \mathbb{I}+\mathbb{I}\otimes \tilde{h}_{\sigma^{-1}}\ \P(\sigma^{-1})]\nonumber\\ &=K \sum_\sigma [\tilde{h}_{\sigma} \P(\sigma)\otimes \mathbb{I}+\sgn(\sigma) \tilde{h}_{\sigma^{-1}}\ \P(\sigma)\otimes \mathbb{I}]=0\ , 
\end{align}
where the second and third equalities follow from \cref{Def:K1} and \cref{fdr}, respectively. Finally, by taking the adjoint of both sides of this equation, and using the fact that $H$ and $K$ are Hermitian, we find that $K$ commutes with $[H\otimes \mathbb{I}+\mathbb{I}\otimes H]$. This completes the proof of one side of the lemma (see \cref{app:genZ2} for the proof of the other side and a generalization of this lemma).
\end{proof}

It is worth noting that the statement of the lemma can be slightly generalized by allowing shifts by a multiple of the identity operators: there exists a real number $\alpha$, such that Hamiltonian $H-\alpha \mathbb{I}$ has a decomposition satisfying the above properties, if and only if $K(H\otimes \mathbb{I}+ \mathbb{I}\otimes H)=2\alpha K$. 

It turns out that for systems with $n>d^2$ qudits the operator $K$ is zero. This is because $K$ can be thought of as the projector to the totally anti-symmetric subspace of $n$ pairs of qudits, i.e., $n$ systems with $d^2$-dimensional Hilbert space (see \cref{Def:K1}). For $n>d^2$, such systems cannot have a totally anti-symmetric subspace, which implies $K=0$ (see \cref{App:Kf} for further details). Therefore, the above lemma implies that for systems with $n>d^2$, any $\SU(d)$-invariant Hermitian operator $H$ has a decomposition satisfying \cref{assump}. 

In particular, this means that $K=0$ for $n>4$ qubits, and therefore, the condition in \cref{lem13} is always satisfied for such systems. Interestingly, it turns out that, even though for $n=3, 4$ qubits $K\neq 0$, after a proper shift by a multiple of the identity operator, any  $\SU(2)$-invariant  Hermitian operator $H$ satisfies the condition in \cref{bbc}, and hence by the lemma it has a decomposition satisfying the $\mathbb{Z}_2$ symmetry. For instance, we have already seen that for $n=3$ qubits the operator $\P_{12} \P_{23} + \P_{23} \P_{12}$, which apparently does not satisfy the $\mathbb{Z}_2$ symmetry, up to a shift by the identity operator can be written as $\P_{12}+\P_{23}+\P_{13}$, which respects the condition in \cref{bbc}. In \cref{Sec:examp:4}, we prove this can always be achieved for $n=3, 4$ qubits. In summary,

\begin{corollary}\label{tyt}
  Any $\SU(d)$-invariant Hermitian operator $H$ on $n>d^2$ qudits has a decomposition as $H=\sum_\sigma h_\sigma \P(\sigma) $ satisfying the condition $h_\sigma=-\sgn(\sigma) h^{\ast}_\sigma$ for all $\sigma\in\mathbb{S}_n$. Furthermore, in the case of $n=3, 4$ qubits, any $\SU(2)$-invariant Hermitian operator $H$, up to shift by a multiple of the identity operator, has a decomposition satisfying the above property. 
\end{corollary}

\subsection{A conservation law imposed by the \texorpdfstring{$\mathbb{Z}_2$}{Z2} symmetry}

It turns out that, under time evolution generated by Hamiltonians satisfying this property, certain functions of state remain conserved, whereas they can vary under general $\SU(d)$-invariant Hamiltonians: consider an $\SU(d)$-invariant Hamiltonian $H(t)=\sum_{\sigma\in\mathbb{S}_n} h_{\sigma}(t)\ \P(\sigma)$ satisfying the constraint in \cref{assump} and let $V(t): t\ge 0$ be the family of unitaries satisfying the Schr\"{o}dinger equation $\frac{\d V(t)}{\d t}=-\i H(t){V}(t) $, with the initial condition $V(0)=\mathbb{I}$. Then, using \cref{lem13} we find 
\begin{align}
  &\quad~ \frac{\d}{\d t}K[{V}(t)\otimes{V}(t)]\nonumber\\ &=-\i K[H(t)\otimes \mathbb{I}+\mathbb{I}\otimes H(t)] [{V}(t)\otimes{V}(t)]= 0\ ,
\end{align}
which, in turn, implies 
\be
K[{V}(t)\otimes {V}(t)]=K\ .
\ee 
This means that if under an $\SU(d)$-invariant Hamiltonian that satisfies \cref{assump}, a pair of initial states $|\psi_{1,2}\>$ evolves to the final states $|\psi'_{1,2}\>$, then
\be\label{form}
K\big(|\psi'_1\>\otimes |\psi'_2\>\big)= K\big(|\psi_1\>\otimes |\psi_2\>\big) \ .
\ee

Motivated by this observation, for any system with $n\ge 2$ qudits, define $f_{\sgn}[\psi]$ as the square of the norm of the vector $K(|\psi\>\otimes |\psi\>)$, or equivalently as
\begin{equation}\label{def:fskew}
  \begin{split}
    f_{\sgn}[\psi]&\equiv (\langle\psi|\otimes \langle\psi|) K (|\psi\>\otimes |\psi\>) \\ &=\frac{1}{n!} \sum_{\sigma\in\mathbb{S}_n} \sgn(\sigma)\ \<\psi| \P(\sigma)|\psi\>^2\ ,
  \end{split}
\end{equation}
where $\psi=|\psi\rangle\langle\psi|$ denotes the density operator associated with the state $|\psi\rangle$. The fact that $K$ is a projector implies that $f_{\sgn}[\psi]$ is a real number between 0 and 1. Furthermore, \cref{form} implies 
\begin{theorem}\label{Thm5}
  Suppose the $\SU(d)$-invariant Hamiltonian $H(t)$ has a decomposition as $H(t)=\sum_{\sigma\in\mathbb{S}_n} h_{\sigma}(t)\ \P(\sigma)$ that satisfies  $h_{\sigma}(t)=-\sgn(\sigma) h^\ast_{\sigma}(t)$ for all $t\ge 0$ and all $\sigma\in\mathbb{S}_n$, except possibly the identity element $e\in\mathbb{S}_n$ . If under the time evolution generated by this Hamiltonian an initial state $|\psi\rangle$ evolves to $|\psi'\rangle$, then 
  \be\label{cons-fskew}
  f_{\sgn}[\psi']=f_{\sgn}[\psi]\ .
  \ee
\end{theorem}
In particular, this conservation law holds for any $\SU(d)$-invariant Hamiltonian that can be written as a sum of 2-local terms, namely Hamiltonians in \cref{Hamiltonian}. Equivalently, this means that if the initial state $|\psi\rangle$ can be converted to $|\psi'\rangle$ by a sequence of 2-local $\SU(d)$-invariant unitaries, then the above conservation law holds.

Furthermore, as we saw in \cref{tyt}, up to a shift by a multiple of the identity operator, any $\SU(2)$-invariant Hamiltonian on qubits has a decomposition satisfying \cref{assump}. Therefore, in the case of $d=2$, $f_{\sgn}$ is trivially conserved. On the other hand, as we see in the following example for $d\ge 3$, this conservation law can be violated by Hamiltonians that contain 3-local $\SU(d)$-invariant terms. 

Finally, we note that the conservation of $f_{\sgn}$ does not capture all the consequences of the presence of the $\mathbb{Z}_2$ symmetry.\footnote{For instance, under the same condition in \cref{Thm5}, the expectation value of the operator $K(\Pi_\lambda\otimes \mathbb{I})$ for the state $|\psi\rangle|\psi\rangle$ also remains conserved, for any irrep $\lambda$ of $\SU(d)$.} We discuss more about these additional conservation laws in a follow-up work \cite{HLM_2021}, where we also explain how they are related to a certain family of bilinear forms that are used by Marin \cite{marin2007algebre} for characterizing the Lie algebra generated by swaps.

\subsection{Example: A 6 qutrit system with \texorpdfstring{$\SU(3)$}{SU(3)} symmetry}\label{sec:exZ2}
In \cref{fig:4-local-fskew} we consider a system of 6 qutrits, with the total Hilbert space $(\mathbb{C}^3)^{\otimes 6}$ in the initial state $$(|0\rangle \wedge |1\rangle \wedge |2\rangle) \otimes (|0\rangle \wedge |1\rangle) \otimes |0\rangle\ ,$$ which is restricted to an irrep\footnote{For the readers familiar with Young diagrams, this irrep corresponds to the diagram $\tiny\ydiagram{3,2,1}$. See \cref{app:Young-diagrams} for more details.} of $\SU(3)$ (see the definition of the wedge product in \cref{eq:wedge}). The value of the function $f_{\sgn}$ for this initial state is zero. We first evolve the system under an $\SU(3)$-invariant Hamiltonian that can be written as $\sum_{i<j} h_{ij}\P_{ij}$, i.e., is 2-local (see the caption for further details). Under this time evolution the value of $f_{\sgn}$ remains zero. At $t=100$ we add the 3-local term $\P_{12}\P_{23} +\P_{23}\P_{12}$, which violates the condition in \cref{bbc}, to the Hamiltonian. As we see in this plot, the function $f_{\sgn}$ starts changing for $t>100$. In this example we see that, even for states restricted to a single charge sector, there are additional conservation laws that are respected by 2-local $\SU(3)$-invariant unitaries, but are violated by general 3-local ones.

\begin{figure}[htp]
  \includegraphics[width=0.49\textwidth]{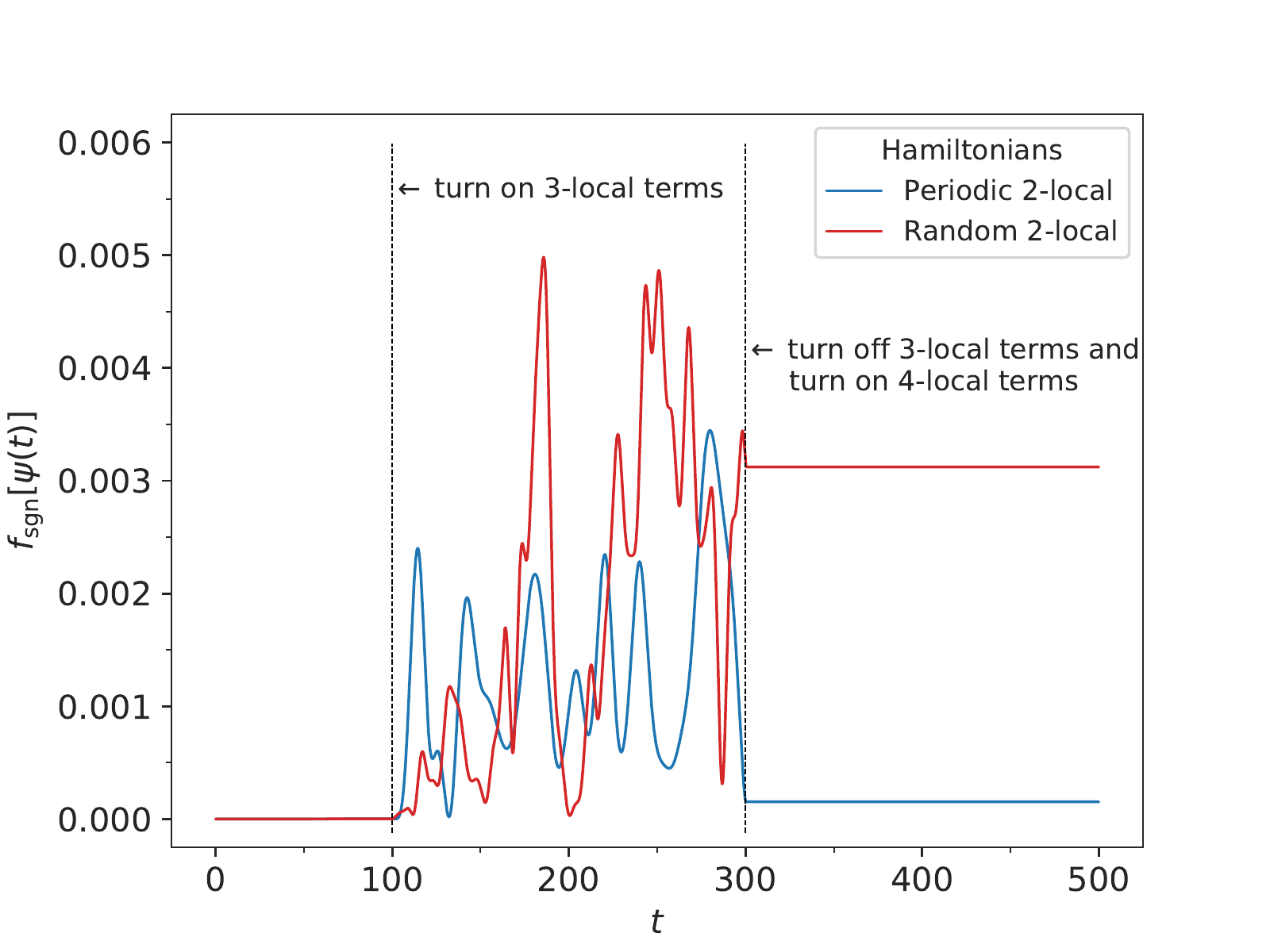}
  \caption{\emph{An example of conservation laws imposed by the $\mathbb{Z}_2$ symmetry}: Consider a system with 6 qutrits in the initial state $ (|0\rangle \wedge |1\rangle \wedge |2\rangle) \otimes (|0\rangle \wedge |1\rangle) \otimes |0\rangle$. It turns out that this state is restricted to a single irrep of $\SU(3)$ (see \cref{app:Young-diagrams}). For this state the value of $f_{\sgn}$ is zero. For $t\le 100$, this initial state evolves under Hamiltonians that can be written as a sum of 2-local $\SU(3)$-invariant terms, namely a Hamiltonian with random 2-local interactions between all pairs of qutrits (the red curve) and a translationally-invariant 2-local Hamiltonian with nearest-neighbor interactions on a closed chain (the blue curve). In both cases $f_{\sgn}$ remains zero. At $t=100$, we turn on the 3-local interaction $\P_{12}\P_{23} +\P_{23}\P_{12}$, which violates the $\mathbb{Z}_2$ symmetry, and $f_{\sgn}$ starts changing. At $t=300$, we turn off the previous 3-local interaction and turn on the 4-local interaction $\P(1234) + \P(4321)$. This 4-local interaction satisfies the $\mathbb{Z}_2$ symmetry and $f_{\sgn}$ remains constant (compare this with the conservation law discussed in \cref{fig:l-shape}).} 
  \label{fig:4-local-fskew}
\end{figure}

It is also worth noting that the function $f_{\sgn}$ may also remain conserved by $\SU(d)$-invariant Hamiltonians that are not 2-local, as long as they satisfy \cref{bbc}. To demonstrate this fact, in the above example we turn off the term $\P_{12}\P_{23} +\P_{23}\P_{12}$ at $t=300$ and turn on the 4-local Hermitian term $\P(1234)+\P(4321)$, where $\P(abcd)\equiv\P_{ab}\P_{bc}\P_{cd}$ is the cyclic permutation of $a,b,c,d$. Since this permutation has even parity and in the Hamiltonian it has a real coefficient, the symmetry condition in \cref{assump} is satisfied, and hence $f_{\sgn}$ remains conserved (see \cref{fig:4-local-fskew}). Note that this 4-local term is not in the Lie algebra generated by 2-local $\SU(3)$-invariant Hamiltonians.\footnote{ This can be seen, for instance, using the result of \cref{conserved}. In particular, in the example in \cref{sect:examp}, the time evolution under this Hamiltonian violates a conservation law that should be satisfied by Hamiltonians in the Lie algebra generated by 2-local $\SU(d)$-invariant Hamiltonians.} Therefore, the unitary time evolution generated by this Hamiltonian, in general, cannot be simulated using 2-local $\SU(3)$-invariant unitaries. We conclude that going beyond the group generated by 2-local unitaries does not guarantee violation of this conservation law. 
 
Finally, we remark on a peculiar feature of this conservation law, which makes it different from the standard conservation laws in \cref{consv}. As we mentioned before, for systems of $n>d^2$ qudits, the projector $K$ is zero which implies $f_{\sgn}$ vanishes for all states. In particular, in the context of the above example, this means that if we consider a system with, e.g., $n=10$ qutrits, then, even in the presence of the term $\P_{12}\P_{23} +\P_{23}\P_{12}$, the function $f_{\sgn}$ remains zero.

\subsection{Further properties of function \texorpdfstring{$f_{\sgn}$}{fsgn}}

The function $f_{\sgn}$ satisfies various other useful properties. In particular, as we show in \cref{App:Kf}:
\begin{enumerate}[(i)]
\item It remains invariant under permutations and global rotations, i.e., for all $U\in \SU(d)$ and $\sigma\in\mathbb{S}_n$, the value of $f_{\sgn}$ for states $|\psi\>$, $U^{\otimes n}|\psi\>$ and $\P(\sigma)|\psi\>$ are equal.
\item It vanishes for states that remain invariant under a swap, up to a sign, i.e., if $\P_{ij} |\psi\>=\pm |\psi\>$, then $f_{\sgn}(|\psi\>)=0$.
\item It is zero for $n>d^2$. In the special case of $n=d^2$, it is non-zero if and only if ${d(d-1)}/{2}$ is even.
\item In the special case of product states, where $|\psi\>=\bigotimes_{j=1}^n |\phi_j\>$, we find
\begin{align}\label{det}
  f_{\sgn}[\psi] &= \frac{1}{n!} \sum_{\sigma\in\mathbb{S}_n} \sgn(\sigma)\prod_{j} \langle\phi_j|\phi_{\sigma(j)}\rangle^2 \nonumber\\ 
              &=\frac{1}{n!}\det(\<\phi_i|\phi_j\>^2) \ ,
\end{align}
where $\det(\<\phi_i|\phi_j\>^2)$ is the determinant of $n\times n$ matrix with the matrix elements $\<\phi_i|\phi_j\>^2$. Note that this matrix is the Gram matrix of states $\{|\phi_i\>|\phi_i\>: i=1,\cdots, n\}$. This means that $f_{\sgn}$ is strictly greater than zero if and only if this is a linearly independent set of vectors (for example, $f_{\sgn}$ is non-zero for the 3 qubit state $|0\>|1\>|+\>$, whereas according to property (iii) it is zero for any 4 qubit state).
\end{enumerate}

\section{Correspondence to a free fermionic system and its implications} \label{Sec:ferm}

Next, we discuss a different type of constraint on the group generated by 2-local $\SU(d)$-invariant unitary transformations for $d\ge 3$, and derive its corresponding conservation laws. In contrast to the one found in the previous section, the following conservation laws remain non-trivial for systems with arbitrarily large number of qudits.

To explain this type of constraint, again we start with a simple example. Let $\{|m\>: m=0, \cdots, d-1\}$ be an arbitrary orthonormal basis of a qudit Hilbert space. Consider the subspace spanned by the state $|1\>|0\>^{\otimes (n-1)}$ and its permuted versions. This subspace is invariant under the action of 2-local unitaries $\{\e^{\i \theta \P_{ab}}\}$. Then, we can interpret $|1\rangle$ as a single ``particle'' moving on $n$ sites. In particular, the state $|0\>^{\otimes (j-1)} |1\> |0\>^{\otimes (n-j)}$ can be interpreted as a particle located on site $j$. If the particle is not on site $a$ or $b$, under the action of the unitary $\e^{\i \theta \P_{ab}}$ the state obtains a phase $\e^{\i\theta}$; on the other hand, if it is located on sites $a$ or $b$, then with the amplitude $\cos \theta$ it remains on the same site, and with the amplitude $\i\sin \theta$ moves to the other site. This interpretation can be extended to the cases with more than one particle. What is remarkable about this interpretation is that if the particles are in an anti-symmetric state, e.g., $\frac{1}{\sqrt{2}}(|1\>|2\>-|2\>|1\>) |0\>^{\otimes (n-2)}$ then the time evolution under unitaries $\{ \e^{\i \theta \P_{ab}}\}$ corresponds to the time evolution of free (non-interacting) fermionic particles. 

In the following, we further formalize and prove this claim. It should be noted that the final result, as presented in \cref{lem3} and the conservation laws in \cref{Thm3}, are about qudit systems and can be understood and applied independent of this fermionic picture. 

\subsection{The Lie group generated by exponentials of swap operators on the sites of a fermionic system}\label{fermionic group}

In this section we introduce a representation of the permutation group $\mathbb{S}_n$ on a fermionic system and characterize the Lie group generated by the exponential of swaps (transpositions) on the sites of this system. In the next section, we apply this result to study systems of qudits. It is worth noting that the results presented in this section can be understood independently of the other sections and could be of independent interest.
 
For a fermionic system with $n$ sites, let $c^\dag_i$ be the creation operator that creates a particle on site $i$ for $i=1,\cdots, n$. These operators satisfy the usual fermionic anti-commutation relations $\{c_i^\dag, c_j\}=\delta_{ij}$ and $\{c_i^\dag, c^\dag_j\}=0$. Let $|\mathrm{vac}\>$ be the Fock vacuum, satisfying $c_i|\mathrm{vac}\>=0$, for all $i=1,\cdots, n$.

We can define a natural representation of the permutation group on the space of creation operators: under the permutation $\sigma\in\mathbb{S}_n$, the creation operator $c^\dag_j$ is transformed to $c_{\sigma(j)}^\dag$. In particular, for all transpositions (swaps) $\sigma_{ab}\in\mathbb{S}_n$, it holds that 
\be\label{perm0}
{\P}^f_{ab}\ c^\dag_j\ {{\P}^f_{ab}}^\dag=c^\dag_{\sigma_{ab}(j)}\ \ \ \ \ : j=1,\cdots, n\ . 
\ee
It turns out that these $n$ equations uniquely specify the operator ${\P}^f_{ab}$, up to a global phase, namely 
\be\label{def:P}
\P^f_{ab}\equiv \mathbb{I}^f - (c_a^\dagger - c_b^\dagger) (c_a - c_b)\ ,
\ee
where $\mathbb{I}^f $ is the identity operator. More precisely, as we show in \cref{app:fermi-Pf}, the operator in \cref{def:P} is the unique operator satisfying both \cref{perm0} and equation 
\be\label{perm5}
\P^f_{ab} |\mathrm{vac}\rangle= |\mathrm{vac}\rangle\ .
\ee
Note that the operator $\P^f_{ab}$ defined in \cref{def:P} is both Hermitian and unitary.

Since transpositions generate the permutation group $\mathbb{S}_n$, \cref{perm0} defines a representation of this group on the vector space spanned by the creation operators. This also induces a representation of this group on the $2^n$-dimensional Hilbert space of the fermionic system, also known as the Fock space. In particular, under the transposition $\sigma_{ab}\in\mathbb{S}_n$, the $L$-particle state $c^\dag_{i_L} \cdots c^\dag_{i_1} |\mathrm{vac}\> $ is mapped to
\be\label{perm11}
\P^f_{ab} \big(c^\dag_{i_L}\cdots c^\dag_{i_1} |\mathrm{vac}\rangle\big)=c^\dag_{\sigma_{ab}(i_L)} \cdots c^\dag_{\sigma_{ab}(i_1)} |\mathrm{vac}\rangle\ .
\ee 
Having defined this representation of the permutation group, we now ask: what is the group of unitaries generated by Hamiltonians that can be written as a sum of swaps, i.e., Hamiltonians in the form $\sum_{a<b} h_{ab}(t) \P^f_{ab}$, where $h_{ab}$ is an arbitrary function of time? Equivalently, what is the group generated by unitaries $\{ \e^{\i\theta \P^f_{ab}}: \theta\in[0,2\pi), 1\le a<b\le n \}$? 

As we explain in the following, this group has a simple characterization. The key observation that allows us to characterize this group is the fact that, up to a constant shift, the swap operators $\{\P_{ab}^f\}$ are quadratic in the creation and annihilation operators. This implies that Hamiltonians in the form $\sum_{a<b} h_{ab}(t) \P^f_{ab}$ correspond to non-interacting (free) fermionic systems. It is worth noting that a similar representation of the permutation group can also be defined for bosons (see \cref{app:fermi-Pf}). However, unlike in the fermionic case, in the bosonic case the unitary that satisfies \cref{perm0} is not quadratic in creation/annihilation operators, and therefore the corresponding Hamiltonian is not free. 
 
Using \cref{def:P} together with the fermionic anti-commutation relations, it can be shown (see \cref{app:fermi-conjugate-c}) that 
\be\label{rep:fermi}
\e^{\i\theta \P^f_{ab}} c^\dag_j \e^{-\i\theta \P^f_{ab}}=\e^{-\i\theta}(\cos \theta~ c^\dag_{j} +\i \sin\theta~ c^\dag_{\sigma_{ab}(j)})\ .
\ee
The right-hand side of \cref{rep:fermi} can be interpreted as a unitary (Bogoliubov) transformation on the $n$-dimensional vector space spanned by the creation operators, which again indicates that under the Hamiltonian $ \P^f_{ab}$ the particles are not interacting with each other. To understand this better, it is useful to consider the single-particle version of the operator $\P_{ab}^f$: let $|j\rangle: j=1,\cdots, n$ be an orthonormal basis for $\mathbb{C}^n$. Consider the operator $E_{ab}$ acting on $\mathbb{C}^{ n}$ satisfying equations
\be
 E_{ab} |j\rangle=|\sigma_{ab}(j)\rangle\ \ \ \ \ : j=1,\cdots, n \ ,
 \ee
which correspond to the single-particle version of \cref{perm11}. This operator is explicitly given by
\be\label{eq:def-Eab}
E_{ab}\equiv\mathbb{I}- (|a\rangle-|b\rangle)(\langle a|-\langle b|)\ ,
\ee
which is also unitary and Hermitian. The set of unitaries $\{E_{ab}:1 \le a<b\le n\}$ generates a reducible representation of the permutation group $\mathbb{S}_n$. In particular, this representation leaves the vector $\sum_{j=1}^n |j\rangle$ invariant. It turns out that the $(n-1)$-dimensional subspace orthogonal to this vector, corresponds to an irrep of $\mathbb{S}_n$, called the \emph{standard} representation \cite{fulton2013representation}.

The operator $\mathbb{I}-E_{ab}= (|a\rangle-|b\rangle)(\langle a|-\langle b|)$ can be interpreted as the single-particle version of the operator $ \mathbb{I}^f-\P^f_{ab}= (c_a^\dagger - c_b^\dagger) (c_a - c_b)$ in the following sense: using the identity $\e^{\i \theta E_{ab}}=\cos\theta \mathbb{I}+\i \sin\theta E_{ab}$, \cref{rep:fermi} can be rewritten as 
\be\label{rep:fermi2}
\e^{\i\theta (\P^f_{ab}-\mathbb{I}^f)} c^\dag_j \e^{-\i\theta (\P^f_{ab}-\mathbb{I}^f)}=\sum_{l=1}^n \langle l|\e^{\i \theta (E_{ab}-\mathbb{I})}|j\rangle \ c^\dag_l\ .
\ee
Furthermore, using \cref{perm5} we find 
\be\label{inv-vac}
\e^{\i\theta (\P^f_{ab}-\mathbb{I}^f)} |\mathrm{vac}\rangle=|\mathrm{vac}\rangle\ .
\ee
Combining this with \cref{rep:fermi2}, we conclude that for any $L$-particle state $\prod_{s=1}^L c^\dag_{j_s} |\mathrm{vac}\rangle$ with $j_1,\cdots, j_s\in\{1,\cdots, n\}$, it holds that 
\begin{align}\label{rep:fermi3}
  \e^{\i\theta (\P^f_{ab}-\mathbb{I}^f)}& \prod_{s=1}^L c^\dag_{j_s} |\mathrm{vac}\rangle= \prod_{s=1}^L \Big[\sum_{k=1}^n \langle k|\e^{\i \theta (E_{ab}-\mathbb{I})}|j_s\rangle c^\dag_{k}\Big] |\mathrm{vac}\rangle\ .
\end{align}
Given that vectors $\{\prod_{s=1}^L c^\dag_{j_s} |\mathrm{vac}\rangle\}$ span the $L$-particle subspace of the Fock space, \cref{rep:fermi3} defines a representation of the Lie group generated by $n\times n$ unitaries $\{\e^{\i\theta (E_{ab}-\mathbb{I})}: \theta\in[0,2\pi), 1 \le a<b \le n\}$, inside this subspace. In particular, since this subspace is unitarily equivalent to the totally anti-symmetric subspace of $(\mathbb{C}^n)^{\otimes L}$, we conclude that inside the $L$-particle sector the unitary $\e^{\i \theta (E_{ab}-\mathbb{I})} $ is represented by
\be\label{homomorph}
\e^{\i \theta (E_{ab}-\mathbb{I})} \mapsto \e^{\i\theta (\P^f_{ab}-\mathbb{I}^f)}\Pi_L \cong \big(\e^{\i \theta (E_{ab}-\mathbb{I})}\big)^{\otimes L}P^{-}_{L, n} \ ,
\ee
where $\Pi_L$ is the projector to the $L$-particle sector of the Fock space and $P^{-}_{L, n}$ is the projector to the totally anti-symmetric subspace of $(\mathbb{C}^n)^{\otimes L}$. In \cref{app:fermi-single-fermion} we show that the unitaries $\{\e^{\i\theta (E_{ab}-\mathbb{I})}\}$ generate the group of all $n\times n$ unitaries that leave the vector $\sum^n_{j=1}|j\rangle$ invariant, which is isomorphic to the group $\U(n-1)$. We conclude that 
\begin{lemma}\label{lem9}
  Consider the Lie groups generated by unitaries $\{\e^{\i\theta (E_{ab}-\mathbb{I})}\}$ and $\{\e^{\i\theta (\P^f_{ab}-\mathbb{I}^f)}\}$ respectively,
\begin{align}\nonumber
  G_{\mathrm{single}}&\equiv\big\langle \e^{\i\theta (E_{ab}-\mathbb{I})}: \theta\in[0,2\pi), 1 \le a<b \le n\big\rangle\subset\U(n)\\ 
  G_\fermi&\equiv\big\langle\e^{\i\theta (\P^f_{ab}-\mathbb{I}^f)}: \theta\in[0,2\pi), 1 \le a<b \le n\big\rangle\subset\U(2^n).\nonumber
\end{align}
Both groups are isomorphic to $\U(n-1)$. Furthermore, for any integer $L$ in the interval $0<L<n$ the projection of $G_\fermi$ into the $L$-particle sector is also isomorphic to $\U(n-1)$, and is unitarily equivalent to the $L$-fold tensor product of $G_\mathrm{single}$ projected to the totally anti-symmetric subspace, as defined in \cref{homomorph}. 
\end{lemma}

In summary, the main observations that allow us to find this simple characterization are that the swap operators are free Hamiltonians on the fermionic space and the action of unitaries $\{\e^{\i\theta \P^f_{ab}}\}$ can be understood as a unitary (Bogoliubov) transformation on the creation operators (see \cref{rep:fermi}).

This argument reveals a simple physical interpretation and an independent proof of a remarkable mathematical result by Marin (namely, Lemma 12 of \cite{marin2007algebre}). As we mentioned before, Marin has found the simple factors of the Lie algebra generated by transpositions \cite{marin2007algebre}. It turns out that the simple subgroup $\SU(n-1)$ of the Lie group $\U(n-1)$ in our \cref{lem9} corresponds to one of these simple factors, which acts on certain irreps of $\mathbb{S}_n$ (namely, those with L-shape Young diagrams). These are exactly those irreps that appear in the above representation of $\mathbb{S}_n$ on the  fermionic system. In \cite{HLM_2021}, we discuss more about this relation and explain how the above connection with non-interacting fermionic systems provides new insight into this important result.\footnote{In particular, we show that the above lemma is essentially equivalent to Lemma 12 of \cite{marin2007algebre}.}

\vspace{6pt}

\noindent\textbf{Remark.} It is worth noting that $G_\mathrm{single}$ is a reducible representation of $\U(n-1)$. It contains a 1-dimensional invariant subspace corresponding to vector $\sum_{j=1}^n |j\rangle$ and the orthogonal $(n-1)$-dimensional subspace. Relative to this decomposition unitaries in $G_\mathrm{single}$ can be written as $G_\mathrm{single}=\{{1}\oplus T: T\in\U(n-1)\}$. This means that by projecting the group $G_\fermi$ to the $L$-particle sector, we obtain the group of unitaries 
\be\label{eq:equiv}
(1\oplus T)^{\otimes L}P^{-}_{L, n}\cong T^{\otimes L}P^{-}_{L, n-1}\oplus T^{\otimes (L-1)}P^{-}_{L-1, n-1}\ ,
\ee
where $T\in \U(n-1)$. We can rewrite this more compactly as 
\be\label{eq:equivwedge}
(1\oplus T)^{\wedge L}\cong T^{\wedge L}\oplus T^{\wedge (L-1)}\ ,
\ee
where $A^{\wedge L}\equiv A^{\otimes L}P^{-}_{L, m}$ if $A$ is an operator acting on a Hilbert space of dimension $m$. In other words, this means that for $L$ in the interval $0< L< n$, the $L$-particle sector contains two irreps of $\U(n-1)$: one corresponding to $L$ copies of $\U(n-1)$ projected to the totally anti-symmetric subspace of $(\mathbb{C}^{n-1})^{\otimes L}$ and another corresponding to $L-1$ copies of $\U(n-1)$ projected to the totally anti-symmetric subspace of $(\mathbb{C}^{n-1})^{\otimes (L-1)}$. See \cref{app:wedge} for a proof of these isomorphisms.

\subsection{Correspondence between the fermionic system and a subspace of the qudit system}\label{corr}

\begin{table*}[htb]
  \centering
  \renewcommand{\arraystretch}{2}
  \setlength{\tabcolsep}{4pt}
  \begin{tabular}{|cl|l|}
    \TopRule
    & Qudit system & Fermionic system \\ \TopRule
    & $\frac{1}{\sqrt{2}} (|1\>|2\> - |2\>|1\>) \otimes |0\>^{\otimes (n-2)}$ & $c_2^\dagger c_1^\dagger |\mathrm{vac}\>$ \\ \MidRule
    & $\P(\sigma) \big(\frac{1}{\sqrt{2}} (|1\>|2\> - |2\>|1\>) \otimes |0\>^{\otimes (n-2)}\big)$ & $c_{\sigma(2)}^\dagger c_{\sigma(1)}^\dagger |\mathrm{vac}\>$ \\ \MidRule
    Swaps & $\P_{ab}$ & $\P^f_{ab} = \mathbb{I}^f - (c_a^\dagger - c_b^\dagger) (c_a - c_b)$ \\ \MidRule
    2-local & $H(t) = \sum\limits_{i<j} h_{ij}(t)\P_{ij}$ & $H^f(t) = \sum\limits_{i<j} h_{ij}(t) \big(\mathbb{I}^f - (c_i^\dagger - c_j^\dagger) (c_i - c_j) \big)$ \hfill Free \\ \MidRule
    3-local & $H(t) = \sum\limits_{i,j,k} h_{ijk}(t) \P_{ij} \P_{jk}$ & \thinmuskip=0mu \medmuskip=0mu \thickmuskip=0mu $H^f(t) = \sum\limits_{i,j,k} h_{ijk}(t) \big(\mathbb{I}^f - (c_i^\dagger - c_j^\dagger) (c_i - c_j) \big) \big(\mathbb{I}^f - (c_j^\dagger - c_k^\dagger) (c_j - c_k) \big)$ \hfill Interacting \\ \BotRule
  \end{tabular}
  \caption{The correspondence between the fermionic system and a subspace of the qudit system}
  \label{tab:qudit-fermion-correspondence}
\end{table*}

Next, we introduce a correspondence between states in a certain subspace of $(\mathbb{C}^d)^{\otimes n}$ and states of the above fermionic system. For $L\le \min\{n, d-1\}$, consider the state $\big(\bigwedge_{m=1}^{L}|m\>\big) \otimes |0\>^{\otimes (n-L)}$, where 
\begin{equation}\label{eq:wedge}
  \bigwedge_{m = 1}^L|m\rangle \equiv \frac{1}{\sqrt{L !}}\sum_{\sigma \in \mathbb{S}_L} \sgn(\sigma) |\sigma(1)\rangle \cdots |\sigma(L)\rangle\ ,
\end{equation}
is a totally anti-symmetric state in $({\mathbb{C}^d})^{\otimes L}$ and $\{|m\>: m=0, \cdots, d-1\}$ is the \emph{computational} basis for a single qudit. We can identify the subspace spanned by the state $\big(\bigwedge_{m=1}^{L}|m\>\big) \otimes |0\>^{\otimes (n-L)}$, and its permuted versions, with the $L$-particle sector of the above fermionic system via the map
\be\label{eq:perm}
U^f \P(\sigma) \Big[\big(\bigwedge_{m=1}^{L}|m\>\big) \otimes |0\>^{\otimes (n-L)}\Big] = \prod_{i=1}^{L} c_{\sigma(i)}^\dag |\mathrm{vac}\>\ ,
\ee
for all $\sigma\in\mathbb{S}_n$. It can be shown that $U^f$ is a linear isometry, i.e., preserves inner products (see \cref{app:fermi-Uf}). This correspondence holds for all values of $L\le \min\{n, d-1\}$. Considering the subspaces corresponding to all values of $L\le \min\{n, d-1\}$ together, we obtain the following subspace of the qudit system:
\begin{align}\label{subspace}
  \hilbert_\comp\equiv \mathrm{span}_\mathbb{C}\Big\{ &\P(\sigma)\Big[\big(\bigwedge_{m=1}^{L}|m\>\big) \otimes |0\>^{\otimes (n-L)}\Big]\nonumber \\ 
                                                &:\sigma\in\mathbb{S}_n , L\le\min\{n, d-1\} \Big\}\ .
\end{align} 
Note that this definition means $\hilbert_\comp$ is invariant under permutations. Equivalently, the projector to $\hilbert_\comp$, denoted by $\Pi_\comp$, satisfies $\P(\sigma)\Pi_\comp=\Pi_\comp \P(\sigma)$ for all $\sigma\in\mathbb{S}_n$. Furthermore, according to Schur-Weyl duality (see \cref{app:duality}), any $\SU(d)$-invariant unitary $V\in\mathcal{V}_{n}$ can be written as a linear combination of permutations, and hence preserves $\hilbert_\comp$, i.e., $\Pi_\comp V=V\Pi_\comp$. On the other hand, $\hilbert_\comp$ depends on the choice of basis $\{|m\>: m=0, \cdots, d-1\}$, i.e., is not invariant under the action of $\SU(d)$. We discuss more about this later (see the discussion around \cref{eq:renyi2}). 

Using \cref{perm11}, we can see that $U^f$ intertwines the representation of $\mathbb{S}_n$ on the subspace $\hilbert_\comp$ of $(\mathbb{C}^d)^{\otimes n}$ with the representation on the fermionic system, i.e., 
\be
\forall \sigma\in\mathbb{S}_n:\ \ \ U^f {\P}(\sigma) \Pi_\comp= {\P}^f(\sigma) U^f \Pi_\comp\ .
\ee
In particular, this implies 
\be\label{perm65}
U^f \e^{\i\theta {\P}_{ab}} \Pi_\comp=\e^{\i\theta {\P}^f_{ab}} U^f \Pi_\comp\ ,
\ee
for all $a,b \in\{1,\cdots, n\}$ and $\theta\in[0,2\pi)$. 

The linear map $U^f$ establishes a useful correspondence between states and observables in the subspace $\hilbert_\comp$ of the qudit system on the one hand, and those of the fermionic system on the other hand. Using this correspondence together with the results of the previous section on the fermionic system, we can easily understand and characterize the properties of the group generated by 2-local $\SU(d)$-invariant unitaries inside the subspace $\hilbert_\comp$. 

In particular, combining \cref{perm65} together with \cref{def:P}, we arrive at a remarkable conclusion: for states inside $\hilbert_\comp$, the time evolution of qudits under a general 2-local $\SU(d)$-invariant Hamiltonian $H(t)=\sum_{i<j}^{} h_{ij}(t)\ \P_{ij}$ is equivalent to the time evolution of the fermionic system under the non-interacting (free) Hamiltonian
\be\label{Hamiltonian2}
H^f(t)=\sum_{i< j} h_{ij}(t) \Big[\mathbb{I}^f - (c_i^\dagger - c_j^\dagger) (c_i - c_j)\Big] \ ,
\ee
also known as the tight-binding model. Note that for $d>2$, the fermionic model associated to the qudit system contains more than one particle (i.e., $U^f \Pi_\comp$ has components in sectors with more than one particle). In these cases the fact that under the Hamiltonian in \cref{Hamiltonian2} these particles are non-interacting has important implications, such as conservation laws presented in the next section. 

Also, note that according to Schur-Weyl duality, any $\SU(d)$-invariant Hamiltonian can be written as a polynomial of swaps $\{{\P}_{ab}: a<b\}$. For any such Hamiltonian, we can obtain the corresponding fermionic Hamiltonian by applying the mapping ${\P}_{ab}\mapsto {\P}^f_{ab}$. Then, as presented in \cref{tab:qudit-fermion-correspondence}, for a general 3-local $\SU(d)$-invariant Hamiltonian the corresponding fermionic Hamiltonian is interacting. 

In fact, using this correspondence we can fully characterize the action of $\mathcal{V}_{2}$, the group generated by 2-local $\SU(d)$-invariant unitaries, inside $\hilbert_\comp$. In \cref{lem9} we characterized the group $G_\fermi$ generated by $\{\e^{\i\theta (\P^f_{ab}-\mathbb{I})}: \theta\in[0,2\pi), 1 \le a<b \le n\}$ and showed that it is isomorphic to $ \U(n-1)$. Combining this fact with the correspondence in \cref{perm65} we find that the projection of the group generated by unitaries $\{\e^{\i \theta (\P_{ab}-\mathbb{I})}: \theta\in[0,2\pi), 1\le a<b\le n\}$ to the subspace $\hilbert_\comp$, is also isomorphic 
 to $\U(n-1)$. Recall that together with the global phases $\{\e^{\i\theta }\mathbb{I}\}$, these unitaries generate 
$\mathcal{V}_{2}$. This implies that
 \be
 \{V\Pi_\comp: V\in\mathcal{V}_{2} \} \cong \U(1)\times \U(n-1)\ ,
 \ee 
where the operators on the left-hand side are interpreted as unitary transformations on $\hilbert_\comp$. To understand the origin of the $\U(1)$ factor on the right-hand side, note that the vector $\sum_{r=1}^n |0\>^{\otimes (r-1)} |1\> |0\>^{\otimes (n-r)}$ is inside $\hilbert_\comp$, and remains invariant under all unitaries $\{\e^{\i \theta (\P_{ab}-\mathbb{I})}: \theta\in[0,2\pi), 1\le a<b\le n\}$, whereas it obtains a phase under global phases $\{\e^{\i\theta }\mathbb{I}\}$. Therefore, this 1D subspace is a faithful representation of the $\U(1)$ group corresponding to the global phases.

As a simple example, in the case of $n$ qubits, $\hilbert_\comp$ is $n$-dimensional and decomposes to a 1D subspace corresponding to the vector $\sum_{r=1}^n |0\>^{\otimes (r-1)} |1\> |0\>^{\otimes (n-r)}$, which lives in the highest angular momentum sector $j_{\max}=\frac{n}{2}$, and an $(n-1)$-dimensional orthogonal subspace that lives in the sector with angular momentum $j_{\max}-1$. Then, the above statement implies that all unitaries inside the latter subspace can be realized using 2-local rotationally-invariant unitaries. 

\subsection{Conservation laws based on the qudit-fermion correspondence}\label{conserved}

For systems evolving under non-interacting Hamiltonians, the entanglement between particles remains conserved. Combining this fact with the above correspondence, in the following, we derive new conservation laws that hold for 2-local $\SU(d)$-invariant unitaries and are violated by 3-local ones. 

For any state $|\psi\rangle\in(\mathbb{C}^d)^{\otimes n}$ of qudits consider its component in $\hilbert_\comp$, i.e., $\Pi_\comp|\psi\rangle$. Then, 
\be
|\Psi^f\rangle\equiv U^f \Pi_\comp|\psi\rangle
\ee
is the corresponding (unnormalized) state of the fermionic system. Hence, the corresponding single-particle reduced state $\Omega[\psi]=\sum_{i,j=1}^n \Omega_{ij}[\psi]\ |i\rangle\langle j|\ $ is an (unnormalized) density operator defined on the Hilbert space $\mathbb{C}^{n}$, with matrix elements
\begin{align}\label{matrix}
  \Omega_{ij}[\psi] &\equiv {\< \Psi^f |c^\dag_j c_i |\Psi^f\>}: \ \ i,j=1,\cdots, n\ ,
\end{align}
where $\psi=|\psi\rangle\langle\psi|$. As an example, in \cref{app:fermi-1p-example} we show that for state $|\psi\rangle=\sum_{i=1}^n \psi_i |0\>^{\otimes (i-1)} |1\> |0\>^{\otimes (n-i)}$, the corresponding single-particle reduced state is $\Omega[\psi]=\sum_{ij} \psi_i \psi_j^\ast |i\rangle\langle j|$, which is the density operator for the state vector $\sum_{i=1}^n \psi_i |i\rangle$. This justifies the interpretation of $\Omega[\psi]$ as the ``single-particle'' reduced state. 

Note that this definition can be easily generalized to the case of mixed states. For a general density operator $\rho$ on $(\mathbb{C}^d)^{\otimes n}$, define
\be\label{def:omega}
\Omega[\rho]\equiv\sum_{i,j=1}^n \Omega_{ij}[\rho]\ |i\rangle\langle j|\ ,
 \ee
where
\begin{align}\label{eq:Omega_fermion}
  \Omega_{ij}[\rho]\equiv\Tr(c^\dag_j c_i U^f\Pi_\comp\rho\Pi_\comp U^{f\dagger})\ ,
\end{align}
or equivalently, in terms of the qudit operators
\thinmuskip=1mu 
\medmuskip=2mu 
\thickmuskip=3mu
\begin{align}\label{single-particle}
  \Omega_{ij}[\rho]\equiv \begin{cases}
    \Tr\big(\Pi_\comp \rho\Pi_\comp\ [\P_{ij} Q_{ij}]\ \big) &:\ i\neq j,\\
    \Tr\big(\Pi_\comp \rho\Pi_\comp\ [\mathbb{I}_i-|0\rangle\langle 0|_i]\big) &:\ i=j,
  \end{cases}
\end{align}
\thinmuskip=3mu 
\medmuskip=4mu 
\thickmuskip=5mu
where $A_i$ denotes $\mathbb{I}^{\otimes (i-1)}\otimes A \otimes \mathbb{I}^{\otimes (n-i)}$, and therefore $\mathbb{I}_i-|0\rangle\langle 0|_i$ is the projector to the subspace where qudit $i$ is orthogonal to state $|0\rangle$. Similarly, 
\begin{align}
  Q_{ij} \equiv (\mathbb{I}_i - |0\>\<0|_i) |0\>\<0|_j
\end{align}
is the projector to the subspace of states in which qudit $i$ is orthogonal to state $|0\rangle$ and qudit $j$ is in state $|0\rangle$. In \cref{app:fermi-equiv} we show that the above two expressions for $\Omega_{ij}[\rho]$ are indeed equivalent. Furthermore, using the fermionic picture, in particular \cref{rep:fermi2}, in \cref{app:fermi-1p-reduced} we show
\begin{lemma}\label{lem3}
  The linear map $\Omega$ defined in \cref{def:omega,single-particle}, satisfies the covariance condition 
\be\label{cov}
\Omega[\e^{\i\theta \P_{ab}}\rho\e^{-\i\theta \P_{ab}}]=\e^{\i\theta E_{ab}}\Omega[\rho]\e^{-\i\theta E_{ab}}\ ,
\ee
for all $\theta\in[0,2\pi)$ and any pair of sites $a,b\in\{1,\cdots, n\}$, where $E_{ab}$ is given in \cref{eq:def-Eab}. Furthermore, $\Omega$ is a completely-positive map. 
\end{lemma}
\cref{cov} is again a manifestation of the fact that inside $\hilbert_\comp$ the time evolution under Hamiltonian $\P_{ab}$ is equivalent to the time evolution under the single-particle Hamiltonian $E_{ab}$ (note that using \cref{lem9} this condition can be understood as the covariance of the map $\Omega$ with respect to the group $\U(n-1)$). 

Using this result we can derive simple constraints on the time evolution under 2-local $\SU(d)$-invariant unitaries. In particular, \cref{cov} implies that, if under $\e^{\i\theta \P_{ab}}$ the initial state $\rho$ evolves to $\rho'=\e^{\i\theta \P_{ab}}\rho\e^{-\i\theta \P_{ab}}$, then the eigenvalues of $\Omega[\rho]$ and $\Omega[\rho']$ are equal. An immediate corollary of this covariance property is the following conservation laws: 
 
\begin{theorem}\label{Thm3}
  Suppose under a unitary transformation $V$ in the group generated by 2-local $\SU(d)$-invariant unitaries, the initial density operator $\rho$ of $n$ qudits evolves to the density operator $\rho'=V\rho V^\dag$. Then for all positive integers $l$, we have
\be\label{eq:renyi}
\Tr(\Omega[\rho']^l)=\Tr(\Omega[\rho]^l)\ ,
\ee
where $\Omega$ is the completely-positive linear map defined in \cref{def:omega,single-particle}. 
\end{theorem}
In the fermionic picture these conservation laws can be understood as the conservation of correlations between particles: a free Hamiltonian does not generate/destroy such correlations. Specifically, \cref{eq:renyi} means that all R\'enyi entropies of the single-particle density operator $\omega[\rho] = \Omega[\rho]/\Tr(\Omega[\rho])$ are conserved. In the next section, we consider an example of a 6 qutrit system and show that the conservation law for $l=2$ is satisfied by the time evolutions generated by Hamiltonians that can be written as a sum of 2-local $\SU(3)$-invariant terms whereas it is violated by those general Hamiltonians that contain 3-local and 4-local terms.

It is worth noting that in the special case of $l=1$, the above conservation law holds for all $\SU(d)$-invariant unitaries $V\in \mathcal{V}_{n}$. 
To see this note that 
\begin{align}
  \Tr(\Omega[\psi]) &= \<\psi |\Pi_\comp \Big[\sum_i \mathbb{I}_i-|0\>\< 0|_i\Big] \Pi_\comp |\psi\>\ .
\end{align}
Then, because any $\SU(d)$-invariant unitary $V\in \mathcal{V}_{n}$ can be written as a linear combination of permutations, it commutes with permutationally-invariant operators $\Pi_\comp$ and $\sum_i \mathbb{I}_i-|0\>\< 0|_i$, which implies $\Tr(\Omega[V\psi V^\dag]) = \Tr(\Omega[\psi])$ (in fact, using \cref{eq:Omega_fermion}, we can interpret $\Tr(\Omega[\psi])$ as the expected number of ``particles'' in the fermionic subspace). 

Another special case where the conservation laws are satisfied trivially is the case of $d=2$: \cref{eq:renyi} holds for all $\SU(2)$-invariant unitaries and all $l$. In order to see this, notice that a general state in the the Hilbert space $(\mathbb{C}^2)^{\otimes n}$ can be decomposed as $|\psi\rangle=\sum_i \psi_i |\overline{\imath}\rangle+|\psi_\perp\rangle $, where $|\psi_\perp\rangle$ is orthogonal to the ``one-particle'' subspace spanned by $\{|\overline{\imath}\rangle\equiv |0\>^{\otimes (i-1)} |1\> |0\>^{\otimes (n-i)}: i=1,\cdots, n\}$. Under any $\SU(2)$-invariant unitary $V$, the one-particle subspace remains invariant and therefore $|\psi'\rangle=V|\psi\rangle=\sum_{i,j} V_{ji}\ \psi_i |\overline{\imath}\rangle+|\psi'_\perp\rangle $, where again $|\psi'_\perp\rangle$ is orthogonal to the one-particle subspace and $V_{ji}=\langle\overline{\jmath}| V |\overline{\imath}\rangle$. As we show in \cref{app:fermi-1p-example}, $\Omega[\psi]=\sum_{ij} \psi_i\psi^\ast_j |i\>\< j|$. We conclude that $\Omega[\psi']$ can be obtained from $\Omega[\psi]$ by conjugation with $n\times n$ unitary $\sum_{ij} V_{ji} |j\rangle\langle i|$, which implies the conservation laws in \cref{eq:renyi} hold trivially for all $\SU(2)$-invariant unitaries. 
 
Finally, we note that \cref{lem3} and the qudit-fermion correspondence discussed above impose further constraints on the dynamics, which are not fully captured by the above conservation laws in \cref{eq:renyi}. First, note that if under an $\SU(d)$-invariant unitary $V$, an initial density operator $\rho$ evolves to $\rho'=V\rho V^\dag$, then for any unitary $U\in\SU(d)$, the initial state $U^{\otimes n}\rho{U^{\otimes n}}^\dag$ evolves to $U^{\otimes n}\rho'{U^{\otimes n}}^\dag$. Therefore, applying the conservation laws in \cref{eq:renyi}, we find 
\be\label{eq:renyi2}
 \Tr\big(\Omega[U^{\otimes n}\rho' U^{\dag\otimes n}]^l\big)=\Tr\big(\Omega[U^{\otimes n}\rho U^{\dag\otimes n}]^l\big)\ ,
\ee
for all $U\in\SU(d)$. In general, these conservation laws are independent of those in \cref{eq:renyi}. This is a consequence of the fact that $\hilbert_\comp$ is defined in terms of the computational basis and is not invariant under unitaries $U^{\otimes n}: U\in\SU(d)$. This in turn implies the linear map $\Omega$ is not covariant under the action of $\SU(d)$.\footnote{Note that $\Omega[U^{\otimes n} \rho U^{\dagger\otimes n}]$ is the same as the single-particle reduced density matrix which would be obtained if instead of the computational basis we had used the basis $\cset{U \nket{m}}{m = 0, \cdots, d - 1}$. In that case, the fermionic correspondence would be between the fermionic system and the subspace $U^{\otimes n} \hilbert_\comp$, composed of the rotated versions of the states in \cref{subspace}.} 

In addition to these conservation laws, \cref{lem3} also implies that the dynamics in subspaces corresponding to different particle numbers are correlated: the ``single-particle'' dynamics determines the dynamics in all other sectors. In \cite{HLM_2021} we will discuss more about this lemma and other implications of the qudit-fermion correspondence.

\subsection{Example: A 6 qutrit system with \texorpdfstring{$\SU(3)$}{SU(3)} symmetry}\label{sect:examp}

Consider a system of 6 qutrits in the initial state
$$(|0\rangle\wedge |1\rangle\wedge |2\rangle)\otimes |0\rangle^{\otimes 3}\ .$$ 
Similar to the state considered in the example in \cref{sec:exZ2}, this state is also restricted to a single (but inequivalent) irrep\footnote{For the readers familiar with Young diagrams, this irrep corresponds to the diagram $\tiny\ydiagram{4,1,1}$. See \cref{app:Young-diagrams} for more details.} of $\SU(3)$. We study the dynamics of this system, and especially the time evolution of $\Tr(\omega(t)^2)$, i.e., the purity of the normalized density operator $\omega(t)=\Omega[\psi(t)]/ \Tr(\Omega[\psi(t)])$. Recall that $\Tr(\Omega[\psi(t)])$ is always automatically conserved under all $\SU(3)$-invariant unitaries and in this case is equal to $2$.

We study the time evolution under the same $\SU(3)$-invariant Hamiltonians considered in the previous example in \cref{fig:4-local-fskew}: first, we let the system evolve under a Hamiltonian that can be written as the sum of 2-local terms, i.e. $H=\sum_{i<j} h_{ij} \P_{ij}$, then we add the 3-local term $\P_{12}\P_{23} +\P_{23}\P_{12}$ to this Hamiltonian, and finally, we turn off this 3-local term and turn on the 4-local term $\P(1234) +\P(4321)$ (see the caption of \cref{fig:l-shape} for further details). The plot in \cref{fig:l-shape} clearly shows that $\Tr(\omega(t)^2)$ remains conserved for the first family of Hamiltonians, whereas it evolves under the second and third families. In particular, note that in the presence of the 4-local term $\P(1234) + \P(4321)$, the function $\Tr(\omega(t)^2)$ does not remain conserved, whereas the function $f_{\sgn}$ \emph{does}. We conclude that the two conservation laws are independent of each other. As we mentioned before, there also exist $\SU(3)$-invariant Hamiltonians that are not 2-local and yet respect these conservation laws (e.g., $\i(\P_{12}\P_{23}- \P_{23}\P_{12})=\i[\P_{12} , \P_{23}]$ is in the Lie algebra by 2-local $\SU(3)$-invariant Hamiltonians and therefore satisfies this property).

\begin{figure}[t]
  \includegraphics[width=0.49\textwidth]{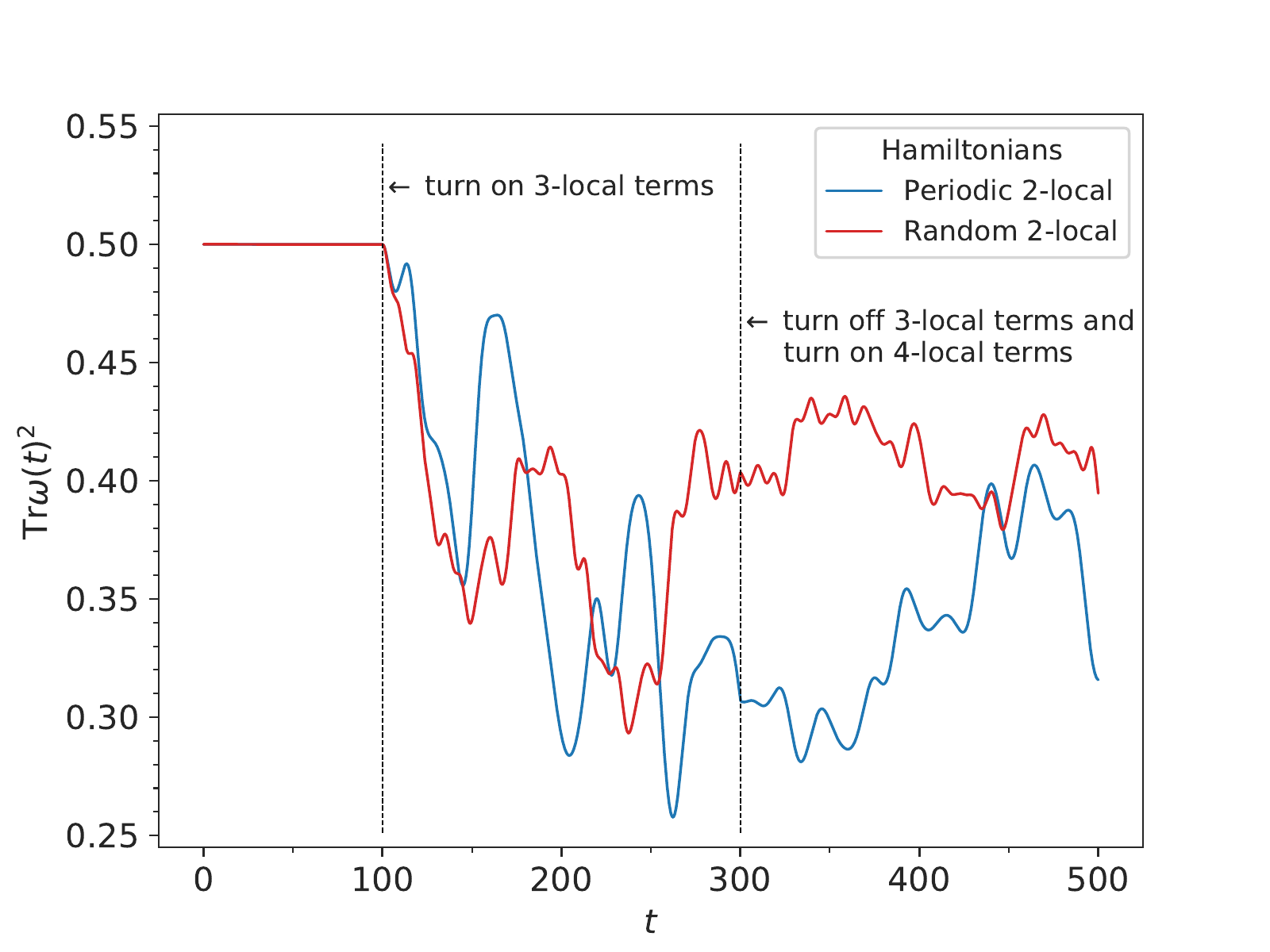}
  \caption{\emph{An example of the conservation laws implied by qudit-fermion correspondence}: 
    Consider a system with 6 qutrits in the initial state $|\psi\rangle = (|0\> \wedge |1\> \wedge |2\>) \otimes |0\>^{\otimes 3}$, which is restricted to a single irrep of $\SU(3)$ (see \cref{app:Young-diagrams}). Let $\omega(t)=\Omega[\psi(t)]/\Tr(\Omega[\psi(t)])$ be the single-particle reduced state associated to the state $|\psi(t)\rangle$, defined in \cref{single-particle}. The vertical axis is the purity of this state, i.e. $\Tr(\omega(t)^2)$. For $t\le 100$, the system evolves under 2-local $\SU(3)$-invariant Hamiltonians, namely a Hamiltonian with random 2-local interactions between all pairs of qutrits (the red curve) and a translationally-invariant 2-local Hamiltonian with nearest-neighbor interactions on a closed chain (the blue curve). In both cases the purity $\omega(t)$ remains constant. At $t=100$, we turn on the 3-local interaction $\P_{12}\P_{23} +\P_{23}\P_{12}$ and the purity starts changing. Finally, at $t=300$, we turn off this 3-local interaction and turn on the 4-local interaction $\P(1234) + \P(4321)$, and the purity is still not conserved. Recall that in \cref{fig:4-local-fskew} the function $f_{\sgn}$ remains conserved in the presence of this 4-local term. This clearly demonstrates that the conservation laws based on the qudit-fermion correspondence are independent of those based on the $\mathbb{Z}_2$ symmetry.}
  \label{fig:l-shape}
\end{figure}

\subsection*{Forbidden superpositions: understanding the phenomenon discussed in \texorpdfstring{\cref{Fig}}{Fig. 1}}

Using the tools and ideas developed in this section, we can now explain the phenomenon discussed in \cref{Fig}. Again, consider a system of 6 qutrits in the same initial state $(|0\rangle\wedge |1\rangle\wedge |2\rangle)\otimes |0\rangle^{\otimes 3}$. Starting with this state, via a sequence of 2-local $\SU(3)$-invariant unitaries we can obtain the orthogonal state $ |0\rangle^{\otimes 3}\otimes (|0\rangle\wedge |1\rangle\wedge |2\rangle)$, which lives in the same charge sector of $\SU(3)$ (this state can be obtained, e.g., by swapping qudits 1 with 4, 2 with 5, and 3 with 6). Consider the 2D subspace spanned by these two orthogonal vectors. Up to a global phase, any state in this subspace can be written as $|\psi(\theta,\phi)\rangle$ equal to
\thinmuskip=1mu 
\medmuskip=1mu 
\thickmuskip=1mu
\be
\cos\frac{\theta}{2}(|0\rangle\wedge |1\rangle\wedge |2\rangle)\otimes |0\rangle^{\otimes 3} +\e^{\i\phi}\sin\frac{\theta}{2} |0\rangle^{\otimes 3}\otimes (|0\rangle\wedge |1\rangle\wedge |2\rangle)\ ,
\ee
\thinmuskip=3mu 
\medmuskip=4mu 
\thickmuskip=5mu
with $\phi \in[0,2\pi)$ and $\theta\in[0,\pi]$. It can be easily shown that any pair of states in this subspace can be converted to each other via a general $\SU(3)$-invariant unitary. In particular, under unitaries generated by the 6-local $\SU(3)$-invariant Hamiltonian $\P_{14}\P_{25}\P_{36}$, the initial state $|\psi(0,0)\rangle$ evolves to state $|\psi(\theta,\frac{\pi}{2})\rangle$ for arbitrary $\theta\in[0,\pi]$. Then, by applying unitaries generated by 2-local $\SU(3)$-invariant Hamiltonian $\mathbb{I}-\P_{45}$, this state can be converted to $|\psi(\theta,\phi)\rangle$ for arbitrary $\phi\in[0,2\pi)$. Since these transformations are reversible via $\SU(3)$-invariant unitaries, we conclude that state $|\psi(\theta,\phi)\rangle$ can be converted to any other $|\psi(\theta',\phi')\rangle$ via $\SU(3)$-invariant unitaries.\footnote{ Alternatively, this can be shown using the fact that for such pairs of states, the condition in \cref{consv} is satisfied, i.e., for all $U\in \SU(3)$, the expectation value $\langle\psi(\theta,\phi)|U^{\otimes 6}|\psi(\theta,\phi)\rangle\ $, is independent of $\theta$ and $\phi$. Therefore, by the result of \cite{marvian2014asymmetry}, we conclude that states with different values of $\theta$ and $\phi$ can be converted to each other via $\SU(3)$-invariant unitaries.}

Next, we ask what transitions in this subspace are possible under unitaries generated by 2-local $\SU(3)$-invariant unitaries. In particular, what are the constraints imposed by the conservation law in \cref{eq:renyi}, e.g., for $l=2$. To apply this conservation law we need to find the single-particle density operator $\omega$ associated to the state $|\psi(\theta,\phi)\rangle$, which can be found using \cref{matrix}, or equivalently, using \cref{single-particle}. Using the fact that $\Pi_\comp|\psi(\theta,\phi)\rangle=|\psi(\theta,\phi)\rangle $, we find 
\be
\omega=\frac{\Omega[\psi(\theta,\phi)]}{\Tr\big(\Omega[\psi(\theta,\phi)]\big)}=\left(
\begin{array}{cc}
  [\cos\frac{\theta}{2}]^2\ \rho& \\
                       & [\sin\frac{\theta}{2}]^2\ \rho 
\end{array}
\right)\ ,
\ee
where $\rho$ is given by
\be
\rho=\frac{1}{6}\left(
  \begin{array}{ccc}
    +2& -1 & -1 \\
    -1 & +2 & -1 \\
    -1 & -1 & +2 
  \end{array}
\right)\ .
\ee
Note that $\omega$ is also the single-particle state associated to the fermionic state $U^f |\psi(\theta,\phi)\rangle$, which is equal to
\thinmuskip=0mu 
\medmuskip=0mu 
\thickmuskip=0mu
\be
\Big(\frac{\cos\frac{\theta}{2}}{\sqrt{3}}[c_1^\dag c_2^\dag-c_1^\dag c_3^\dag+c_2^\dag c_3^\dag]+\frac{\e^{\i\phi}\sin\frac{\theta}{2}}{\sqrt{3}} [c_4^\dag c_5^\dag-c_4^\dag c_6^\dag+c_5^\dag c_6^\dag]\Big)|\mathrm{vac}\rangle .
\ee
\thinmuskip=3mu 
\medmuskip=4mu 
\thickmuskip=5mu
Equivalently, $\rho$ can be interpreted as the single-particle reduced state for the fermionic state $\frac{1}{\sqrt{3}}[c_1^\dag c_2^\dag-c_1^\dag c_3^\dag+c_2^\dag c_3^\dag]|\mathrm{vac}\rangle$ defined on a lattice with 3 sites. Eigenvalues of $\rho$ are $\{\frac{1}{2},\frac{1}{2},0\}$, which implies the purity of $\omega$ is 
\be
\Tr(\omega^2)=\frac{\Tr\big(\Omega[\psi(\theta,\phi)]^2\big)}{\Tr\big(\Omega[\psi(\theta,\phi)]\big)^2}=\frac{1}{8}(3+\cos 2 \theta)\ .
\ee
From the conservation law in \cref{eq:renyi}, we know that this quantity remains conserved under 2-local $\SU(3)$-invariant unitaries. This means that under such unitary transformations the initial state $|\psi(\theta,\phi)\rangle $ evolves to the state $|\psi(\theta',\phi')\rangle $, only if $\theta' =\theta$ or $\theta' = \pi-\theta$. It turns out that this necessary condition is also sufficient: the sequence of 2-local unitaries $\P_{14}\P_{25}\P_{36}$ converts $|\psi(\theta,\phi)\rangle $ to $|\psi(\pi-\theta,-\phi)\rangle $. Furthermore, as we have seen above, under the unitaries generated by 2-local $\SU(3)$-invariant Hamiltonian $\mathbb{I}-\P_{45}$, initial state $|\psi(\theta,\phi)\rangle $ can be transformed to $|\psi(\theta,\phi')\rangle $ for any $\phi\in[0,2\pi)$. Therefore, to summarize, we have shown that 
 \be\label{con62}
|\psi(\theta,\phi)\rangle \xleftrightarrow[\text{$\SU(3)$-inv}]{\text{2-local}} |\psi(\theta',\phi')\rangle\ \ \Longleftrightarrow\ \theta'=\theta, \pi-\theta \ ,
\ee
which means inside this 2-dimensional subspace, the conservation law in \cref{eq:renyi} for $l=2$ fully characterizes the possible state transitions under unitaries generated by 2-local $\SU(3)$-invariant unitaries. We conclude that if one is restricted to 2-local $\SU(3)$-invariant unitaries, then starting from the initial state $(|0\rangle\wedge |1\rangle\wedge |2\rangle)\otimes |0\rangle^{\otimes 3}$ the only other reachable state in this subspace is the state $ |0\rangle^{\otimes 3}\otimes (|0\rangle\wedge |1\rangle\wedge |2\rangle)$. That is, superpositions of these two states are not reachable.

Finally, in \cite{HLM_2021} we show that the constraints in \cref{con62} can be circumvented if the 6 qutrit system can interact with another 3 ancillary qutrits, which are initially prepared in the singlet state and return to the same state at the end of the process. That is,
\thinmuskip=1mu 
\medmuskip=2mu 
\thickmuskip=3mu
\be
|\psi(\theta,\phi)\rangle \otimes (|0\rangle\wedge |1\rangle\wedge |2\rangle)\xleftrightarrow{}|\psi(\theta',\phi')\rangle \otimes (|0\rangle\wedge |1\rangle\wedge |2\rangle)\ ,
\ee
\thinmuskip=3mu 
\medmuskip=4mu 
\thickmuskip=5mu
where the arrow indicates that the transition is possible under 2-local $\SU(3)$-invariant unitaries. In the language of quantum resource theories, the three ancillary qutrits in the above state conversion can be interpreted as a catalyst. 

\section{Quantum Circuits with random \texorpdfstring{$\SU(d)$}{SU(d)}-invariant unitaries}\label{Sec:Random}

Statistical properties of quantum circuits with random local gates have been extensively studied in the recent years (see e.g., \cite{emerson2005convergence, harrow2009random, brandao2016efficient, brandao2016local, Hamma:2012ov, hamma2012quantum}). Besides their applications in quantum information science (see e.g., \cite{gard2020efficient, streif2020quantum}), such circuits have become a standard model for studying complex quantum systems. For instance, they have been  used to  investigate thermalization and the {scrambling of information} in chaotic systems \cite{roberts2017chaos, nahum2017quantum, nahum2018operator, von2018operator}, as quantified by the  {out-of-time-ordered correlation} functions \cite{maldacena2016bound,  roberts2017chaos, shenker2014black},  and to probe  the conjectured role \cite{susskind2016computational, brown2016holographic, stanford2014complexity}  of  {quantum complexity}  in quantum gravity. In particular, quantum circuits with local symmetric gates have been considered as a model for quantum chaos in systems with conserved charges (see e.g., \cite{khemani2018operator, rakovszky2018diffusive, nakata2020generic, kong2021charge}). 
 
In the absence of symmetries, the distribution of unitaries generated by quantum circuits with random local unitaries converges to the uniform (Haar) distribution over the unitary group \cite{emerson2005convergence, harrow2009random}. One may naturally expect that a similar fact also holds in the presence of symmetries. That is, for circuits with sufficiently large number of random symmetric local unitaries, the distribution of unitaries generated by the circuit converges to the uniform distribution over the group of all symmetric unitaries. However, the results of \cite{marvian2022restrictions} and the present paper imply that this conjecture is wrong. 

Let $\mu_\mathrm{Haar}$ be the uniform distribution over the group of all $\SU(d)$-invariant unitaies $\mathcal{V}_{n}$. For circuits with random 2-local $\SU(d)$-invariant unitaries the distribution of generated unitaries converges to $\mu_{2\text{-}\mathrm{loc}}$, the uniform distribution over the subgroup $\mathcal{V}_{2}\subset \mathcal{V}_{n}$ generated by 2-local $\SU(d)$-invariant unitaries. Note that since both $\mathcal{V}_{2}$ and $\mathcal{V}_{n}$ are compact Lie groups, they both have a unique notion of uniform (Haar) distribution. Furthermore, because $\mathcal{V}_{2}$ is a proper Lie subgroup of $\mathcal{V}_{n}$ in general, $\mu_\mathrm{Haar}$ and $\mu_{2\text{-}\mathrm{loc}}$ are distinct distributions. 

In fact, as we explain below, our results imply that for $d\ge 3$, even the second moments of these distributions are different, i.e., 
\begin{align}\label{design}
  \mathbb{E}_{V\sim\mu_{2\text{-}\mathrm{loc}}}[V^{\otimes t}\otimes {V^\ast}^{\otimes t}]\neq \mathbb{E}_{V\sim\mu_\mathrm{Haar}}[V^{\otimes t}\otimes {V^\ast}^{\otimes t}]\ ,
\end{align}
for $t\ge 2$. A general distribution $\mu$ that satisfies \cref{design} as equality is called a $t$-design for $\mu_\mathrm{Haar}$. Therefore, the above claim means $\mu_{2\text{-}\mathrm{loc}}$ is not a 2-design for $\mu_\mathrm{Haar}$. 

This is an immediate corollary of our results in the previous section. Specifically, we saw that there exists a function, namely $\Tr(\Omega[\psi]^2)$, with the following properties: (i) it remains invariant under unitaries in $V\in \mathcal{V}_{2}$, i.e., $\Tr(\Omega[V\psi V^\dag]^2)=\Tr(\Omega[\psi]^2)$, whereas it can change under general $\SU(d)$-invariant unitaries in $\mathcal{V}_{n}$; and (ii) it is a quadratic polynomial in the density operator $\psi=|\psi\rangle\langle\psi|$.

To see how these properties imply the claim in \cref{design}, consider the function $f(V)\equiv \Tr(\Omega[ V\psi V^\dag ]^2)$ for a fixed pure state $\psi$ and arbitrary $V\in \mathcal{V}_{n}$. For arbitrary $W\in \mathcal{V}_{n}$ consider the expected value 
\be\label{w1}
\mathbb{E}_{V\sim\mu_\mathrm{Haar}} \ f(VW) =\mathbb{E}_{V\sim\mu_\mathrm{Haar}}\ f(V)\ ,
\ee
where the equality follows from the fact that $\mu_\mathrm{Haar}$ is invariant under all unitaries in $\mathcal{V}_{n}$. Therefore, for the unitary $V$ chosen according to the Haar measure over $\mathcal{V}_{n}$, the expected value of $f(VW)$ is independent of $W\in\mathcal{V}_{n}$. Next, consider the expected value of $f(VW)$, where $V$ is chosen uniformly from $\mathcal{V}_{2}$. Using the fact that $\Tr(\Omega[\psi]^2)$ remains conserved under unitaries in $\mathcal{V}_{2}$, we find 
\be\label{w2}
\mathbb{E}_{V\sim\mu_{2\text{-}\mathrm{loc}}}\ f(VW) =\Tr(\Omega[W\psi W^\dag]^2)=f(W)\ .
\ee
But, for a general $\SU(d)$-invariant unitary $W\in\mathcal{V}_{n}$, $f(W) = \Tr(\Omega[W\psi W^\dag]^2)$ depends non-trivially on $W$ (this is because the conservation law in \cref{eq:renyi} is violated by general $\SU(d)$-invariant unitaries). We conclude that for general $W\in\mathcal{V}_{n}$ the above two expected values in \cref{w1,w2} are not equal. Finally, we note that the function $f(V) = \Tr(\Omega[ V\psi V^\dag ]^2)$ can be written in the form of a linear functional of $V^{\otimes 2}\otimes {V^\ast}^{\otimes 2}$. In summary, we find that the expected values of $V^{\otimes 2}\otimes {V^\ast}^{\otimes 2}$ for two distributions $\mu_\mathrm{Haar}$ and $\mu_{2\text{-}\mathrm{loc}}$ are not equal, i.e., \cref{design} does not holds for $t\ge 2$.\footnote{It is worth noting that this argument can be generalized: suppose there is a function $f$ from states to complex numbers satisfying the following properties: (i) $f$ can be written as a polynomial of degree $t$ in the density operator; and (ii) $f$ remains conserved under unitaries in a compact subgroup $H$ of a compact group $G$, but is not conserved under some unitaries in $G$. Then, the uniform distribution over $H$ is not a $t$-design for the uniform distribution over $G$.}

We note that, instead of the above argument, this result can also be established using a general result of \cite{zimboras2015symmetry, zeier2015squares}. These works study the commutant of the two-fold tensor product representations of 
compact semi-simple Lie algebras and show that for any proper sub-algebra the corresponding commutant will be larger.  Combining this with our results in the previous section, one arrives at the conclusion that  for $d\ge 3$ the uniform distribution over group $\mathcal{V}_{k}$ is not a 2-design for the group of all $\SU(d)$-invariant unitaries.  
\color{black}

Finally, we note that $\mu_{2\text{-}\mathrm{loc}}$ satisfies \cref{design} for $t=1$, i.e., it is a 1-design for $\mu_\mathrm{Haar}$. To prove this it suffices to show that for any operator $A$, 
\be\label{yr}
\mathbb{E}_{V\sim\mu_{2\text{-}\mathrm{loc}}} V A V^\dag=\mathbb{E}_{V\sim\mu_\mathrm{Haar}} V A V^\dag=\frac{1}{n!}\sum_{\sigma\in\mathbb{S}_n} \P_\sigma A \P^\dag_\sigma\ .
 \ee
Denote the above three operators by $\mathcal{E}_{2\text{-}\mathrm{loc}}(A)$, $\mathcal{E}_\mathrm{Haar}(A)$ and $\mathcal{E}_{\mathrm{perm}}(A)$, respectively, where $\mathcal{E}_{2\text{-}\mathrm{loc}}$, $\mathcal{E}_\mathrm{Haar}$ and $\mathcal{E}_{\mathrm{perm}}$ are linear super-operators. To prove \cref{yr}, first we show 
\be\label{hg}
\mathcal{E}_{2\text{-}\mathrm{loc}}=\mathcal{E}_{2\text{-}\mathrm{loc}}\circ \mathcal{E}_{\mathrm{perm}}=\mathcal{E}_{\mathrm{perm}}\ .
\ee
Recall that $\mu_{2\text{-}\mathrm{loc}}$ is the uniform distribution over $\mathcal{V}_{2}$ and 
\be
\{\P(\sigma):\sigma\in \mathbb{S}_n\}\subset \mathcal{V}_{2}\subset \mathcal{V}_{n}\ .
\ee
Therefore, $\mu_{2\text{-}\mathrm{loc}}$ remains invariant under $\P(\sigma)$ for all $\sigma\in\mathbb{S}_n$, which implies the first equality in \cref{hg}. The second equality in \cref{hg} follows from the fact that $\mathcal{E}_{\mathrm{perm}}(A)$ is a permutationally-invariant operator and by Schur-Weyl duality (see \cref{app:duality}) any such operator can be written as a linear combination of $\{U^{\otimes n} : U\in\SU(d)\}$. But, because all unitaries in $\mathcal{V}_{2}$ commute with $\{U^{\otimes n} : U\in\SU(d)\}$, this means $\forall U: \mathcal{E}_{2\text{-}\mathrm{loc}}(U^{\otimes n})=U^{\otimes n}$. This implies the second equality in \cref{hg}. Using a similar argument we can show that $\mathcal{E}_\mathrm{Haar}=\mathcal{E}_\mathrm{Haar}\circ \mathcal{E}_{\mathrm{perm}}=\mathcal{E}_{\mathrm{perm}}$. Together with \cref{hg}, this proves \cref{yr}. 

In summary, we conclude that the uniform distribution over the group $\mathcal{V}_{2}$ generated by $2$-local $\SU(d)$-invariant unitaries is a 1-design for the uniform distribution over the group $\mathcal{V}_{n}$ of all $\SU(d)$-invariant unitaries but, for $d\ge 3$, is not a 2-design.

\section{Summary and Discussion}
 
In summary, in this paper we presented four main results:
\begin{enumerate}

\item In \cref{Sec:Z2} we identified a $\mathbb{Z}_2$ symmetry of decompositions of $\SU(d)$-invariant operators in terms of permutation operators, related to the notion of parity of permutations (see \cref{assump,thn}). This property, which is automatically satisfied by 2-local $\SU(d)$-invariant unitaries, implies the existence of a new type of conservation law. In the case of $d \geq 3$, we showed that for systems with $n \leq d^2$ qudits, general $k$-local $\SU(d)$-invariant unitaries violate this conservation law for $k>2$ (see \cref{fig:4-local-fskew}). 
  
\item In \cref{fermionic group} we studied 
the representation of the permutation group $\mathbb{S}_n$ on the sites of a fermionic system. In particular, we showed that, up to a constant shift, swaps (transpositions) are represented by operators that are quadratic in creation/annihilation operators, and therefore correspond to non-interacting (free) Hamiltonians. Using this observation we fully characterized the Lie group generated by exponentials of swaps on the sites of a fermionic system. This result could be understood independent of the qudit problem studied in this paper.

\item In \cref{corr} we introduced a correspondence between the dynamics in a certain subspace of systems of qudits and a fermionic system and showed that for 2-local $\SU(d)$-invariant unitaries the corresponding fermionic system is free (non-interacting), whereas for 3-local $\SU(d)$-invariant unitaries the corresponding system is interacting. Using this observation, we found a family of functions that remain conserved under 2-local $\SU(d)$-invariant unitaries and are violated under general $\SU(d)$-invariant unitaries (see \cref{fig:l-shape}). For systems with $d\ge 3$, these conservation laws are non-trivial for arbitrarily large number of qudits $n$. 

\item Using this result, in \cref{Sec:Random} we showed that the distribution of unitaries generated by random 2-local $\SU(d)$-invariant unitaries does not converge to the Haar measure over the group of $\SU(d)$-invariant unitaries, and, in fact, for $d\ge 3$, is not even a 2-design for the Haar distribution. 
\end{enumerate}

In a follow-up paper \cite{HLM_2021}, we show that the group generated by 2-local $\SU(d)$-invariant unitaries is fully characterized by 3 types of constraints, namely (i) constraints on the relative phases between subspaces with different irreps of symmetry, which are characterized in \cref{app:relative-phases}, (ii) constraints imposed by the $\mathbb{Z}_2$ symmetry discussed in \cref{Sec:Z2}, and (iii) the constraints related to the qudit-fermion correspondence identified in \cref{corr}. We will also explain how these constraints are related to the Marin's result \cite{marin2007algebre} that fully characterizes simple factors of the Lie subalgebra generated by the transpositions.

The new conservation laws found in this paper can be useful for understanding the behaviour of complex quantum systems with $\SU(d)$ symmetries, which appear in different areas of physics. Specifically, $\SU(3)$ symmetry plays a central role in nuclear and high-energy physics. We saw that in the presence of this symmetry, even for a simple system with 6 qutrits, the additional conservation laws put interesting and non-trivial constraints on the dynamics, which are not captured by the standard (Noether) conservation laws.

A natural area of applications for these results is the study of quantum chaos and thermalization of systems with conserved charges \cite{maldacena2016bound, piroli2020random, khemani2018operator, rakovszky2018diffusive}. Specifically, it is interesting to understand how the presence of the additional conservation laws found in this paper, may slow down thermalization (See \cite{halpern2020noncommuting} for a related discussion).  Another potential area of applications is in the context of circuit complexity and its conjectured role in holography and quantum gravity \cite{susskind2016computational}. Circuit complexity is the minimum number of elementary local gates that are needed to implement a unitary transformation or to prepare a state from a fixed initial state. The results of \cite{marvian2022restrictions} and the current paper imply that if the local gates respect a symmetry, the set of realizable unitaries depend non-trivially on the locality of these elementary gates. A natural follow-up question is to study how the complexity of realizable symmetric unitaries changes with the locality of symmetric local gates, and how this notion of complexity compares with the standard notion of complexity when local gates are not restricted to be symmetric (see \cite{marvian2022restrictions} for further discussion). 

Symmetric quantum circuits have also various applications in the context of quantum computing. A desirable feature of such circuits is their resilience against noise. In particular, the symmetry guarantees that random collective rotations of qudits do not affect the performance of the quantum computer. While, in the case of qubits, universal quantum computing with exchange-only interactions are well understood, this problem is not much studied in the case of qudits and our work provides a new insight into this problem. We discuss more about this in \cite{HLM_2021} (we also note that very recently a related independent work on this topic has appeared on arXiv \cite{van2021universality}).

This work highlighted some peculiar features of qudit circuits, which do not appear in the case of qubits. While qudits can provide additional information-processing power, in general, quantum computation and communication with qudits is not as well-understood as qubits (see \cite{bartlett2002quantum,  wang2003entangling,keet2010quantum, bombin2005entanglement} and \cite{wang2020qudits} for a recent review on this topic). Given the recent experimental developments in the control of qudit systems (see, e.g., \cite{blok2021quantum}), it seems crucial to further explore these uncharted territories.

\color{black}

\begin{acknowledgments}
  IM and AH are supported by  ARL-ARO QCISS grant number 313-1049,  NSF Phy-2046195, and NSF QLCI grant OMA-2120757.  HL is supported by the U.S. Department of Energy, Office of Science, Nuclear Physics program under Award Number DE-FG02-05ER41368.
\end{acknowledgments}

\bibliography{Ref_2021_v3}

\begin{thebibliography}{75}%
\makeatletter
\providecommand \@ifxundefined [1]{%
 \@ifx{#1\undefined}
}%
\providecommand \@ifnum [1]{%
 \ifnum #1\expandafter \@firstoftwo
 \else \expandafter \@secondoftwo
 \fi
}%
\providecommand \@ifx [1]{%
 \ifx #1\expandafter \@firstoftwo
 \else \expandafter \@secondoftwo
 \fi
}%
\providecommand \natexlab [1]{#1}%
\providecommand \enquote  [1]{``#1''}%
\providecommand \bibnamefont  [1]{#1}%
\providecommand \bibfnamefont [1]{#1}%
\providecommand \citenamefont [1]{#1}%
\providecommand \href@noop [0]{\@secondoftwo}%
\providecommand \href [0]{\begingroup \@sanitize@url \@href}%
\providecommand \@href[1]{\@@startlink{#1}\@@href}%
\providecommand \@@href[1]{\endgroup#1\@@endlink}%
\providecommand \@sanitize@url [0]{\catcode `\\12\catcode `\$12\catcode
  `\&12\catcode `\#12\catcode `\^12\catcode `\_12\catcode `\%12\relax}%
\providecommand \@@startlink[1]{}%
\providecommand \@@endlink[0]{}%
\providecommand \url  [0]{\begingroup\@sanitize@url \@url }%
\providecommand \@url [1]{\endgroup\@href {#1}{\urlprefix }}%
\providecommand \urlprefix  [0]{URL }%
\providecommand \Eprint [0]{\href }%
\providecommand \doibase [0]{https://doi.org/}%
\providecommand \selectlanguage [0]{\@gobble}%
\providecommand \bibinfo  [0]{\@secondoftwo}%
\providecommand \bibfield  [0]{\@secondoftwo}%
\providecommand \translation [1]{[#1]}%
\providecommand \BibitemOpen [0]{}%
\providecommand \bibitemStop [0]{}%
\providecommand \bibitemNoStop [0]{.\EOS\space}%
\providecommand \EOS [0]{\spacefactor3000\relax}%
\providecommand \BibitemShut  [1]{\csname bibitem#1\endcsname}%
\let\auto@bib@innerbib\@empty
\bibitem [{\citenamefont {Noether}(1918)}]{noether1918nachrichten}%
  \BibitemOpen
  \bibfield  {author} {\bibinfo {author} {\bibfnamefont {E.}~\bibnamefont
  {Noether}},\ }\bibfield  {title} {\bibinfo {title} {Nachrichten der
  koniglichen gesellschaft der wissenschaften, gottingen,
  mathematisch-physikalische klasse 2, 235--257},\ }\href@noop {} {\bibfield
  {journal} {\bibinfo  {journal} {Invariante Variationsprobleme}\ } (\bibinfo
  {year} {1918})}\BibitemShut {NoStop}%
\bibitem [{\citenamefont {Noether}(1971)}]{noether1971invariant}%
  \BibitemOpen
  \bibfield  {author} {\bibinfo {author} {\bibfnamefont {E.}~\bibnamefont
  {Noether}},\ }\bibfield  {title} {\bibinfo {title} {Invariant variation
  problems},\ }\href@noop {} {\bibfield  {journal} {\bibinfo  {journal}
  {Transport Theory and Statistical Physics}\ }\textbf {\bibinfo {volume}
  {1}},\ \bibinfo {pages} {186} (\bibinfo {year} {1971})}\BibitemShut {NoStop}%
\bibitem [{\citenamefont {Marvian}(2022)}]{marvian2022restrictions}%
  \BibitemOpen
  \bibfield  {author} {\bibinfo {author} {\bibfnamefont {I.}~\bibnamefont
  {Marvian}},\ }\bibfield  {title} {\bibinfo {title} {Restrictions on
  realizable unitary operations imposed by symmetry and locality},\ }\href@noop
  {} {\bibfield  {journal} {\bibinfo  {journal} {Nature Physics}\ ,\ \bibinfo
  {pages} {1}} (\bibinfo {year} {2022})}\BibitemShut {NoStop}%
\bibitem [{\citenamefont {Lloyd}(1995)}]{lloyd1995almost}%
  \BibitemOpen
  \bibfield  {author} {\bibinfo {author} {\bibfnamefont {S.}~\bibnamefont
  {Lloyd}},\ }\bibfield  {title} {\bibinfo {title} {Almost any quantum logic
  gate is universal},\ }\href@noop {} {\bibfield  {journal} {\bibinfo
  {journal} {Physical Review Letters}\ }\textbf {\bibinfo {volume} {75}},\
  \bibinfo {pages} {346} (\bibinfo {year} {1995})}\BibitemShut {NoStop}%
\bibitem [{\citenamefont {DiVincenzo}(1995)}]{divincenzo1995two}%
  \BibitemOpen
  \bibfield  {author} {\bibinfo {author} {\bibfnamefont {D.~P.}\ \bibnamefont
  {DiVincenzo}},\ }\bibfield  {title} {\bibinfo {title} {Two-bit gates are
  universal for quantum computation},\ }\href@noop {} {\bibfield  {journal}
  {\bibinfo  {journal} {Physical Review A}\ }\textbf {\bibinfo {volume} {51}},\
  \bibinfo {pages} {1015} (\bibinfo {year} {1995})}\BibitemShut {NoStop}%
\bibitem [{\citenamefont {Deutsch}\ \emph {et~al.}(1995)\citenamefont
  {Deutsch}, \citenamefont {Barenco},\ and\ \citenamefont
  {Ekert}}]{deutsch1995universality}%
  \BibitemOpen
  \bibfield  {author} {\bibinfo {author} {\bibfnamefont {D.~E.}\ \bibnamefont
  {Deutsch}}, \bibinfo {author} {\bibfnamefont {A.}~\bibnamefont {Barenco}},\
  and\ \bibinfo {author} {\bibfnamefont {A.}~\bibnamefont {Ekert}},\ }\bibfield
   {title} {\bibinfo {title} {Universality in quantum computation},\
  }\href@noop {} {\bibfield  {journal} {\bibinfo  {journal} {Proceedings of the
  Royal Society of London. Series A: Mathematical and Physical Sciences}\
  }\textbf {\bibinfo {volume} {449}},\ \bibinfo {pages} {669} (\bibinfo {year}
  {1995})}\BibitemShut {NoStop}%
\bibitem [{\citenamefont {Brylinski}\ and\ \citenamefont
  {Brylinski}(2002)}]{brylinski2002universal}%
  \BibitemOpen
  \bibfield  {author} {\bibinfo {author} {\bibfnamefont {J.-L.}\ \bibnamefont
  {Brylinski}}\ and\ \bibinfo {author} {\bibfnamefont {R.}~\bibnamefont
  {Brylinski}},\ }\bibfield  {title} {\bibinfo {title} {Universal quantum
  gates},\ }\href@noop {} {\bibfield  {journal} {\bibinfo  {journal}
  {Mathematics of quantum computation}\ }\textbf {\bibinfo {volume} {79}}
  (\bibinfo {year} {2002})}\BibitemShut {NoStop}%
\bibitem [{\citenamefont {Chen}\ \emph {et~al.}(2011)\citenamefont {Chen},
  \citenamefont {Gu},\ and\ \citenamefont {Wen}}]{chen2011classification}%
  \BibitemOpen
  \bibfield  {author} {\bibinfo {author} {\bibfnamefont {X.}~\bibnamefont
  {Chen}}, \bibinfo {author} {\bibfnamefont {Z.-C.}\ \bibnamefont {Gu}},\ and\
  \bibinfo {author} {\bibfnamefont {X.-G.}\ \bibnamefont {Wen}},\ }\bibfield
  {title} {\bibinfo {title} {Classification of gapped symmetric phases in
  one-dimensional spin systems},\ }\href@noop {} {\bibfield  {journal}
  {\bibinfo  {journal} {Physical review b}\ }\textbf {\bibinfo {volume} {83}},\
  \bibinfo {pages} {035107} (\bibinfo {year} {2011})}\BibitemShut {NoStop}%
\bibitem [{\citenamefont {Horodecki}\ and\ \citenamefont
  {Oppenheim}(2013)}]{FundLimitsNature}%
  \BibitemOpen
  \bibfield  {author} {\bibinfo {author} {\bibfnamefont {M.}~\bibnamefont
  {Horodecki}}\ and\ \bibinfo {author} {\bibfnamefont {J.}~\bibnamefont
  {Oppenheim}},\ }\bibfield  {title} {\bibinfo {title} {{Fundamental
  limitations for quantum and nanoscale thermodynamics}},\ }\href@noop {}
  {\bibfield  {journal} {\bibinfo  {journal} {Nat. Commun.}\ }\textbf {\bibinfo
  {volume} {4}},\ \bibinfo {pages} {1} (\bibinfo {year} {2013})}\BibitemShut
  {NoStop}%
\bibitem [{\citenamefont {Brandao}\ \emph {et~al.}(2013)\citenamefont
  {Brandao}, \citenamefont {Horodecki}, \citenamefont {Oppenheim},
  \citenamefont {Renes},\ and\ \citenamefont {Spekkens}}]{brandao2013resource}%
  \BibitemOpen
  \bibfield  {author} {\bibinfo {author} {\bibfnamefont {F.~G.}\ \bibnamefont
  {Brandao}}, \bibinfo {author} {\bibfnamefont {M.}~\bibnamefont {Horodecki}},
  \bibinfo {author} {\bibfnamefont {J.}~\bibnamefont {Oppenheim}}, \bibinfo
  {author} {\bibfnamefont {J.~M.}\ \bibnamefont {Renes}},\ and\ \bibinfo
  {author} {\bibfnamefont {R.~W.}\ \bibnamefont {Spekkens}},\ }\bibfield
  {title} {\bibinfo {title} {Resource theory of quantum states out of thermal
  equilibrium},\ }\href@noop {} {\bibfield  {journal} {\bibinfo  {journal}
  {Physical review letters}\ }\textbf {\bibinfo {volume} {111}},\ \bibinfo
  {pages} {250404} (\bibinfo {year} {2013})}\BibitemShut {NoStop}%
\bibitem [{\citenamefont {Janzing}\ \emph {et~al.}(2000)\citenamefont
  {Janzing}, \citenamefont {Wocjan}, \citenamefont {Zeier}, \citenamefont
  {Geiss},\ and\ \citenamefont {Beth}}]{janzing2000thermodynamic}%
  \BibitemOpen
  \bibfield  {author} {\bibinfo {author} {\bibfnamefont {D.}~\bibnamefont
  {Janzing}}, \bibinfo {author} {\bibfnamefont {P.}~\bibnamefont {Wocjan}},
  \bibinfo {author} {\bibfnamefont {R.}~\bibnamefont {Zeier}}, \bibinfo
  {author} {\bibfnamefont {R.}~\bibnamefont {Geiss}},\ and\ \bibinfo {author}
  {\bibfnamefont {T.}~\bibnamefont {Beth}},\ }\bibfield  {title} {\bibinfo
  {title} {{Thermodynamic cost of reliability and low temperatures: tightening
  Landauer's principle and the Second Law}},\ }\href@noop {} {\bibfield
  {journal} {\bibinfo  {journal} {Int. J. Theor. Phys.}\ }\textbf {\bibinfo
  {volume} {39}},\ \bibinfo {pages} {2717} (\bibinfo {year}
  {2000})}\BibitemShut {NoStop}%
\bibitem [{\citenamefont {Guryanova}\ \emph {et~al.}(2016)\citenamefont
  {Guryanova}, \citenamefont {Popescu}, \citenamefont {Short}, \citenamefont
  {Silva},\ and\ \citenamefont {Skrzypczyk}}]{guryanova2016thermodynamics}%
  \BibitemOpen
  \bibfield  {author} {\bibinfo {author} {\bibfnamefont {Y.}~\bibnamefont
  {Guryanova}}, \bibinfo {author} {\bibfnamefont {S.}~\bibnamefont {Popescu}},
  \bibinfo {author} {\bibfnamefont {A.~J.}\ \bibnamefont {Short}}, \bibinfo
  {author} {\bibfnamefont {R.}~\bibnamefont {Silva}},\ and\ \bibinfo {author}
  {\bibfnamefont {P.}~\bibnamefont {Skrzypczyk}},\ }\bibfield  {title}
  {\bibinfo {title} {Thermodynamics of quantum systems with multiple conserved
  quantities},\ }\href@noop {} {\bibfield  {journal} {\bibinfo  {journal}
  {Nature communications}\ }\textbf {\bibinfo {volume} {7}},\ \bibinfo {pages}
  {ncomms12049} (\bibinfo {year} {2016})}\BibitemShut {NoStop}%
\bibitem [{\citenamefont {Lostaglio}\ \emph
  {et~al.}(2015{\natexlab{a}})\citenamefont {Lostaglio}, \citenamefont
  {Korzekwa}, \citenamefont {Jennings},\ and\ \citenamefont
  {Rudolph}}]{lostaglio2015quantumPRX}%
  \BibitemOpen
  \bibfield  {author} {\bibinfo {author} {\bibfnamefont {M.}~\bibnamefont
  {Lostaglio}}, \bibinfo {author} {\bibfnamefont {K.}~\bibnamefont {Korzekwa}},
  \bibinfo {author} {\bibfnamefont {D.}~\bibnamefont {Jennings}},\ and\
  \bibinfo {author} {\bibfnamefont {T.}~\bibnamefont {Rudolph}},\ }\bibfield
  {title} {\bibinfo {title} {Quantum coherence, time-translation symmetry, and
  thermodynamics},\ }\href@noop {} {\bibfield  {journal} {\bibinfo  {journal}
  {Physical Review X}\ }\textbf {\bibinfo {volume} {5}},\ \bibinfo {pages}
  {021001} (\bibinfo {year} {2015}{\natexlab{a}})}\BibitemShut {NoStop}%
\bibitem [{\citenamefont {Lostaglio}\ \emph {et~al.}(2017)\citenamefont
  {Lostaglio}, \citenamefont {Jennings},\ and\ \citenamefont
  {Rudolph}}]{lostaglio2017thermodynamic}%
  \BibitemOpen
  \bibfield  {author} {\bibinfo {author} {\bibfnamefont {M.}~\bibnamefont
  {Lostaglio}}, \bibinfo {author} {\bibfnamefont {D.}~\bibnamefont
  {Jennings}},\ and\ \bibinfo {author} {\bibfnamefont {T.}~\bibnamefont
  {Rudolph}},\ }\bibfield  {title} {\bibinfo {title} {Thermodynamic resource
  theories, non-commutativity and maximum entropy principles},\ }\href@noop {}
  {\bibfield  {journal} {\bibinfo  {journal} {New Journal of Physics}\ }\textbf
  {\bibinfo {volume} {19}},\ \bibinfo {pages} {043008} (\bibinfo {year}
  {2017})}\BibitemShut {NoStop}%
\bibitem [{\citenamefont {Lostaglio}\ \emph
  {et~al.}(2015{\natexlab{b}})\citenamefont {Lostaglio}, \citenamefont
  {Jennings},\ and\ \citenamefont {Rudolph}}]{lostaglio2015description}%
  \BibitemOpen
  \bibfield  {author} {\bibinfo {author} {\bibfnamefont {M.}~\bibnamefont
  {Lostaglio}}, \bibinfo {author} {\bibfnamefont {D.}~\bibnamefont
  {Jennings}},\ and\ \bibinfo {author} {\bibfnamefont {T.}~\bibnamefont
  {Rudolph}},\ }\bibfield  {title} {\bibinfo {title} {Description of quantum
  coherence in thermodynamic processes requires constraints beyond free
  energy},\ }\href@noop {} {\bibfield  {journal} {\bibinfo  {journal} {Nature
  communications}\ }\textbf {\bibinfo {volume} {6}} (\bibinfo {year}
  {2015}{\natexlab{b}})}\BibitemShut {NoStop}%
\bibitem [{\citenamefont {Halpern}\ \emph {et~al.}(2016)\citenamefont
  {Halpern}, \citenamefont {Faist}, \citenamefont {Oppenheim},\ and\
  \citenamefont {Winter}}]{halpern2016microcanonical}%
  \BibitemOpen
  \bibfield  {author} {\bibinfo {author} {\bibfnamefont {N.~Y.}\ \bibnamefont
  {Halpern}}, \bibinfo {author} {\bibfnamefont {P.}~\bibnamefont {Faist}},
  \bibinfo {author} {\bibfnamefont {J.}~\bibnamefont {Oppenheim}},\ and\
  \bibinfo {author} {\bibfnamefont {A.}~\bibnamefont {Winter}},\ }\bibfield
  {title} {\bibinfo {title} {Microcanonical and resource-theoretic derivations
  of the thermal state of a quantum system with noncommuting charges},\
  }\href@noop {} {\bibfield  {journal} {\bibinfo  {journal} {Nature
  communications}\ }\textbf {\bibinfo {volume} {7}},\ \bibinfo {pages} {12051}
  (\bibinfo {year} {2016})}\BibitemShut {NoStop}%
\bibitem [{\citenamefont {Halpern}\ and\ \citenamefont
  {Renes}(2016)}]{halpern2016beyond}%
  \BibitemOpen
  \bibfield  {author} {\bibinfo {author} {\bibfnamefont {N.~Y.}\ \bibnamefont
  {Halpern}}\ and\ \bibinfo {author} {\bibfnamefont {J.~M.}\ \bibnamefont
  {Renes}},\ }\bibfield  {title} {\bibinfo {title} {Beyond heat baths:
  Generalized resource theories for small-scale thermodynamics},\ }\href@noop
  {} {\bibfield  {journal} {\bibinfo  {journal} {Physical Review E}\ }\textbf
  {\bibinfo {volume} {93}},\ \bibinfo {pages} {022126} (\bibinfo {year}
  {2016})}\BibitemShut {NoStop}%
\bibitem [{\citenamefont {Narasimhachar}\ and\ \citenamefont
  {Gour}(2015)}]{narasimhachar2015low}%
  \BibitemOpen
  \bibfield  {author} {\bibinfo {author} {\bibfnamefont {V.}~\bibnamefont
  {Narasimhachar}}\ and\ \bibinfo {author} {\bibfnamefont {G.}~\bibnamefont
  {Gour}},\ }\bibfield  {title} {\bibinfo {title} {Low-temperature
  thermodynamics with quantum coherence},\ }\href@noop {} {\bibfield  {journal}
  {\bibinfo  {journal} {Nature communications}\ }\textbf {\bibinfo {volume}
  {6}},\ \bibinfo {pages} {7689} (\bibinfo {year} {2015})}\BibitemShut
  {NoStop}%
\bibitem [{\citenamefont {Bartlett}\ \emph {et~al.}(2007)\citenamefont
  {Bartlett}, \citenamefont {Rudolph},\ and\ \citenamefont
  {Spekkens}}]{QRF_BRS_07}%
  \BibitemOpen
  \bibfield  {author} {\bibinfo {author} {\bibfnamefont {S.~D.}\ \bibnamefont
  {Bartlett}}, \bibinfo {author} {\bibfnamefont {T.}~\bibnamefont {Rudolph}},\
  and\ \bibinfo {author} {\bibfnamefont {R.~W.}\ \bibnamefont {Spekkens}},\
  }\bibfield  {title} {\bibinfo {title} {Reference frames, superselection
  rules, and quantum information},\ }\href@noop {} {\bibfield  {journal}
  {\bibinfo  {journal} {Reviews of Modern Physics}\ }\textbf {\bibinfo {volume}
  {79}},\ \bibinfo {pages} {555} (\bibinfo {year} {2007})}\BibitemShut
  {NoStop}%
\bibitem [{\citenamefont {Gour}\ and\ \citenamefont
  {Spekkens}(2008)}]{gour2008resource}%
  \BibitemOpen
  \bibfield  {author} {\bibinfo {author} {\bibfnamefont {G.}~\bibnamefont
  {Gour}}\ and\ \bibinfo {author} {\bibfnamefont {R.~W.}\ \bibnamefont
  {Spekkens}},\ }\bibfield  {title} {\bibinfo {title} {The resource theory of
  quantum reference frames: manipulations and monotones},\ }\href@noop {}
  {\bibfield  {journal} {\bibinfo  {journal} {New Journal of Physics}\ }\textbf
  {\bibinfo {volume} {10}},\ \bibinfo {pages} {033023} (\bibinfo {year}
  {2008})}\BibitemShut {NoStop}%
\bibitem [{\citenamefont {Marvian}\ and\ \citenamefont
  {Spekkens}(2013)}]{marvian2013theory}%
  \BibitemOpen
  \bibfield  {author} {\bibinfo {author} {\bibfnamefont {I.}~\bibnamefont
  {Marvian}}\ and\ \bibinfo {author} {\bibfnamefont {R.~W.}\ \bibnamefont
  {Spekkens}},\ }\bibfield  {title} {\bibinfo {title} {The theory of
  manipulations of pure state asymmetry: I. basic tools, equivalence classes
  and single copy transformations},\ }\href@noop {} {\bibfield  {journal}
  {\bibinfo  {journal} {New Journal of Physics}\ }\textbf {\bibinfo {volume}
  {15}},\ \bibinfo {pages} {033001} (\bibinfo {year} {2013})}\BibitemShut
  {NoStop}%
\bibitem [{\citenamefont {Marvian}\ and\ \citenamefont
  {Mann}(2008)}]{marvian2008building}%
  \BibitemOpen
  \bibfield  {author} {\bibinfo {author} {\bibfnamefont {I.}~\bibnamefont
  {Marvian}}\ and\ \bibinfo {author} {\bibfnamefont {R.}~\bibnamefont {Mann}},\
  }\bibfield  {title} {\bibinfo {title} {Building all time evolutions with
  rotationally invariant hamiltonians},\ }\href@noop {} {\bibfield  {journal}
  {\bibinfo  {journal} {Physical Review A}\ }\textbf {\bibinfo {volume} {78}},\
  \bibinfo {pages} {022304} (\bibinfo {year} {2008})}\BibitemShut {NoStop}%
\bibitem [{\citenamefont {Maldacena}\ \emph {et~al.}(2016)\citenamefont
  {Maldacena}, \citenamefont {Shenker},\ and\ \citenamefont
  {Stanford}}]{maldacena2016bound}%
  \BibitemOpen
  \bibfield  {author} {\bibinfo {author} {\bibfnamefont {J.}~\bibnamefont
  {Maldacena}}, \bibinfo {author} {\bibfnamefont {S.~H.}\ \bibnamefont
  {Shenker}},\ and\ \bibinfo {author} {\bibfnamefont {D.}~\bibnamefont
  {Stanford}},\ }\bibfield  {title} {\bibinfo {title} {A bound on chaos},\
  }\href@noop {} {\bibfield  {journal} {\bibinfo  {journal} {Journal of High
  Energy Physics}\ }\textbf {\bibinfo {volume} {2016}},\ \bibinfo {pages} {106}
  (\bibinfo {year} {2016})}\BibitemShut {NoStop}%
\bibitem [{\citenamefont {Piroli}\ \emph {et~al.}(2020)\citenamefont {Piroli},
  \citenamefont {S{\"u}nderhauf},\ and\ \citenamefont {Qi}}]{piroli2020random}%
  \BibitemOpen
  \bibfield  {author} {\bibinfo {author} {\bibfnamefont {L.}~\bibnamefont
  {Piroli}}, \bibinfo {author} {\bibfnamefont {C.}~\bibnamefont
  {S{\"u}nderhauf}},\ and\ \bibinfo {author} {\bibfnamefont {X.-L.}\
  \bibnamefont {Qi}},\ }\bibfield  {title} {\bibinfo {title} {A random unitary
  circuit model for black hole evaporation},\ }\href@noop {} {\bibfield
  {journal} {\bibinfo  {journal} {Journal of High Energy Physics: JHEP}\ }
  (\bibinfo {year} {2020})}\BibitemShut {NoStop}%
\bibitem [{\citenamefont {Susskind}(2016)}]{susskind2016computational}%
  \BibitemOpen
  \bibfield  {author} {\bibinfo {author} {\bibfnamefont {L.}~\bibnamefont
  {Susskind}},\ }\bibfield  {title} {\bibinfo {title} {Computational complexity
  and black hole horizons},\ }\href@noop {} {\bibfield  {journal} {\bibinfo
  {journal} {Fortschritte der Physik}\ }\textbf {\bibinfo {volume} {64}},\
  \bibinfo {pages} {24} (\bibinfo {year} {2016})}\BibitemShut {NoStop}%
\bibitem [{\citenamefont {Khemani}\ \emph {et~al.}(2018)\citenamefont
  {Khemani}, \citenamefont {Vishwanath},\ and\ \citenamefont
  {Huse}}]{khemani2018operator}%
  \BibitemOpen
  \bibfield  {author} {\bibinfo {author} {\bibfnamefont {V.}~\bibnamefont
  {Khemani}}, \bibinfo {author} {\bibfnamefont {A.}~\bibnamefont
  {Vishwanath}},\ and\ \bibinfo {author} {\bibfnamefont {D.~A.}\ \bibnamefont
  {Huse}},\ }\bibfield  {title} {\bibinfo {title} {Operator spreading and the
  emergence of dissipative hydrodynamics under unitary evolution with
  conservation laws},\ }\href@noop {} {\bibfield  {journal} {\bibinfo
  {journal} {Physical Review X}\ }\textbf {\bibinfo {volume} {8}},\ \bibinfo
  {pages} {031057} (\bibinfo {year} {2018})}\BibitemShut {NoStop}%
\bibitem [{\citenamefont {Rakovszky}\ \emph {et~al.}(2018)\citenamefont
  {Rakovszky}, \citenamefont {Pollmann},\ and\ \citenamefont {von
  Keyserlingk}}]{rakovszky2018diffusive}%
  \BibitemOpen
  \bibfield  {author} {\bibinfo {author} {\bibfnamefont {T.}~\bibnamefont
  {Rakovszky}}, \bibinfo {author} {\bibfnamefont {F.}~\bibnamefont
  {Pollmann}},\ and\ \bibinfo {author} {\bibfnamefont {C.}~\bibnamefont {von
  Keyserlingk}},\ }\bibfield  {title} {\bibinfo {title} {Diffusive
  hydrodynamics of out-of-time-ordered correlators with charge conservation},\
  }\href@noop {} {\bibfield  {journal} {\bibinfo  {journal} {Physical Review
  X}\ }\textbf {\bibinfo {volume} {8}},\ \bibinfo {pages} {031058} (\bibinfo
  {year} {2018})}\BibitemShut {NoStop}%
\bibitem [{\citenamefont {Hulse}\ \emph {et~al.}(2022)\citenamefont {Hulse},
  \citenamefont {Liu},\ and\ \citenamefont {Marvian}}]{HLM_2021}%
  \BibitemOpen
  \bibfield  {author} {\bibinfo {author} {\bibfnamefont {A.}~\bibnamefont
  {Hulse}}, \bibinfo {author} {\bibfnamefont {H.}~\bibnamefont {Liu}},\ and\
  \bibinfo {author} {\bibfnamefont {I.}~\bibnamefont {Marvian}},\ }\bibfield
  {title} {\bibinfo {title} {Qudit circuitswith su(d) symmetry (under
  preparation)},\ }\href@noop {} {\  (\bibinfo {year} {2022})}\BibitemShut
  {NoStop}%
\bibitem [{\citenamefont {Marin}(2007)}]{marin2007algebre}%
  \BibitemOpen
  \bibfield  {author} {\bibinfo {author} {\bibfnamefont {I.}~\bibnamefont
  {Marin}},\ }\bibfield  {title} {\bibinfo {title} {L'alg{\`e}bre de lie des
  transpositions},\ }\href@noop {} {\bibfield  {journal} {\bibinfo  {journal}
  {Journal of Algebra}\ }\textbf {\bibinfo {volume} {310}},\ \bibinfo {pages}
  {742} (\bibinfo {year} {2007})}\BibitemShut {NoStop}%
\bibitem [{\citenamefont {DiVincenzo}\ \emph {et~al.}(2000)\citenamefont
  {DiVincenzo}, \citenamefont {Bacon}, \citenamefont {Kempe}, \citenamefont
  {Burkard},\ and\ \citenamefont {Whaley}}]{DiVincenzo:2000kx}%
  \BibitemOpen
  \bibfield  {author} {\bibinfo {author} {\bibfnamefont {D.~P.}\ \bibnamefont
  {DiVincenzo}}, \bibinfo {author} {\bibfnamefont {D.}~\bibnamefont {Bacon}},
  \bibinfo {author} {\bibfnamefont {J.}~\bibnamefont {Kempe}}, \bibinfo
  {author} {\bibfnamefont {G.}~\bibnamefont {Burkard}},\ and\ \bibinfo {author}
  {\bibfnamefont {K.~B.}\ \bibnamefont {Whaley}},\ }\bibfield  {title}
  {\bibinfo {title} {Universal quantum computation with the exchange
  interaction},\ }\href {https://doi.org/10.1038/35042541} {\bibfield
  {journal} {\bibinfo  {journal} {Nature}\ }\textbf {\bibinfo {volume} {408}},\
  \bibinfo {pages} {339} (\bibinfo {year} {2000})}\BibitemShut {NoStop}%
\bibitem [{\citenamefont {Levy}(2002)}]{levy2002universal}%
  \BibitemOpen
  \bibfield  {author} {\bibinfo {author} {\bibfnamefont {J.}~\bibnamefont
  {Levy}},\ }\bibfield  {title} {\bibinfo {title} {Universal quantum
  computation with spin-1/2 pairs and heisenberg exchange},\ }\href@noop {}
  {\bibfield  {journal} {\bibinfo  {journal} {Physical Review Letters}\
  }\textbf {\bibinfo {volume} {89}},\ \bibinfo {pages} {147902} (\bibinfo
  {year} {2002})}\BibitemShut {NoStop}%
\bibitem [{\citenamefont {Kempe}\ \emph {et~al.}(2001)\citenamefont {Kempe},
  \citenamefont {Bacon}, \citenamefont {Lidar},\ and\ \citenamefont
  {Whaley}}]{kempe2001theory}%
  \BibitemOpen
  \bibfield  {author} {\bibinfo {author} {\bibfnamefont {J.}~\bibnamefont
  {Kempe}}, \bibinfo {author} {\bibfnamefont {D.}~\bibnamefont {Bacon}},
  \bibinfo {author} {\bibfnamefont {D.~A.}\ \bibnamefont {Lidar}},\ and\
  \bibinfo {author} {\bibfnamefont {K.~B.}\ \bibnamefont {Whaley}},\ }\bibfield
   {title} {\bibinfo {title} {Theory of decoherence-free fault-tolerant
  universal quantum computation},\ }\href@noop {} {\bibfield  {journal}
  {\bibinfo  {journal} {Physical Review A}\ }\textbf {\bibinfo {volume} {63}},\
  \bibinfo {pages} {042307} (\bibinfo {year} {2001})}\BibitemShut {NoStop}%
\bibitem [{\citenamefont {Bacon}\ \emph {et~al.}(2001)\citenamefont {Bacon},
  \citenamefont {Brown},\ and\ \citenamefont {Whaley}}]{bacon2001coherence}%
  \BibitemOpen
  \bibfield  {author} {\bibinfo {author} {\bibfnamefont {D.}~\bibnamefont
  {Bacon}}, \bibinfo {author} {\bibfnamefont {K.~R.}\ \bibnamefont {Brown}},\
  and\ \bibinfo {author} {\bibfnamefont {K.~B.}\ \bibnamefont {Whaley}},\
  }\bibfield  {title} {\bibinfo {title} {Coherence-preserving quantum bits},\
  }\href@noop {} {\bibfield  {journal} {\bibinfo  {journal} {Physical Review
  Letters}\ }\textbf {\bibinfo {volume} {87}},\ \bibinfo {pages} {247902}
  (\bibinfo {year} {2001})}\BibitemShut {NoStop}%
\bibitem [{\citenamefont {Rudolph}\ and\ \citenamefont
  {Virmani}(2005)}]{rudolph2005relational}%
  \BibitemOpen
  \bibfield  {author} {\bibinfo {author} {\bibfnamefont {T.}~\bibnamefont
  {Rudolph}}\ and\ \bibinfo {author} {\bibfnamefont {S.~S.}\ \bibnamefont
  {Virmani}},\ }\bibfield  {title} {\bibinfo {title} {A relational quantum
  computer using only two-qubit total spin measurement and an initial supply of
  highly mixed single-qubit states},\ }\href@noop {} {\bibfield  {journal}
  {\bibinfo  {journal} {New Journal of Physics}\ }\textbf {\bibinfo {volume}
  {7}},\ \bibinfo {pages} {228} (\bibinfo {year} {2005})}\BibitemShut {NoStop}%
\bibitem [{\citenamefont {van Meter}(2021)}]{van2021universality}%
  \BibitemOpen
  \bibfield  {author} {\bibinfo {author} {\bibfnamefont {J.~R.}\ \bibnamefont
  {van Meter}},\ }\bibfield  {title} {\bibinfo {title} {Universality of swap
  for qudits: a representation theory approach},\ }\href@noop {} {\bibfield
  {journal} {\bibinfo  {journal} {arXiv preprint arXiv:2103.12303}\ } (\bibinfo
  {year} {2021})}\BibitemShut {NoStop}%
\bibitem [{\citenamefont {Zanardi}(2001)}]{zanardi2001virtual}%
  \BibitemOpen
  \bibfield  {author} {\bibinfo {author} {\bibfnamefont {P.}~\bibnamefont
  {Zanardi}},\ }\bibfield  {title} {\bibinfo {title} {Virtual quantum
  subsystems},\ }\href@noop {} {\bibfield  {journal} {\bibinfo  {journal}
  {Physical Review Letters}\ }\textbf {\bibinfo {volume} {87}},\ \bibinfo
  {pages} {077901} (\bibinfo {year} {2001})}\BibitemShut {NoStop}%
\bibitem [{\citenamefont {Marvian}\ and\ \citenamefont
  {Spekkens}(2014{\natexlab{a}})}]{marvian2014extending}%
  \BibitemOpen
  \bibfield  {author} {\bibinfo {author} {\bibfnamefont {I.}~\bibnamefont
  {Marvian}}\ and\ \bibinfo {author} {\bibfnamefont {R.~W.}\ \bibnamefont
  {Spekkens}},\ }\bibfield  {title} {\bibinfo {title} {Extending noether's
  theorem by quantifying the asymmetry of quantum states},\ }\href@noop {}
  {\bibfield  {journal} {\bibinfo  {journal} {Nature communications}\ }\textbf
  {\bibinfo {volume} {5}},\ \bibinfo {pages} {3821} (\bibinfo {year}
  {2014}{\natexlab{a}})}\BibitemShut {NoStop}%
\bibitem [{\citenamefont {Marvian}\ and\ \citenamefont
  {Spekkens}(2014{\natexlab{b}})}]{marvian2014asymmetry}%
  \BibitemOpen
  \bibfield  {author} {\bibinfo {author} {\bibfnamefont {I.}~\bibnamefont
  {Marvian}}\ and\ \bibinfo {author} {\bibfnamefont {R.~W.}\ \bibnamefont
  {Spekkens}},\ }\bibfield  {title} {\bibinfo {title} {Asymmetry properties of
  pure quantum states},\ }\href@noop {} {\bibfield  {journal} {\bibinfo
  {journal} {Physical Review A}\ }\textbf {\bibinfo {volume} {90}},\ \bibinfo
  {pages} {014102} (\bibinfo {year} {2014}{\natexlab{b}})}\BibitemShut
  {NoStop}%
\bibitem [{\citenamefont {Harrow}(2005)}]{harrow2005applications}%
  \BibitemOpen
  \bibfield  {author} {\bibinfo {author} {\bibfnamefont {A.~W.}\ \bibnamefont
  {Harrow}},\ }\bibfield  {title} {\bibinfo {title} {Applications of coherent
  classical communication and the schur transform to quantum information
  theory},\ }\href@noop {} {\bibfield  {journal} {\bibinfo  {journal} {arXiv
  preprint quant-ph/0512255}\ } (\bibinfo {year} {2005})}\BibitemShut {NoStop}%
\bibitem [{\citenamefont {Goodman}\ and\ \citenamefont
  {Wallach}(2009)}]{goodman2009symmetry}%
  \BibitemOpen
  \bibfield  {author} {\bibinfo {author} {\bibfnamefont {R.}~\bibnamefont
  {Goodman}}\ and\ \bibinfo {author} {\bibfnamefont {N.~R.}\ \bibnamefont
  {Wallach}},\ }\href@noop {} {\emph {\bibinfo {title} {Symmetry,
  representations, and invariants}}},\ Vol.\ \bibinfo {volume} {255}\ (\bibinfo
   {publisher} {Springer},\ \bibinfo {year} {2009})\BibitemShut {NoStop}%
\bibitem [{\citenamefont {D'Alessandro}(2002)}]{d2002uniform}%
  \BibitemOpen
  \bibfield  {author} {\bibinfo {author} {\bibfnamefont {D.}~\bibnamefont
  {D'Alessandro}},\ }\bibfield  {title} {\bibinfo {title} {Uniform finite
  generation of compact lie groups},\ }\href@noop {} {\bibfield  {journal}
  {\bibinfo  {journal} {Systems \& control letters}\ }\textbf {\bibinfo
  {volume} {47}},\ \bibinfo {pages} {87} (\bibinfo {year} {2002})}\BibitemShut
  {NoStop}%
\bibitem [{\citenamefont {Marvian}\ \emph {et~al.}(2022)\citenamefont
  {Marvian}, \citenamefont {Liu},\ and\ \citenamefont {Hulse}}]{MLH_2022}%
  \BibitemOpen
  \bibfield  {author} {\bibinfo {author} {\bibfnamefont {I.}~\bibnamefont
  {Marvian}}, \bibinfo {author} {\bibfnamefont {H.}~\bibnamefont {Liu}},\ and\
  \bibinfo {author} {\bibfnamefont {A.}~\bibnamefont {Hulse}},\ }\bibfield
  {title} {\bibinfo {title} {Under preparation},\ }\href@noop {} {\  (\bibinfo
  {year} {2022})}\BibitemShut {NoStop}%
\bibitem [{\citenamefont {Liu}\ \emph {et~al.}(2022)\citenamefont {Liu},
  \citenamefont {Hulse},\ and\ \citenamefont {Marvian}}]{LHM_2022}%
  \BibitemOpen
  \bibfield  {author} {\bibinfo {author} {\bibfnamefont {H.}~\bibnamefont
  {Liu}}, \bibinfo {author} {\bibfnamefont {A.}~\bibnamefont {Hulse}},\ and\
  \bibinfo {author} {\bibfnamefont {I.}~\bibnamefont {Marvian}},\ }\bibfield
  {title} {\bibinfo {title} {Simple criteria for the universality of symmetric
  unitary transformations (under preparation)},\ }\href@noop {} {\  (\bibinfo
  {year} {2022})}\BibitemShut {NoStop}%
\bibitem [{\citenamefont {Marin}(2010)}]{marin2010group}%
  \BibitemOpen
  \bibfield  {author} {\bibinfo {author} {\bibfnamefont {I.}~\bibnamefont
  {Marin}},\ }\bibfield  {title} {\bibinfo {title} {Group algebras of finite
  groups as lie algebras},\ }\href@noop {} {\bibfield  {journal} {\bibinfo
  {journal} {Communications in Algebra{\textregistered}}\ }\textbf {\bibinfo
  {volume} {38}},\ \bibinfo {pages} {2572} (\bibinfo {year}
  {2010})}\BibitemShut {NoStop}%
\bibitem [{\citenamefont {Fulton}\ and\ \citenamefont
  {Harris}(2013)}]{fulton2013representation}%
  \BibitemOpen
  \bibfield  {author} {\bibinfo {author} {\bibfnamefont {W.}~\bibnamefont
  {Fulton}}\ and\ \bibinfo {author} {\bibfnamefont {J.}~\bibnamefont
  {Harris}},\ }\href@noop {} {\emph {\bibinfo {title} {Representation theory: a
  first course}}},\ Vol.\ \bibinfo {volume} {129}\ (\bibinfo  {publisher}
  {Springer Science \& Business Media},\ \bibinfo {year} {2013})\BibitemShut
  {NoStop}%
\bibitem [{\citenamefont {Emerson}\ \emph {et~al.}(2005)\citenamefont
  {Emerson}, \citenamefont {Livine},\ and\ \citenamefont
  {Lloyd}}]{emerson2005convergence}%
  \BibitemOpen
  \bibfield  {author} {\bibinfo {author} {\bibfnamefont {J.}~\bibnamefont
  {Emerson}}, \bibinfo {author} {\bibfnamefont {E.}~\bibnamefont {Livine}},\
  and\ \bibinfo {author} {\bibfnamefont {S.}~\bibnamefont {Lloyd}},\ }\bibfield
   {title} {\bibinfo {title} {Convergence conditions for random quantum
  circuits},\ }\href@noop {} {\bibfield  {journal} {\bibinfo  {journal}
  {Physical Review A}\ }\textbf {\bibinfo {volume} {72}},\ \bibinfo {pages}
  {060302} (\bibinfo {year} {2005})}\BibitemShut {NoStop}%
\bibitem [{\citenamefont {Harrow}\ and\ \citenamefont
  {Low}(2009)}]{harrow2009random}%
  \BibitemOpen
  \bibfield  {author} {\bibinfo {author} {\bibfnamefont {A.~W.}\ \bibnamefont
  {Harrow}}\ and\ \bibinfo {author} {\bibfnamefont {R.~A.}\ \bibnamefont
  {Low}},\ }\bibfield  {title} {\bibinfo {title} {Random quantum circuits are
  approximate 2-designs},\ }\href@noop {} {\bibfield  {journal} {\bibinfo
  {journal} {Communications in Mathematical Physics}\ }\textbf {\bibinfo
  {volume} {291}},\ \bibinfo {pages} {257} (\bibinfo {year}
  {2009})}\BibitemShut {NoStop}%
\bibitem [{\citenamefont {Brandao}\ \emph
  {et~al.}(2016{\natexlab{a}})\citenamefont {Brandao}, \citenamefont {Harrow},\
  and\ \citenamefont {Horodecki}}]{brandao2016efficient}%
  \BibitemOpen
  \bibfield  {author} {\bibinfo {author} {\bibfnamefont {F.~G.}\ \bibnamefont
  {Brandao}}, \bibinfo {author} {\bibfnamefont {A.~W.}\ \bibnamefont
  {Harrow}},\ and\ \bibinfo {author} {\bibfnamefont {M.}~\bibnamefont
  {Horodecki}},\ }\bibfield  {title} {\bibinfo {title} {Efficient quantum
  pseudorandomness},\ }\href@noop {} {\bibfield  {journal} {\bibinfo  {journal}
  {Physical review letters}\ }\textbf {\bibinfo {volume} {116}},\ \bibinfo
  {pages} {170502} (\bibinfo {year} {2016}{\natexlab{a}})}\BibitemShut
  {NoStop}%
\bibitem [{\citenamefont {Brandao}\ \emph
  {et~al.}(2016{\natexlab{b}})\citenamefont {Brandao}, \citenamefont {Harrow},\
  and\ \citenamefont {Horodecki}}]{brandao2016local}%
  \BibitemOpen
  \bibfield  {author} {\bibinfo {author} {\bibfnamefont {F.~G.}\ \bibnamefont
  {Brandao}}, \bibinfo {author} {\bibfnamefont {A.~W.}\ \bibnamefont
  {Harrow}},\ and\ \bibinfo {author} {\bibfnamefont {M.}~\bibnamefont
  {Horodecki}},\ }\bibfield  {title} {\bibinfo {title} {Local random quantum
  circuits are approximate polynomial-designs},\ }\href@noop {} {\bibfield
  {journal} {\bibinfo  {journal} {Communications in Mathematical Physics}\
  }\textbf {\bibinfo {volume} {346}},\ \bibinfo {pages} {397} (\bibinfo {year}
  {2016}{\natexlab{b}})}\BibitemShut {NoStop}%
\bibitem [{\citenamefont {Hamma}\ \emph
  {et~al.}(2012{\natexlab{a}})\citenamefont {Hamma}, \citenamefont {Santra},\
  and\ \citenamefont {Zanardi}}]{Hamma:2012ov}%
  \BibitemOpen
  \bibfield  {author} {\bibinfo {author} {\bibfnamefont {A.}~\bibnamefont
  {Hamma}}, \bibinfo {author} {\bibfnamefont {S.}~\bibnamefont {Santra}},\ and\
  \bibinfo {author} {\bibfnamefont {P.}~\bibnamefont {Zanardi}},\ }\bibfield
  {title} {\bibinfo {title} {Ensembles of physical states and random quantum
  circuits on graphs},\ }\href
  {http://link.aps.org/doi/10.1103/PhysRevA.86.052324} {\bibfield  {journal}
  {\bibinfo  {journal} {Physical Review A}\ }\textbf {\bibinfo {volume} {86}},\
  \bibinfo {pages} {052324} (\bibinfo {year} {2012}{\natexlab{a}})}\BibitemShut
  {NoStop}%
\bibitem [{\citenamefont {Hamma}\ \emph
  {et~al.}(2012{\natexlab{b}})\citenamefont {Hamma}, \citenamefont {Santra},\
  and\ \citenamefont {Zanardi}}]{hamma2012quantum}%
  \BibitemOpen
  \bibfield  {author} {\bibinfo {author} {\bibfnamefont {A.}~\bibnamefont
  {Hamma}}, \bibinfo {author} {\bibfnamefont {S.}~\bibnamefont {Santra}},\ and\
  \bibinfo {author} {\bibfnamefont {P.}~\bibnamefont {Zanardi}},\ }\bibfield
  {title} {\bibinfo {title} {Quantum entanglement in random physical states},\
  }\href@noop {} {\bibfield  {journal} {\bibinfo  {journal} {Physical review
  letters}\ }\textbf {\bibinfo {volume} {109}},\ \bibinfo {pages} {040502}
  (\bibinfo {year} {2012}{\natexlab{b}})}\BibitemShut {NoStop}%
\bibitem [{\citenamefont {Gard}\ \emph {et~al.}(2020)\citenamefont {Gard},
  \citenamefont {Zhu}, \citenamefont {Barron}, \citenamefont {Mayhall},
  \citenamefont {Economou},\ and\ \citenamefont {Barnes}}]{gard2020efficient}%
  \BibitemOpen
  \bibfield  {author} {\bibinfo {author} {\bibfnamefont {B.~T.}\ \bibnamefont
  {Gard}}, \bibinfo {author} {\bibfnamefont {L.}~\bibnamefont {Zhu}}, \bibinfo
  {author} {\bibfnamefont {G.~S.}\ \bibnamefont {Barron}}, \bibinfo {author}
  {\bibfnamefont {N.~J.}\ \bibnamefont {Mayhall}}, \bibinfo {author}
  {\bibfnamefont {S.~E.}\ \bibnamefont {Economou}},\ and\ \bibinfo {author}
  {\bibfnamefont {E.}~\bibnamefont {Barnes}},\ }\bibfield  {title} {\bibinfo
  {title} {Efficient symmetry-preserving state preparation circuits for the
  variational quantum eigensolver algorithm},\ }\href@noop {} {\bibfield
  {journal} {\bibinfo  {journal} {npj Quantum Information}\ }\textbf {\bibinfo
  {volume} {6}},\ \bibinfo {pages} {1} (\bibinfo {year} {2020})}\BibitemShut
  {NoStop}%
\bibitem [{\citenamefont {Streif}\ \emph {et~al.}(2020)\citenamefont {Streif},
  \citenamefont {Leib}, \citenamefont {Wudarski}, \citenamefont {Rieffel},\
  and\ \citenamefont {Wang}}]{streif2020quantum}%
  \BibitemOpen
  \bibfield  {author} {\bibinfo {author} {\bibfnamefont {M.}~\bibnamefont
  {Streif}}, \bibinfo {author} {\bibfnamefont {M.}~\bibnamefont {Leib}},
  \bibinfo {author} {\bibfnamefont {F.}~\bibnamefont {Wudarski}}, \bibinfo
  {author} {\bibfnamefont {E.}~\bibnamefont {Rieffel}},\ and\ \bibinfo {author}
  {\bibfnamefont {Z.}~\bibnamefont {Wang}},\ }\bibfield  {title} {\bibinfo
  {title} {Quantum algorithms with local particle number conservation: noise
  effects and error correction},\ }\href@noop {} {\bibfield  {journal}
  {\bibinfo  {journal} {arXiv preprint arXiv:2011.06873}\ } (\bibinfo {year}
  {2020})}\BibitemShut {NoStop}%
\bibitem [{\citenamefont {Roberts}\ and\ \citenamefont
  {Yoshida}(2017)}]{roberts2017chaos}%
  \BibitemOpen
  \bibfield  {author} {\bibinfo {author} {\bibfnamefont {D.~A.}\ \bibnamefont
  {Roberts}}\ and\ \bibinfo {author} {\bibfnamefont {B.}~\bibnamefont
  {Yoshida}},\ }\bibfield  {title} {\bibinfo {title} {Chaos and complexity by
  design},\ }\href@noop {} {\bibfield  {journal} {\bibinfo  {journal} {Journal
  of High Energy Physics}\ }\textbf {\bibinfo {volume} {2017}},\ \bibinfo
  {pages} {121} (\bibinfo {year} {2017})}\BibitemShut {NoStop}%
\bibitem [{\citenamefont {Nahum}\ \emph {et~al.}(2017)\citenamefont {Nahum},
  \citenamefont {Ruhman}, \citenamefont {Vijay},\ and\ \citenamefont
  {Haah}}]{nahum2017quantum}%
  \BibitemOpen
  \bibfield  {author} {\bibinfo {author} {\bibfnamefont {A.}~\bibnamefont
  {Nahum}}, \bibinfo {author} {\bibfnamefont {J.}~\bibnamefont {Ruhman}},
  \bibinfo {author} {\bibfnamefont {S.}~\bibnamefont {Vijay}},\ and\ \bibinfo
  {author} {\bibfnamefont {J.}~\bibnamefont {Haah}},\ }\bibfield  {title}
  {\bibinfo {title} {Quantum entanglement growth under random unitary
  dynamics},\ }\href@noop {} {\bibfield  {journal} {\bibinfo  {journal}
  {Physical Review X}\ }\textbf {\bibinfo {volume} {7}},\ \bibinfo {pages}
  {031016} (\bibinfo {year} {2017})}\BibitemShut {NoStop}%
\bibitem [{\citenamefont {Nahum}\ \emph {et~al.}(2018)\citenamefont {Nahum},
  \citenamefont {Vijay},\ and\ \citenamefont {Haah}}]{nahum2018operator}%
  \BibitemOpen
  \bibfield  {author} {\bibinfo {author} {\bibfnamefont {A.}~\bibnamefont
  {Nahum}}, \bibinfo {author} {\bibfnamefont {S.}~\bibnamefont {Vijay}},\ and\
  \bibinfo {author} {\bibfnamefont {J.}~\bibnamefont {Haah}},\ }\bibfield
  {title} {\bibinfo {title} {Operator spreading in random unitary circuits},\
  }\href@noop {} {\bibfield  {journal} {\bibinfo  {journal} {Physical Review
  X}\ }\textbf {\bibinfo {volume} {8}},\ \bibinfo {pages} {021014} (\bibinfo
  {year} {2018})}\BibitemShut {NoStop}%
\bibitem [{\citenamefont {Von~Keyserlingk}\ \emph {et~al.}(2018)\citenamefont
  {Von~Keyserlingk}, \citenamefont {Rakovszky}, \citenamefont {Pollmann},\ and\
  \citenamefont {Sondhi}}]{von2018operator}%
  \BibitemOpen
  \bibfield  {author} {\bibinfo {author} {\bibfnamefont {C.}~\bibnamefont
  {Von~Keyserlingk}}, \bibinfo {author} {\bibfnamefont {T.}~\bibnamefont
  {Rakovszky}}, \bibinfo {author} {\bibfnamefont {F.}~\bibnamefont
  {Pollmann}},\ and\ \bibinfo {author} {\bibfnamefont {S.~L.}\ \bibnamefont
  {Sondhi}},\ }\bibfield  {title} {\bibinfo {title} {Operator hydrodynamics,
  otocs, and entanglement growth in systems without conservation laws},\
  }\href@noop {} {\bibfield  {journal} {\bibinfo  {journal} {Physical Review
  X}\ }\textbf {\bibinfo {volume} {8}},\ \bibinfo {pages} {021013} (\bibinfo
  {year} {2018})}\BibitemShut {NoStop}%
\bibitem [{\citenamefont {Shenker}\ and\ \citenamefont
  {Stanford}(2014)}]{shenker2014black}%
  \BibitemOpen
  \bibfield  {author} {\bibinfo {author} {\bibfnamefont {S.~H.}\ \bibnamefont
  {Shenker}}\ and\ \bibinfo {author} {\bibfnamefont {D.}~\bibnamefont
  {Stanford}},\ }\bibfield  {title} {\bibinfo {title} {Black holes and the
  butterfly effect},\ }\href@noop {} {\bibfield  {journal} {\bibinfo  {journal}
  {Journal of High Energy Physics}\ }\textbf {\bibinfo {volume} {2014}},\
  \bibinfo {pages} {67} (\bibinfo {year} {2014})}\BibitemShut {NoStop}%
\bibitem [{\citenamefont {Brown}\ \emph {et~al.}(2016)\citenamefont {Brown},
  \citenamefont {Roberts}, \citenamefont {Susskind}, \citenamefont {Swingle},\
  and\ \citenamefont {Zhao}}]{brown2016holographic}%
  \BibitemOpen
  \bibfield  {author} {\bibinfo {author} {\bibfnamefont {A.~R.}\ \bibnamefont
  {Brown}}, \bibinfo {author} {\bibfnamefont {D.~A.}\ \bibnamefont {Roberts}},
  \bibinfo {author} {\bibfnamefont {L.}~\bibnamefont {Susskind}}, \bibinfo
  {author} {\bibfnamefont {B.}~\bibnamefont {Swingle}},\ and\ \bibinfo {author}
  {\bibfnamefont {Y.}~\bibnamefont {Zhao}},\ }\bibfield  {title} {\bibinfo
  {title} {Holographic complexity equals bulk action?},\ }\href@noop {}
  {\bibfield  {journal} {\bibinfo  {journal} {Physical review letters}\
  }\textbf {\bibinfo {volume} {116}},\ \bibinfo {pages} {191301} (\bibinfo
  {year} {2016})}\BibitemShut {NoStop}%
\bibitem [{\citenamefont {Stanford}\ and\ \citenamefont
  {Susskind}(2014)}]{stanford2014complexity}%
  \BibitemOpen
  \bibfield  {author} {\bibinfo {author} {\bibfnamefont {D.}~\bibnamefont
  {Stanford}}\ and\ \bibinfo {author} {\bibfnamefont {L.}~\bibnamefont
  {Susskind}},\ }\bibfield  {title} {\bibinfo {title} {Complexity and shock
  wave geometries},\ }\href@noop {} {\bibfield  {journal} {\bibinfo  {journal}
  {Physical Review D}\ }\textbf {\bibinfo {volume} {90}},\ \bibinfo {pages}
  {126007} (\bibinfo {year} {2014})}\BibitemShut {NoStop}%
\bibitem [{\citenamefont {Nakata}\ and\ \citenamefont
  {Murao}(2020)}]{nakata2020generic}%
  \BibitemOpen
  \bibfield  {author} {\bibinfo {author} {\bibfnamefont {Y.}~\bibnamefont
  {Nakata}}\ and\ \bibinfo {author} {\bibfnamefont {M.}~\bibnamefont {Murao}},\
  }\bibfield  {title} {\bibinfo {title} {Generic entanglement entropy for
  quantum states with symmetry},\ }\href@noop {} {\bibfield  {journal}
  {\bibinfo  {journal} {Entropy}\ }\textbf {\bibinfo {volume} {22}},\ \bibinfo
  {pages} {684} (\bibinfo {year} {2020})}\BibitemShut {NoStop}%
\bibitem [{\citenamefont {Kong}\ and\ \citenamefont
  {Liu}(2021)}]{kong2021charge}%
  \BibitemOpen
  \bibfield  {author} {\bibinfo {author} {\bibfnamefont {L.}~\bibnamefont
  {Kong}}\ and\ \bibinfo {author} {\bibfnamefont {Z.-W.}\ \bibnamefont {Liu}},\
  }\bibfield  {title} {\bibinfo {title} {Charge-conserving unitaries typically
  generate optimal covariant quantum error-correcting codes},\ }\href@noop {}
  {\bibfield  {journal} {\bibinfo  {journal} {arXiv preprint arXiv:2102.11835}\
  } (\bibinfo {year} {2021})}\BibitemShut {NoStop}%
\bibitem [{\citenamefont {Zimbor{\'a}s}\ \emph {et~al.}(2015)\citenamefont
  {Zimbor{\'a}s}, \citenamefont {Zeier}, \citenamefont
  {Schulte-Herbr{\"u}ggen},\ and\ \citenamefont
  {Burgarth}}]{zimboras2015symmetry}%
  \BibitemOpen
  \bibfield  {author} {\bibinfo {author} {\bibfnamefont {Z.}~\bibnamefont
  {Zimbor{\'a}s}}, \bibinfo {author} {\bibfnamefont {R.}~\bibnamefont {Zeier}},
  \bibinfo {author} {\bibfnamefont {T.}~\bibnamefont
  {Schulte-Herbr{\"u}ggen}},\ and\ \bibinfo {author} {\bibfnamefont
  {D.}~\bibnamefont {Burgarth}},\ }\bibfield  {title} {\bibinfo {title}
  {Symmetry criteria for quantum simulability of effective interactions},\
  }\href@noop {} {\bibfield  {journal} {\bibinfo  {journal} {Physical Review
  A}\ }\textbf {\bibinfo {volume} {92}},\ \bibinfo {pages} {042309} (\bibinfo
  {year} {2015})}\BibitemShut {NoStop}%
\bibitem [{\citenamefont {Zeier}\ and\ \citenamefont
  {Zimbor{\'a}s}(2015)}]{zeier2015squares}%
  \BibitemOpen
  \bibfield  {author} {\bibinfo {author} {\bibfnamefont {R.}~\bibnamefont
  {Zeier}}\ and\ \bibinfo {author} {\bibfnamefont {Z.}~\bibnamefont
  {Zimbor{\'a}s}},\ }\bibfield  {title} {\bibinfo {title} {On squares of
  representations of compact lie algebras},\ }\href@noop {} {\bibfield
  {journal} {\bibinfo  {journal} {Journal of Mathematical Physics}\ }\textbf
  {\bibinfo {volume} {56}},\ \bibinfo {pages} {081702} (\bibinfo {year}
  {2015})}\BibitemShut {NoStop}%
\bibitem [{\citenamefont {Halpern}\ \emph {et~al.}(2020)\citenamefont
  {Halpern}, \citenamefont {Beverland},\ and\ \citenamefont
  {Kalev}}]{halpern2020noncommuting}%
  \BibitemOpen
  \bibfield  {author} {\bibinfo {author} {\bibfnamefont {N.~Y.}\ \bibnamefont
  {Halpern}}, \bibinfo {author} {\bibfnamefont {M.~E.}\ \bibnamefont
  {Beverland}},\ and\ \bibinfo {author} {\bibfnamefont {A.}~\bibnamefont
  {Kalev}},\ }\bibfield  {title} {\bibinfo {title} {Noncommuting conserved
  charges in quantum many-body thermalization},\ }\href@noop {} {\bibfield
  {journal} {\bibinfo  {journal} {Physical Review E}\ }\textbf {\bibinfo
  {volume} {101}},\ \bibinfo {pages} {042117} (\bibinfo {year}
  {2020})}\BibitemShut {NoStop}%
\bibitem [{\citenamefont {Bartlett}\ \emph {et~al.}(2002)\citenamefont
  {Bartlett}, \citenamefont {de~Guise},\ and\ \citenamefont
  {Sanders}}]{bartlett2002quantum}%
  \BibitemOpen
  \bibfield  {author} {\bibinfo {author} {\bibfnamefont {S.~D.}\ \bibnamefont
  {Bartlett}}, \bibinfo {author} {\bibfnamefont {H.}~\bibnamefont {de~Guise}},\
  and\ \bibinfo {author} {\bibfnamefont {B.~C.}\ \bibnamefont {Sanders}},\
  }\bibfield  {title} {\bibinfo {title} {Quantum encodings in spin systems and
  harmonic oscillators},\ }\href@noop {} {\bibfield  {journal} {\bibinfo
  {journal} {Physical Review A}\ }\textbf {\bibinfo {volume} {65}},\ \bibinfo
  {pages} {052316} (\bibinfo {year} {2002})}\BibitemShut {NoStop}%
\bibitem [{\citenamefont {Wang}\ \emph {et~al.}(2003)\citenamefont {Wang},
  \citenamefont {Sanders},\ and\ \citenamefont {Berry}}]{wang2003entangling}%
  \BibitemOpen
  \bibfield  {author} {\bibinfo {author} {\bibfnamefont {X.}~\bibnamefont
  {Wang}}, \bibinfo {author} {\bibfnamefont {B.~C.}\ \bibnamefont {Sanders}},\
  and\ \bibinfo {author} {\bibfnamefont {D.~W.}\ \bibnamefont {Berry}},\
  }\bibfield  {title} {\bibinfo {title} {Entangling power and operator
  entanglement in qudit systems},\ }\href@noop {} {\bibfield  {journal}
  {\bibinfo  {journal} {Physical Review A}\ }\textbf {\bibinfo {volume} {67}},\
  \bibinfo {pages} {042323} (\bibinfo {year} {2003})}\BibitemShut {NoStop}%
\bibitem [{\citenamefont {Keet}\ \emph {et~al.}(2010)\citenamefont {Keet},
  \citenamefont {Fortescue}, \citenamefont {Markham},\ and\ \citenamefont
  {Sanders}}]{keet2010quantum}%
  \BibitemOpen
  \bibfield  {author} {\bibinfo {author} {\bibfnamefont {A.}~\bibnamefont
  {Keet}}, \bibinfo {author} {\bibfnamefont {B.}~\bibnamefont {Fortescue}},
  \bibinfo {author} {\bibfnamefont {D.}~\bibnamefont {Markham}},\ and\ \bibinfo
  {author} {\bibfnamefont {B.~C.}\ \bibnamefont {Sanders}},\ }\bibfield
  {title} {\bibinfo {title} {Quantum secret sharing with qudit graph states},\
  }\href@noop {} {\bibfield  {journal} {\bibinfo  {journal} {Physical Review
  A}\ }\textbf {\bibinfo {volume} {82}},\ \bibinfo {pages} {062315} (\bibinfo
  {year} {2010})}\BibitemShut {NoStop}%
\bibitem [{\citenamefont {Bombin}\ and\ \citenamefont
  {Martin-Delgado}(2005)}]{bombin2005entanglement}%
  \BibitemOpen
  \bibfield  {author} {\bibinfo {author} {\bibfnamefont {H.}~\bibnamefont
  {Bombin}}\ and\ \bibinfo {author} {\bibfnamefont {M.~A.}\ \bibnamefont
  {Martin-Delgado}},\ }\bibfield  {title} {\bibinfo {title} {Entanglement
  distillation protocols and number theory},\ }\href@noop {} {\bibfield
  {journal} {\bibinfo  {journal} {Physical Review A}\ }\textbf {\bibinfo
  {volume} {72}},\ \bibinfo {pages} {032313} (\bibinfo {year}
  {2005})}\BibitemShut {NoStop}%
\bibitem [{\citenamefont {Wang}\ \emph {et~al.}(2020)\citenamefont {Wang},
  \citenamefont {Hu}, \citenamefont {Sanders},\ and\ \citenamefont
  {Kais}}]{wang2020qudits}%
  \BibitemOpen
  \bibfield  {author} {\bibinfo {author} {\bibfnamefont {Y.}~\bibnamefont
  {Wang}}, \bibinfo {author} {\bibfnamefont {Z.}~\bibnamefont {Hu}}, \bibinfo
  {author} {\bibfnamefont {B.~C.}\ \bibnamefont {Sanders}},\ and\ \bibinfo
  {author} {\bibfnamefont {S.}~\bibnamefont {Kais}},\ }\bibfield  {title}
  {\bibinfo {title} {Qudits and high-dimensional quantum computing},\
  }\href@noop {} {\bibfield  {journal} {\bibinfo  {journal} {Frontiers in
  Physics}\ }\textbf {\bibinfo {volume} {8}},\ \bibinfo {pages} {479} (\bibinfo
  {year} {2020})}\BibitemShut {NoStop}%
\bibitem [{\citenamefont {Blok}\ \emph {et~al.}(2021)\citenamefont {Blok},
  \citenamefont {Ramasesh}, \citenamefont {Schuster}, \citenamefont {O'Brien},
  \citenamefont {Kreikebaum}, \citenamefont {Dahlen}, \citenamefont {Morvan},
  \citenamefont {Yoshida}, \citenamefont {Yao},\ and\ \citenamefont
  {Siddiqi}}]{blok2021quantum}%
  \BibitemOpen
  \bibfield  {author} {\bibinfo {author} {\bibfnamefont {M.}~\bibnamefont
  {Blok}}, \bibinfo {author} {\bibfnamefont {V.}~\bibnamefont {Ramasesh}},
  \bibinfo {author} {\bibfnamefont {T.}~\bibnamefont {Schuster}}, \bibinfo
  {author} {\bibfnamefont {K.}~\bibnamefont {O'Brien}}, \bibinfo {author}
  {\bibfnamefont {J.}~\bibnamefont {Kreikebaum}}, \bibinfo {author}
  {\bibfnamefont {D.}~\bibnamefont {Dahlen}}, \bibinfo {author} {\bibfnamefont
  {A.}~\bibnamefont {Morvan}}, \bibinfo {author} {\bibfnamefont
  {B.}~\bibnamefont {Yoshida}}, \bibinfo {author} {\bibfnamefont
  {N.}~\bibnamefont {Yao}},\ and\ \bibinfo {author} {\bibfnamefont
  {I.}~\bibnamefont {Siddiqi}},\ }\bibfield  {title} {\bibinfo {title} {Quantum
  information scrambling on a superconducting qutrit processor},\ }\href@noop
  {} {\bibfield  {journal} {\bibinfo  {journal} {Physical Review X}\ }\textbf
  {\bibinfo {volume} {11}},\ \bibinfo {pages} {021010} (\bibinfo {year}
  {2021})}\BibitemShut {NoStop}%
\bibitem [{\citenamefont {Fisher}\ \emph {et~al.}(1989)\citenamefont {Fisher},
  \citenamefont {Weichman}, \citenamefont {Grinstein},\ and\ \citenamefont
  {Fisher}}]{PhysRevB.40.546}%
  \BibitemOpen
  \bibfield  {author} {\bibinfo {author} {\bibfnamefont {M.~P.~A.}\
  \bibnamefont {Fisher}}, \bibinfo {author} {\bibfnamefont {P.~B.}\
  \bibnamefont {Weichman}}, \bibinfo {author} {\bibfnamefont {G.}~\bibnamefont
  {Grinstein}},\ and\ \bibinfo {author} {\bibfnamefont {D.~S.}\ \bibnamefont
  {Fisher}},\ }\bibfield  {title} {\bibinfo {title} {Boson localization and the
  superfluid-insulator transition},\ }\href
  {https://doi.org/10.1103/PhysRevB.40.546} {\bibfield  {journal} {\bibinfo
  {journal} {Phys. Rev. B}\ }\textbf {\bibinfo {volume} {40}},\ \bibinfo
  {pages} {546} (\bibinfo {year} {1989})}\BibitemShut {NoStop}%
\bibitem [{\citenamefont {Matsubara}\ and\ \citenamefont
  {Matsuda}(1956)}]{10.1143/PTP.16.569}%
  \BibitemOpen
  \bibfield  {author} {\bibinfo {author} {\bibfnamefont {T.}~\bibnamefont
  {Matsubara}}\ and\ \bibinfo {author} {\bibfnamefont {H.}~\bibnamefont
  {Matsuda}},\ }\bibfield  {title} {\bibinfo {title} {{A Lattice Model of
  Liquid Helium, I}},\ }\href {https://doi.org/10.1143/PTP.16.569} {\bibfield
  {journal} {\bibinfo  {journal} {Progress of Theoretical Physics}\ }\textbf
  {\bibinfo {volume} {16}},\ \bibinfo {pages} {569} (\bibinfo {year} {1956})},\
  \Eprint
  {https://arxiv.org/abs/https://academic.oup.com/ptp/article-pdf/16/6/569/5383838/16-6-569.pdf}
  {https://academic.oup.com/ptp/article-pdf/16/6/569/5383838/16-6-569.pdf}
  \BibitemShut {NoStop}%
\bibitem [{\citenamefont {Batyev}\ and\ \citenamefont
  {Braginskii}(1984)}]{Batyev84}%
  \BibitemOpen
  \bibfield  {author} {\bibinfo {author} {\bibfnamefont {E.}~\bibnamefont
  {Batyev}}\ and\ \bibinfo {author} {\bibfnamefont {L.}~\bibnamefont
  {Braginskii}},\ }\bibfield  {title} {\bibinfo {title} {Antiferromagnet in a
  strong magnetic field: analogy with bose gas},\ }\href
  {http://jetp.ac.ru/cgi-bin/e/index/e/60/4/p781?a=list} {\bibfield  {journal}
  {\bibinfo  {journal} {Sov. Phys. JETP}\ }\textbf {\bibinfo {volume} {60}},\
  \bibinfo {pages} {781} (\bibinfo {year} {1984})}\BibitemShut {NoStop}%
\bibitem [{\citenamefont {Wilde}(2013)}]{wilde2013quantum}%
  \BibitemOpen
  \bibfield  {author} {\bibinfo {author} {\bibfnamefont {M.~M.}\ \bibnamefont
  {Wilde}},\ }\href@noop {} {\emph {\bibinfo {title} {Quantum information
  theory}}}\ (\bibinfo  {publisher} {Cambridge University Press},\ \bibinfo
  {year} {2013})\BibitemShut {NoStop}%
\end{thebibliography}%

\onecolumngrid
\newpage
\appendix

\newpage

\section*{Appendices}

\newcommand\appitem[2]{\hyperref[{#1}]{\textbf{\cref*{#1} \nameref*{#1}}} \dotfill \pageref{#1}\\ \begin{minipage}[t]{0.8\textwidth} #2\end{minipage}}
\begin{itemize}
\item \appitem{app:Schur-Weyl}{We briefly review the Schur-Weyl duality for a qudit system, and the Young diagrams relevant for the 6 qutrit examples in \cref{fig:4-local-fskew,fig:l-shape}.}
\item \appitem{App:f}{In \cref{App:Kf} we collect some interesting properties of the projector $K$ and the function $f_{\sgn}$, and present the proofs. In \cref{app:genZ2}, we generalize the $\mathbb{Z}_2$ symmetry in \cref{lem13}, which is formalized in \cref{thm:Z2}, and prove this theorem. }
\item \appitem{App:Fermi}{We collect proofs for the statements in \cref{Sec:ferm}.}
\item \appitem{app:relative-phases}{
In the case of the group generated by 2-local $\SU(d)$-invariant unitaries $\mathcal{V}_{2}$, a simple characterization of the relative phases between different charge sectors in terms of the eigenvalues of Casimir operator is presented.}
\end{itemize}

\newpage

\section{Schur-Weyl duality}\label{app:Schur-Weyl}

\subsection{Decomposing systems of qudits to irreps}\label{app:duality}

Recall the decomposition of $n$ qudits into charge sectors,
\begin{equation}\tag{re \ref{decomp:SU}}
  (\complex^d)^{\otimes n} \cong \bigoplus_\lambda \hilbert_\lambda = \bigoplus_\lambda \hilbert[Q]_\lambda \otimes \hilbert[M]_\lambda\ ,
\end{equation}
where $\lambda$ labels inequivalent irreps of $\SU(d)$, $\mathcal{H}_\lambda$ is the sector corresponding to irrep $\lambda$, also known as the isotypic component of $\lambda$, $\hilbert[Q]_\lambda$ is the irrep of $\SU(d)$ corresponding to $\lambda$, and $\hilbert[M]_\lambda$ is the (virtual) multiplicity subsystem \cite{QRF_BRS_07, zanardi2001virtual}, where $\SU(d)$ acts trivially.

Schur's Lemmas imply  that any unitary $V$ which commutes with $U^{\otimes n}$ for all $U \in \SU(d)$ has a decomposition $V \cong \bigoplus \identity_{\hilbert[Q]_\lambda} \otimes v_\lambda$ where each $v_\lambda$ is unitary. In particular, the representation of the permutation group defined in \cref{eq:permrep} clearly commutes with this representation of $\SU(d)$, i.e., for a general permutation $\sigma \in \mathbb{S}_n$ and qudit rotation $U \in \SU(d)$, we have $[\P(\sigma), U^{\otimes n}] = 0$. This immediately defines representations $\mathbf{p}_\lambda$ on the multiplicity spaces $\mathcal{M}_\lambda$ via 
\be
\P(\sigma) \cong \bigoplus_\lambda \identity_{\hilbert[Q]_\lambda} \otimes \mathbf{p}_\lambda(\sigma)\ .
\ee
In fact, the relationships between these two representations is much stronger: the commutant of $\cset{U^{\otimes n}}{U \in \SU(d)}$ (that is, the set of all operators which commute with all $U^{\otimes n}$) is \emph{spanned} by $\cset{\P(\sigma)}{\sigma \in \mathbb{S}_n}$, and vice-versa. This has the consequence that, when decomposing the Hilbert space $(\complex^d)^{\otimes n}$ into irreps $\lambda$ of $\SU(d)$, the multiplicity subsystem is actually an irrep of $\mathbb{S}_n$. That is, the representations $\mathbf{p}_\lambda$ are actually irreducible. Furthermore, for inequivalent irreps $\lambda$ and $\mu$ of $\SU(d)$, the corresponding irreps of $\mathbb{S}_n$ are also inequivalent. This relationship between the irreps of $\SU(d)$ and $\mathbb{S}_n$ in the decomposition of $(\complex^d)^{\otimes n}$ is known as Schur-Weyl duality.

It turns out that there is a natural labeling of the irreps given by the Young diagrams which have $n$ boxes and $\leq d$ rows \cite{harrow2005applications,goodman2009symmetry}, which we use in \cref{app:Young-diagrams} to describe the 6 qutrit examples.

\subsection{6 qutrit examples}\label{app:Young-diagrams}

In this section, we consider the 6 qutrit examples of \cref{fig:4-local-fskew,fig:l-shape}. We use Schur-Weyl duality to describe the irreps in which they appear. In the specific case of $d = 3$ and $n = 6$, the Young diagrams that show up in the decomposition are
\begin{equation}\label{ydiagrams}
  \ydiagram{6}\, , \; \; \ydiagram{5, 1}\, , \; \; \ydiagram{4, 2}\, , \; \; \ydiagram{4, 1, 1}\, , \; \; \ydiagram{3, 3}\, , \; \; \ydiagram{3, 2, 1}\, , \; \; \ydiagram{2, 2, 2}\, .
\end{equation}

The first, second, and fourth diagrams are the `L-shape' diagrams, and the fermionic correspondence is defined inside of these subspaces. In particular, the state $(\nket{0} \wedge \nket{1} \wedge \nket{2}) \otimes \nket{0}^{\otimes 3}$ from \cref{sect:examp,fig:l-shape} is an unentangled state inside of $\hilbert[Q]_\mu \otimes \hilbert[M]_\mu$, where
\begin{equation*}
  \mu = {\:} \vcenter{\hbox{\ydiagram{4, 1, 1}}}\, .
\end{equation*}

The state $(\nket{0} \wedge \nket{1} \wedge \nket{2}) \otimes (\nket{0} \wedge \nket{1}) \otimes \nket{0}$ from \cref{sec:exZ2,fig:4-local-fskew} is an unentangled state inside $\hilbert[Q]_\nu \otimes \hilbert[M]_\nu$, where
\begin{equation*}
  \nu = {\:} \vcenter{\hbox{\ydiagram{3, 2, 1}}}\, .
\end{equation*}
As we show in \cite{HLM_2021}, the diagrams that are symmetric under reflection across the diagonal (such as this one) have a special relationship with the $\mathbb{Z}_2$ symmetry. Namely, the states which are restricted to a single charge sector and also have nonzero $f_{\sgn}$ are necessarily inside a sector with a symmetric Young diagram (see also \cref{SubSec:K} for a related discussion\footnote{The symmetric diagrams are the ones whose permutation irrep $\mathbf{p}_\lambda$ satisfies $\mathbf{p}_\lambda \cong {\sgn} \otimes \mathbf{p}_\lambda$ (see Exercise 4.4 of \cite{fulton2013representation}).}).

\newpage

\newpage

\section{More on \texorpdfstring{$\mathbb{Z}_2$}{Z2} symmetry and properties of function \texorpdfstring{$f_{\sgn}$}{fsgn}}\label{App:f}

\subsection{Properties of projector \texorpdfstring{$K$}{K} and function \texorpdfstring{$f_{\sgn}$}{fsgn}}\label{App:Kf}

Here, we further study the properties of the operator $K$ and function $f_{\sgn}$. Recall that the operator $K$, defined by 
\begin{equation*}\tag{re \ref{Def:K}}
  K = \frac{1}{n!} \sum_{\sigma \in \mathbb{S}_n} \sgn(\sigma) \P(\sigma) \otimes \P(\sigma) 
\end{equation*}
is a Hermitian projector to the sign representation of the permutation group. The fact that it is a Hermitian projector can be easily seen:
\begin{equation}
  K^2 = \frac{1}{(n!)^2} \sum_{\sigma, \tau} \sgn(\sigma \tau) \P(\sigma \tau) \otimes \P(\sigma \tau) = \frac{1}{(n!)^2} \sum_{\sigma, \tau} \sgn(\sigma) \P(\sigma) \otimes \P(\sigma) = K,
\end{equation}
where we resummed in the second equality, and
\begin{equation}
  K^\dagger = \frac{1}{n!} \sum_\sigma \sgn(\sigma) \P(\sigma)^\dagger \otimes \P(\sigma)^\dagger = \frac{1}{n!} \sum_\sigma \sgn(\sigma) \P(\sigma^{-1}) \otimes \P(\sigma^{-1}) = \frac{1}{n!} \sum_\sigma \sgn(\sigma) \P(\sigma) \otimes \P(\sigma) = K\ ,
\end{equation}
where in the third equality we resummed and used the fact that the parity of a permutation and its inverse are the same, i.e., $\sgn(\sigma) = \sgn(\sigma^{-1})$. For any permutation $\sigma\in\mathbb{S}_n$ acting on this projector, we have 
\begin{equation}
  \bigl(\P(\sigma) \otimes \P(\sigma)\bigr) K = \frac{1}{n!} \sum_\tau \sgn(\tau) \P(\sigma \tau) \otimes \P(\sigma \tau) = \frac{\sgn(\sigma)}{n!} \sum_\tau \sgn(\tau) \P(\tau) \otimes \P(\tau) = \sgn(\sigma) K,
\end{equation}
by resumming in the second equality and again using $\sgn(\sigma) = \sgn(\sigma^{-1})$. This means $K$ is the projector to the sign representation of $\mathbb{S}_n$ on $(\mathbb{C}^d \otimes \mathbb{C}^d)^{\otimes n}$ and therefore commutes with the action of $\mathbb{S}_n$, i.e.
\be\label{K-inv-1}
\forall \tau\in \mathbb{S}_n: \ \ \big[\P(\tau)\otimes \P(\tau)\big] K \big[\P(\tau)^\dag\otimes \P(\tau)^\dag\big]= K\ .
\ee
Furthermore, note that since $(\mathbb{C}^d \otimes \mathbb{C}^d)^{\otimes n}$ does not have a totally antisymmetric subspace for $n > d^2$, it follows that $K = 0$ for such values of $n$.

On the other hand, in the case of $n = d^2$, there is a 1D totally anti-symmetric subspace, which means 
\be\label{eq:K-eta}
K = \qproj{\eta}\ ,
\ee
where $\nket{\eta}\in(\mathbb{C}^d\otimes \mathbb{C}^d)^{\otimes n}$ is the unique (up to a global phase) totally antisymmetric state. This means that for any operator $T\in \End(\complex^d \otimes \complex^d)$ over a pair of qudits, it satisfies
\begin{equation*}
  T^{\otimes n} \nket{\eta} = \det (T) \nket{\eta}\ .
\end{equation*}
Considering the swap operator $\mathds{X}$ on a pair of qudits, this implies
\be\label{eq:XK}
\mathds{X}^{\otimes n}K=K \mathds{X}^{\otimes n}=\det (\mathds{X}) K=(-1)^{ \frac{d (d - 1)}{2}}\ K\ ,
\ee
where $\frac{d (d - 1)}{2} = \binom{d}{2}$ is the dimension of the antisymmetric subspace of $\complex^d \otimes \complex^d$. 

Furthermore, using the fact that $[U^{\otimes n}, \mathbf{P}(\sigma)]=0$ for all $U\in\SU(d)$ and all $\sigma\in\mathbb{S}_n$, we find that projector $K$ also satisfies 
\be\label{K-inv3}
\forall U\in \SU(d): \ \ \big[U^{\otimes n}\otimes \mathbb{I}^{\otimes n}\big] K \big[U^{\otimes n}\otimes \mathbb{I}^{\otimes n}\big]^\dag= K\ ,
\ee
and
\be\label{K-inv-2}
\forall U\in \SU(d): \ \ \big[U^{\otimes n}\otimes U^{\otimes n}\big] K \big[U^{\otimes n}\otimes U^{\otimes n}\big]^\dag= K\ .
\ee
Since in the case of $n=d^2$ operator $K=|\eta\rangle\langle\eta|$, (using \cref{K-inv3}) we find
\be
[U^{\otimes n}\otimes \mathbb{I}^{\otimes n}] |\eta\rangle=\e^{i\phi} |\eta\rangle\ ,
\ee
for some phase $\e^{i\phi}$. This means that
\be
\nket{\eta}\in\mathcal{H}_\mathrm{singlet}\otimes \mathcal{H}_\mathrm{singlet}\ ,
\ee
where $\mathcal{H}_\mathrm{singlet}$ is the subspace of $(\mathbb{C}^{d})^{\otimes n}=(\mathbb{C}^{d})^{\otimes d^2}$ formed from states that are invariant under the action of $U^{\otimes n}$ for all $U\in \SU(d)$. In general, this subspace has dimension larger than 1 and the state $|\eta\rangle$ is entangled relative to the bipartite tensor product decomposition $(\mathbb{C}^d)^{\otimes n} \otimes (\mathbb{C}^d)^{\otimes n}$. In \cref{Sec:examp:4} we discuss an example of $n=4$ qubits with $\SU(2)$ symmetry and determine the vector $|\eta\rangle$. \\

Next, we present some useful properties of function $f_{\sgn}$, which all follow from the above observations.

\begin{proposition}
  Recalling the definition of $f_{\sgn}$,
\begin{equation*}\tag{re \ref{def:fskew}}
  f_{\sgn}[\psi] = \bigl(\langle \psi | \otimes \langle \psi | \bigr) K \bigl(| \psi \rangle \otimes | \psi \rangle \bigr) = \frac{1}{n!} \sum_\sigma \sgn(\sigma) \langle \psi | \P(\sigma) | \psi \rangle^2,
\end{equation*}
we have
\begin{enumerate}[(i)]
\item $f_{\sgn}$ is invariant under the action of $\SU(d)$ and permutations.
\item $f_{\sgn}[\psi] = 0$ for all states if $n > d^2$.
\item In the special case of $n = d^2$, $f_{\sgn}$ is a non-zero function if, and only if, $d (d - 1) / 2$ is an even number.
\item If $|\psi\rangle$ is invariant up to a phase under a swap $\P_{ab}$, i.e. $\P_{ab} |\psi\rangle = \pm |\psi\rangle$, then $f_{\sgn}[\psi] = 0$.
\end{enumerate}
\end{proposition}

\begin{proof}
  Item (i) follows immediately from the invariance of $K$ under the actions of $\mathbb{S}_n$ and $\SU(d)$ in \cref{K-inv-1}) and \cref{K-inv-2}. Similarly, item (ii) follows from the fact that $K = 0$ for $n > d^2$, (see the discussion immediately before this proposition). In the case of $n=d^2$, \cref{eq:K-eta} implies 
\begin{equation}
  f_{\sgn}[\psi] = |\p[\big]{\nbra{\psi} \otimes \nbra{\psi}} \nket{\eta}|^2 \ ,
\end{equation}
and \cref{eq:XK} together with $\mathds{X}^{\otimes n} \p[\big]{\nket{\psi} \otimes \nket{\psi}} = \nket{\psi} \otimes \nket{\psi}$ implies
\begin{equation}
  f_{\sgn}[\psi] = (-1)^{\frac{d (d - 1)}{2}} f_{\sgn}[\psi] \ .
\end{equation}
Therefore, $f_{\sgn}[\psi] = -f_{\sgn}[\psi] = 0$ for all $\nket{\psi} \in (\complex^d)^{\otimes n}$ when $d (d - 1) / 2$ is odd. On the other hand, when $d (d - 1) / 2$ is even, it is clear that $\nket{\eta}$ is symmetric with respect to the bipartite decomposition $(\complex^d)^{\otimes n} \otimes (\complex^d)^{\otimes n}$. Since the symmetric subspace is spanned by vectors of the form $\nket{\psi} \otimes \nket{\psi}$ with $\nket{\psi} \in (\complex^d)^{\otimes n}$, there necessarily exists at least one state $\nket{\psi}$ so that
\begin{equation}
  \p[\big]{\nbra{\psi} \otimes \nbra{\psi}} \nket{\eta} \neq 0.
\end{equation}
This means $f_{\sgn}[\psi]=|\p[\big]{\nbra{\psi} \otimes \nbra{\psi}} \nket{\eta}|^2>0$, which proves (iii).

For (iv), we prove something slightly more general. Suppose that $\tau \in \mathbb{S}_n$ is an odd permutation, $\sgn(\tau) = -1$, and $|\psi\rangle$ is such that $\P(\tau) |\psi\rangle = \pm |\psi \rangle$. Then, because for any permutation $\pi$ it is true that $\sgn(\pi) = \sgn(\pi^{-1})$,
\begin{equation}
  f_{\sgn}[\psi] = \frac{1}{n!} \sum_\sigma \sgn(\sigma) \langle \psi | \P(\sigma) \P(\tau) | \psi \rangle^2 = \frac{1}{n!} \sum_\sigma \sgn(\sigma \tau^{-1}) \langle \psi | \P(\sigma) | \psi \rangle^2 = -f_{\sgn}[\psi],
\end{equation}
and thus $f_{\sgn}[\psi] = 0$.
\end{proof}

\subsection{The case of \texorpdfstring{$\SU(2)$}{SU(2)} symmetry for \texorpdfstring{$n=3, 4$}{n=3,4} qubits (the proof of the second part of Corollary \ref{tyt})}\label{Sec:examp:4}

As we saw before, for $n>d^2$ qudits operator $K$ always vanishes. Therefore, since $K(H\otimes \mathbb{I}+ \mathbb{I}\otimes H)=0$ for any $SU(d)$-invariant Hamiltonian $H$ of such systems, by \cref{lem13} we conclude that $H$ has a decomposition as $H=\sum_\sigma h_\sigma \P(\sigma)$ satisfying $h_\sigma=-\sgn(\sigma) h^\ast_\sigma$ for all $\sigma\in\mathbb{S}_n$. 

In particular, in the case of qubits, this means that this property holds for $n>4$ qubits. Interestingly, while for $n=3, 4$ qubits $K\neq 0$, as we show in the following, any Hamiltonian $H$ can be shifted by a multiple of the identity operator to $H'=H+\alpha \mathbb{I}$, such that $H'$ satisfies the condition of \cref{lem13}, that is $K(H'\otimes \mathbb{I}+ \mathbb{I}\otimes H')=0$. 
Using \cref{lem13,Thm5}, this means that function $f_{\sgn}$ is always conserved under $\SU(2)$-invariant Hamiltonians, even though it is generally non-zero for $n=3$ and $4$ qubits.

To prove the above claim for $n=3$ and $4$ qubits, first recall that $K$ is the projector to the sign representation of the permutation group $\mathbb{S}_n$ in $(\complex^2)^{\otimes n} \otimes (\complex^2)^{\otimes n}$, where permutations act as $\P(\sigma) \otimes \P(\sigma)$ for $\P(\sigma)$ a permutation of $n$ qubits. Recall that
 under the action of $\mathbb{S}_n$ the Hilbert space of $n$ qubits decomposes as
\begin{equation}
  (\complex^2)^{\otimes n} \cong \bigoplus_{j=j_{\min}}^{j_{\max}} \mathcal{H}_j \cong \bigoplus_{j=j_{\min}}^{j_{\max}} \mathbb{C}^{2j+1}\otimes \mathcal{M}_j\ ,
\end{equation}
where $j_\text{max}=n/2$ and  $j_\text{min}=0$, or $j_\text{min}=1/2$ for even and odd $n$, respectively. Here, ${2j+1}$ is the dimension of the irrep of SU(2) with angular momentum $j$, corresponding to the eigenvalue $j(j+1)$  of the squared angular momentum operator $J^2=J_x^2+J_y^2+J_z^2$, and  $\mathcal{M}_j$ is the subsystem corresponding to the multiplicity of this irrep.  According to the Schur-Weyl duality,  $\mathbb{S}_n$ acts irreducibly on $\mathcal{M}_j$. That is the representation of permutation $\sigma\in \mathbb{S}_n$ can be decomposed as
\be
\P(\sigma) \cong \bigoplus_j \ \identity_{2j+1} \otimes \mathbf{p}_j(\sigma)\ ,
\ee
where $\mathbf{p}_j$ is an irrep of  $\mathbb{S}_n$. We see that, by distributing the tensor product through the direct sum,
\begin{equation}
  (\complex^2)^{\otimes n} \otimes (\complex^2)^{\otimes n} \cong \bigoplus_{j, j' = j_{\min}}^{j_{\max}} \p[\big]{\complex^{2 j + 1} \otimes \complex^{2 j' + 1}} \otimes \mathcal{M}_j \otimes \mathcal{M}_{j'}.
\end{equation}
Thus, on $(\complex^2)^{\otimes n} \otimes (\complex^2)^{\otimes n}$, the representation of $\mathbb{S}_n$ has a decomposition into (potentially reducible) representations $\pb_j \otimes \pb_{j'}$. Therefore, to characterize operator $K$, we need to determine for which values of $j$ and $j'$, the tensor product representation $\pb_j \otimes \pb_{j'}$ contains the sign representation of $\mathbb{S}_n$.

First, consider the case of $n=4$. Since in this case the condition $n=d^2$ holds, we have $K=|\eta\rangle\langle\eta |$, where 
\be
|\eta\rangle\in \mathcal{H}_{j=0}\otimes \mathcal{H}_{j=0}= (\mathbb{C}\otimes \mathbb{C}^2)\otimes(\mathbb{C}\otimes \mathbb{C}^2) \ ,
\ee
where we have used the fact that the multiplicity of $j=0$, is $2$. Using the fact that $w=d(d-1)/2=1$ is odd, we find that under swap $\mathds{X}^{\otimes 2}$, $|\eta\rangle$ should obtain a minus sign, i.e., it should be in the totally anti-symmetric subspace of $ \mathcal{H}_{j=0}\otimes \mathcal{H}_{j=0}$. We conclude that in this case 
\begin{equation}
  K \cong \identity_1 \otimes \identity_1 \otimes \qproj{\Psi_-} \ ,
\end{equation}
where $\nket{\Psi_-}$ is the singlet state which takes the form (up to an overall phase)
\begin{equation}
  \nket{\Psi_-} = \frac{1}{\sqrt{2}} \p[\big]{\nket{0} \otimes \nket{1} - \nket{1} \otimes \nket{0}}\ , 
\end{equation}
in any orthonormal basis $\sets{\nket{0}, \nket{1}}$ for $\mathcal{M}_0$.

Next, consider the case of $n=3$. In this case $j$ takes values $1/2$ and $3/2$. The case of $j=3/2$ has multiplicity $1$, which corresponds to the totally symmetric subspace, i.e., the trivial representation of $\mathbb{S}_3$. The case of $j=1/2$ has multiplicity $2$, which corresponds to the only non-Abelian irrep of $\mathbb{S}_3$. Therefore, the only case, where $\pb_j \otimes \pb_{j'}$ can contain the sign representation of $\mathbb{S}_n$ is when $j=j'=1/2$, where $\mathcal{M}_{1/2} \otimes \mathcal{M}_{1/2}=\mathbb{C}^{ 2}\otimes \mathbb{C}^{2}$. It can be easily seen that in this case under the action of unitaries $\{\pb_{1/2}(\sigma) \otimes \pb_{1/2}(\sigma):\sigma\in\mathbb{S}_3\}$ the space $\mathcal{M}_{1/2} \otimes \mathcal{M}_{1/2}=\mathbb{C}^{ 2}\otimes \mathbb{C}^{2}$ decomposes to 3 irreps of $\mathbb{S}_3$. In particular, the subspace corresponding to state $\nket{\Psi_-} = \frac{1}{\sqrt{2}} \p[\big]{\nket{0} \otimes \nket{1} - \nket{1} \otimes \nket{0}}\ $ corresponds to the sign representation, that is
\be
\pb_{1/2}(\sigma) \otimes \pb_{1/2}(\sigma)\nket{\Psi_-}=\det(\pb_{1/2}(\sigma))\nket{\Psi_-}={\sgn}(\sigma)\nket{\Psi_-}\ .
\ee
The latter fact can be seen, for instance, by noting that in the 2D irrep of $\mathbb{S}_3$, transpositions are represented by unitary matrices with eigenvalues $+1$ and $-1$, i.e., with matrices with determinant -1. This means $\det(\pb_{1/2}(\sigma))={\sgn}(\sigma)$. We conclude that in the case of $n=3$, $K \cong \identity_{2} \otimes \identity_2 \otimes \qproj{\Psi_-}$. To summarize, we conclude that 
\begin{equation}\label{twocasesK}
  K \cong \begin{cases} \identity_{2} \otimes \identity_2 \otimes \qproj{\Psi_-} & \text{for $n = 3$} \\ \identity_1 \otimes \identity_1 \otimes \qproj{\Psi_-} & \text{for $n = 4$\ .} \end{cases}
\end{equation}
This result implies that in both cases, any $\SU(2)$-invariant Hermitian operator $A$, satisfies $(A \otimes \identity + \identity \otimes A) K = \alpha K\ $ for a number $\alpha$. To see this note that any $\SU(2)$-invariant operator $A$ acts non-trivially only in the multiplicity subsystems $\mathcal{M}_j$, i.e.,
\begin{equation}
  A \cong \bigoplus_j \identity_{2 j + 1} \otimes \mathsf{a}_j.
\end{equation}
Recall that for any qubit operator $B\in\End(\mathbb{C}^2)$, it holds that $(B \otimes \identity + \identity \otimes B) \nket{\Psi_-} = \Tr (B) \nket{\Psi_-}$. This together with \cref{twocasesK} implies 
\begin{equation}
  (A \otimes \identity + \identity \otimes A) K = \begin{cases} \Tr (\mathsf{a}_{1/2}) K\equiv \alpha_3 K & \text{for $n = 3$} \\
    \Tr (\mathsf{a}_0) K\equiv \alpha_4 K & \text{for $n = 4$.}
  \end{cases}
\end{equation}
This implies that $A'=A-\mathbb{I} \alpha_n/2 $ satisfies $(A' \otimes \identity + \identity \otimes A') K=0$. Finally, applying \cref{thm:Z2} we conclude that any $\SU(2)$-invariant Hermitian operator on $n=3, 4$ qubits, up to shift by a multiple of the identity operator, can be written as $\sum_\sigma h_\sigma \P(\sigma)$ satisfying the condition $h_\sigma=-\sgn(\sigma) h^\ast_\sigma$ for all $\sigma\in\mathbb{S}_n$.

\color{black}

\subsection{Generalization of \texorpdfstring{$\mathbb{Z}_2$}{Z2} symmetry and the proof of Lemma \ref{lem13} in the paper}\label{app:genZ2}

Let $\e^{\i \theta}$ be an arbitrary one-dimensional representation of a finite group $G$. Consider a (finite-dimension) unitary representation $R : G \to \U(\hilbert)$ of $G$, where $\U(\hilbert)$ is the group of unitary operators over the Hilbert space $\hilbert$ (in this paper we are mainly focused on the example of symmetric group $ \mathbb{S}_n$ and $\e^{\i \theta}$ is the sign representation). Consider an operator $A$ in the linear span of $\{R(g): g\in G\}$ with the following property: $A$ has a decomposition as $A=\sum_{g\in G} a(g) R(g)$, where $a(g)$ satisfies 
\be\label{equ87}
\forall g\in G:\ \ a(g) = -\e^{\i \theta(g)} a\p[]{g^{-1}} \ .
\ee
For a general group $G$, the map $a(g) \mapsto -\e^{\i \theta(g)} a\p[]{g^{-1}}$ defines a linear transformation on the space of functions over the group $G$, and applying this map twice we obtain the identity map. Therefore, this map generates a $\mathbb{Z}_2$ symmetry, which generalizes the $\mathbb{Z}_2$ symmetry discussed before in the case of $\mathbb{S}_n$. 

As we show in the following, the subset of operators $A\in\mathrm{span}_\mathbb{C}\{R(g):g\in G\}$ that has a decomposition as $A=\sum_g a(g) R(g)$ satisfying \cref{equ87}, forms a Lie algebra. In the special case where $R$ is the regular representation of the group this Lie algebra is previously characterized by Marin \cite{marin2010group}. Note that in this special case there is a one-to-one relation between the operator $A=\sum_g a(g) R(g)$ and function $a$. On the other hand, we are interested in the general case, where this one-to-one relation does not exist. That is, a given operator $A$ in $A\in\mathrm{span}_\mathbb{C}\{R(g):g\in G\}$ may have a decomposition $A=\sum_g a(g) R(g)$ that satisfies this property and other decompositions that do not.

\cref{thm:Z2} below, which is a generalization of \cref{lem13}, provides a simple criterion for testing whether a given operator has such decompositions or not. This criterion is given in terms of a generalization of operator $K$, defined in \cref{eq:genK}. In fact, the theorem applies to arbitrary compact topological groups, e.g., any compact Lie group.

For a one-dimensional representation $\e^{\i \theta}$ of $G$, define
\be\label{eq:genK}
K = \frac{1}{|G|} \sum_{g \in G} \e^{\i \theta(g)} R(g) \otimes R(g)\ ,
\ee
where $|G|$ is the order of group. If $G$ is an arbitrary compact group, then instead we have
\begin{equation}\label{eq:Khaar}
  K = \int \mathrm{d} g \, \e^{\i \theta(g)} R(g) \otimes R(g)
\end{equation}
where $\mathrm{d} g$ is the normalized left-, right-, and inversion-invariant Haar measure. More generally, in the following, any sum over group elements can be replaced by integration over the (suitably normalized) Haar measure. It turns out that many properties of operator $K$ discussed in \cref{App:Kf} for the special case of $\mathbb{S}_n$ can be generalized to arbitrary finite or compact Lie group $G$.

First, it can be easily seen that $K$ is a Hermitian projector. Second,
\begin{equation}\label{eq: 1d rep}
  \forall g\in G:\ \ \ K \p[\big]{R(g) \otimes R(g)} =\p[\big]{R(g) \otimes R(g)}K = \e^{-\i \theta(g)} K\ ,
\end{equation}
by resumming. In other words, $K$ is projector to the subspace of $\mathcal{H}\otimes \mathcal{H}$ that transforms as $\e^{i\theta}$ under representation $R\otimes R$.

This, in particular, implies 
\begin{equation}
  \forall g\in G:\ \ \ \p[\big]{R(g) \otimes R(g)} K \p[\big]{R(g) \otimes R(g)}^\dag = K\ .
\end{equation}
In \cref{SubSec:K}, we discuss more about properties of operator $K$. 

In the following, we prove the following generalization of \cref{lem13} in this more general context:

\begin{theorem}\label{thm:Z2}
  Let $R$ be a finite-dimensional representation of a compact group $G$. Consider an operator $A$ in the linear span of $\{R(g):g\in G\}$. Then 
\begin{equation}\label{eq:z2cond}
  K \p{A \otimes \identity + \identity \otimes A} =\p{A \otimes \identity + \identity \otimes A}K= 0\ ,
\end{equation}
if and only if there exists a function $a(g)$ satisfying
\begin{equation}\label{eq: genz2}
  \forall g\in G:\ a(g) = -\e^{\i \theta(g)} a\p[]{g^{-1}}\ ,
\end{equation}
and $A=\sum_{g\in G} a(g) R(g)$ in the case of finite groups and $A = \int \mathrm{d} g \, a(g) R(g)$ in the case of more general compact topological groups.
\end{theorem}
Note that the set of operators that satisfy the condition of this theorem forms a Lie algebra. In particular if $A_1$ and $A_2$ satisfies $K \p{A_1 \otimes \identity + \identity \otimes A_1}=K \p{A_2 \otimes \identity + \identity \otimes A_2}=0$, then their linear combinations and their commutator also satisfy this property. 

It is also worth noting that, in the case that $A$ is hermitian, the condition in \cref{eq: genz2}) can be rewritten in a slightly different form, which is a generalization of the condition in \cref{assump}: suppose $A$ is a hermitian operator in the linear span of $\cset{R(g)}{g \in G}$. Then the existence of a decomposition as in \cref{eq: genz2} is equivalent to the existence of a decomposition satisfying
\begin{equation}\label{assump2}
  a(g) = -\e^{\i \theta(g)} a^\ast(g).
\end{equation}
To see that these are equivalent, we use a similar trick as in \cref{lem13}, writing $A = (A + A^\dagger) / 2$, so that we have
\begin{equation}\label{eq:Hequiv}
  A = \frac{1}{2} \sum_g a(g) R(g) + \frac{1}{2} \sum_g a^\ast(g) R(g)^\dagger = \sum_g \frac{1}{2} (a(g) + a^\ast(g^{-1})) R(g)
\end{equation}
where we used the fact that $R(g)^\dagger = R(g^{-1})$ because $R$ is a unitary representation and we resummed in the second equality. Let $\tilde{a}(g) = \frac{1}{2} (a(g) + a^\ast(g^{-1}))$. Then clearly $\tilde{a}^\ast(g) = \tilde{a}(g^{-1})$, and so $\tilde{a}$ satisfies \cref{eq: genz2} if and only if it satisfies also \cref{assump2}. It is easy to see by the definition of $\tilde{a}$ that if $a$ satisfies \cref{eq: genz2} then $\tilde{a}$ does also, and likewise if $a$ satisfies \cref{assump2} then so does $\tilde{a}$.

\subsubsection{Decomposing projector \texorpdfstring{$K$}{K} in irreps of group \texorpdfstring{$G$}{G}}\label{SubSec:K}
 
Before presenting the proof of \cref{thm:Z2}, we establish some notation and discuss more about properties of projector $K$. Let $\hat{G}$ be the set of equivalence classes of inequivalent irreps of $G$, and consider the decomposition of $R : G \to \End(\hilbert)$ into inequivalent irreps $\lambda \in \hat{G}$,
\begin{subequations}
  \begin{align}
    \hilbert & \cong \bigoplus_{\lambda \in \Lambda} \hilbert[R]_\lambda \otimes \hilbert[N]_\lambda \\
    R & \cong \sum_\lambda R_\lambda = \sum_\lambda r_\lambda \otimes \identity_{\hilbert[N]_\lambda}\ ,
  \end{align}
\end{subequations}
where $\Lambda \subseteq \hat{G}$ is the set of inequivalent irreps showing up in the representation and $r_\lambda \in \lambda$. It will be helpful to extend the notation $r_\lambda$ to also denote particular representations for each $\lambda \in \hat{G}$, even those which do not show up in the decomposition. In the following, $[r_\lambda]_{ij}$ denotes the matrix elements of $r_\lambda$ in a fixed orthonormal basis $\{|i\rangle: i=1,\cdots, d_\lambda\}$ and $r^\ast_\lambda $ is the complex conjugate of this unitary matrix in that basis. 

In this basis $K$ can be written as 
\begin{equation}\label{eq:Kdecom}
  K = \sum_{\lambda, \mu \in \Lambda} \frac{1}{|G|} \sum_{g \in G} \e^{\i \theta(g)} \p[\big]{r_\lambda(g) \otimes \identity_{\hilbert[N]_\lambda}} \otimes \p[\big]{r_\mu(g) \otimes \identity_{\hilbert[N]_\mu}}= \sum_{\lambda, \mu \in \Lambda} \frac{1}{|G|} \sum_{g \in G} \e^{\i \theta(g)} \p[\big]{r_\lambda(g) \otimes r_\mu(g) } \otimes \p[\big]{\identity_{\hilbert[N]_\mu}\otimes \identity_{\hilbert[N]_\lambda}} .
\end{equation}
To characterize the term $ \frac{1}{|G|} \sum_{g \in G} \e^{\i \theta(g)} \p[\big]{r_\lambda(g) \otimes r_\mu(g) }$, in the following we use the Schur orthogonality relations: 
\begin{equation}\label{ortho67}
  \frac{1}{|G|} \sum_{g \in G} [r_\mu(g)]_{ij} [ r_\nu(g)]_{kl}^\ast =\frac{1}{d_\mu} \delta_{\mu,\nu} \delta_{i,k} \delta_{j,l}\ ,
\end{equation}
where $d_\mu$ is the dimension of the representation $\mu$. 
This can be rewritten as 
\begin{equation}\label{eq:orth}
  \frac{1}{|G|} \sum_{g \in G} r_\mu(g) \otimes r^\ast_{\nu}(g) = \frac{\delta_{\mu \nu}}{d_\mu} \qproj{\Gamma_\mu}\ ,
\end{equation}
where 
\be\label{eq:Gamma}
|\Gamma_\mu\rangle\equiv\sum_{i=1}^{d_\mu} |i\rangle \otimes |i\rangle\in \mathcal{R}_\mu\otimes \mathcal{R}_\mu\ .
\ee
Next, note that for any irrep $r_\lambda$, the complex conjugate $r^\ast_\lambda$ is an irrep of $G$. We denote the equivalency class of this irrep by $\lambda^\ast$ and refer to it as the dual of $\lambda$. Similarly, $\e^{-\i \theta} r^\ast_\lambda$ is also an irrep of $G$. We denote the equivalency class of this irrep by $\lambda^\theta$, and refer to it as the twisted dual of $\lambda$. This definition can be summarized as $\lambda^\theta \equiv \theta^\ast \otimes \lambda^\ast $ (note that $(\lambda^\theta)^\theta=\lambda$). This equivalence means that there exists a unitary transformation $J_\lambda : \hilbert[R]_\lambda \to \hilbert[R]_{\lambda^\theta}$ so that
\begin{equation}\label{eq: Jlambda}
  r_{\lambda^\theta}(g) = \e^{-\i \theta(g)} J_\lambda r_{\lambda}(g)^\ast J_\lambda^\dagger\ .
\end{equation}
Combining this with \cref{eq:orth} we find 
\be
\frac{1}{|G|} \sum_{g \in G} \e^{\i \theta(g)} \p[\big]{r_\lambda(g) \otimes r_\mu(g) }= \frac{\delta_{\lambda, \mu^\theta}}{d_\lambda} (\mathbb{I}\otimes J_\lambda) \qproj{\Gamma_\lambda}(\mathbb{I}\otimes J^\dag_\lambda)\ .
\ee
In other words, this is non-zero iff $\lambda$ is equal to the twisted dual of $\mu$. 

To summarize we conclude that
\begin{equation}\label{eq:K-ent}
  K = \sum_{\lambda\in \Lambda_{\mathrm{sym}}} \frac{1}{d_\lambda} (\mathbb{I}\otimes J_\lambda) \qproj{\Gamma_\lambda}(\mathbb{I}\otimes J^\dag_\lambda) \ ,
\end{equation}
where the summation is over all $\Lambda_\mathrm{sym}=\{\lambda\in\Lambda: \lambda^\theta\in\Lambda\}$, i.e. the set of all irreps of $G$ appearing is representation $R$ with property that their twisted dual also appears in $R$. 

\subsubsection{Proof of Theorem \ref{thm:Z2}} We present the proof in the case of finite groups. The case of compact Lie groups follow in a similar manner (see the end of the proof for a discussion). For the forward direction, use \cref{eq: 1d rep} to write 
\begin{equation}
  \begin{split}
    K \p{A \otimes \identity + \identity \otimes A} & = K \sum_{g \in G} \p{a(g) R(g) \otimes \identity + \identity \otimes a\p[]{g^{-1}} R\p[]{g^{-1}}} \\
    & = K \sum_g \p[\big]{a(g) R(g) \otimes \identity + \e^{\i \theta(g)} a\p[]{g^{-1}} R(g) \otimes \identity} \\
    & = \sum_g \p[\big]{a(g) + \e^{\i \theta(g)} a\p[]{g^{-1}}} K \p[\big]{R(g) \otimes \identity} \\
    & = 0\ .
  \end{split}
\end{equation}
This is a straightforward generalization of the argument in \cref{lem13}. 
 
For the converse, first note that the assumption that $A$ is written as a linear combination of group elements is equivalent to the statement that $A$ acts trivially on the multiplicity subspaces, i.e.,
\begin{equation}\label{eq: A}
  A = \sum_{\lambda \in \Lambda} \Pi_\lambda A \Pi_\lambda = \sum_\lambda A_\lambda = \sum_{\lambda\in\Lambda} \mathsf{a}_\lambda \otimes \identity_{\hilbert[N]_\lambda}.
\end{equation}
Then, combining the assumption $K \p{A \otimes \identity + \identity \otimes A}$ with the decomposition of $K$ in \cref{eq:K-ent} we find that if both $\lambda$ and $\lambda^\theta$ appear in $R$, then
\be\label{maq}
\lambda,\lambda^\theta\in \Lambda\ \implies \ \Big(\mathsf{a}_\lambda\otimes \mathbb{I}_{\lambda^\theta}+ \mathbb{I}_{\lambda}\otimes \mathsf{a}_{\lambda^\theta} \Big) (\mathbb{I}\otimes J_\lambda) |\Gamma_\lambda\rangle=0\ .
\ee
Recall that for any pair of operators $A$ and $B$ on $\mathcal{R}_\lambda$, with the definition \cref{eq:Gamma}, it holds that $(A\otimes \mathbb{I})|\Gamma_\lambda\rangle = (\mathbb{I} \otimes B)|\Gamma_\lambda\rangle $ iff $A=B^T$, where $B^T$ is the transpose of $B$ in the basis $\{|i\rangle: i=1,\cdots, d_\lambda \}$. By applying $(\identity \otimes J_\lambda^\dagger)$ to \cref{maq}, this implies
\begin{equation}\label{eq: Z2}
  \lambda,\lambda^\theta\in \Lambda\ \implies \ \mathsf{a}_\lambda^T = -J_\lambda^\dagger \mathsf{a}_{\lambda^\theta} J_\lambda\ .
\end{equation}
This constraint is relevant if both $\lambda$ and $\lambda^\theta$ are in $\Lambda$, i.e., appear in the representation $R$. Inspired by \cref{eq: Z2}), we extend the definition of $\mathsf{a}_\lambda$ to all irreps of the group $G$, i.e., all $\lambda\in\hat{G}$ using the following rule:
\begin{enumerate}
\item If $\lambda, \lambda^\theta \notin \Lambda$, that is, neither $\lambda$ nor $\lambda^\theta$ show up in the representation $R$, then set $\mathsf{a}_\lambda = \mathsf{a}_{\lambda^\theta} = 0$.
\item If $\lambda\in \Lambda$ but $\lambda^\theta \notin \Lambda$, then set $\mathsf{a}_{\lambda^\theta} = -J_\lambda \mathsf{a}_\lambda^T J_\lambda^\dagger$.
\end{enumerate}
Then, define
\begin{equation}\label{def-func}
  a(g) \equiv \frac{1}{|G|} \sum_{\lambda \in \hat{G}} d_\lambda \Tr (r_\lambda(g)^\dagger \mathsf{a}_\lambda)\ .
\end{equation}
We claim that this function satisfies the desired properties, that is $A=\sum_{g \in G} a(g) R(g)$ and $a(g) = -\e^{\i \theta(g)} a\p[]{g^{-1}}$ for all $g\in G$. The first property follows immediately using Schur orthogonality relations: 
\begin{equation}
  \begin{split}
    \sum_{g \in G} a(g) R(g)&= \sum_{g \in G} \frac{1}{|G|} \sum_{\lambda \in \hat{G}} d_\lambda \Tr(r_\lambda(g)^\dagger \mathsf{a}_\lambda) \sum_{\mu\in\Lambda} r_\mu(g) \otimes \identity_{\hilbert[N]_\mu} \\ &= \sum_{\mu\in\Lambda} \sum_{\lambda \in \hat{G}} \Big[ \sum_{g \in G} \frac{d_\lambda }{|G|} \Tr(r_\lambda(g)^\dagger \mathsf{a}_\lambda) r_\mu(g) \Big]\otimes \identity_{\hilbert[N]_\mu}\\ &= \sum_{\mu\in\Lambda} \sum_{\lambda \in \hat{G}} \delta_{\lambda,\mu} \mathsf{a}_\lambda \otimes \identity_{\hilbert[N]_\mu} \\ &= \sum_{\mu\in\Lambda} \mathsf{a}_\mu \otimes \identity_{\hilbert[N]_\mu}= A\ ,
  \end{split}
\end{equation}
where the second line follows from \cref{ortho67}). To check the second property, note that for all $g\in G$ and $\lambda\in\hat{G}$, it holds that 
\be\label{klkl}
\e^{\i \theta(g)} \Tr(r_{\lambda}(g) \mathsf{a}_\lambda)=\e^{\i \theta(g)} \Tr(r_{\lambda}(g)^T \mathsf{a}_\lambda^T)= \Tr(\e^{\i \theta(g)} J_\lambda r_{\lambda}(g)^T J_\lambda^\dagger J_\lambda\mathsf{a}_\lambda^T J^\dag_\lambda)=- \Tr(r_{\lambda^\theta}(g)^\dag\mathsf{a}_{\lambda^\theta})\ ,
\ee
where we used the facts that $\mathsf{a}_\lambda^T=-J_\lambda^\dagger \mathsf{a}_{\lambda^\theta} J_\lambda$ implies $\mathsf{a}_{\lambda^\theta}=-J_\lambda\mathsf{a}_\lambda^T J^\dag_\lambda$ and $ r_{\lambda^\theta}(g) = \e^{-\i \theta(g)} J_\lambda r_{\lambda}(g)^\ast J_\lambda^\dagger$ implies $ r_{\lambda^\theta}(g)^\dag = \e^{\i \theta(g)} J_\lambda r_{\lambda}(g)^T J_\lambda^\dagger$. Using this identity we find
\begin{equation}
  \label{eq:z2proof}
  \begin{split}
    \e^{\i \theta(g)} a(g^{-1}) & = \frac{1}{|G|} \sum_{\lambda \in \hat{G}} d_\lambda \e^{\i \theta(g)} \Tr (r_\lambda(g^{-1})^\dagger \mathsf{a}_\lambda) \\ & = \frac{1}{|G|} \sum_{\lambda \in \hat{G}} d_\lambda \e^{\i \theta(g)} \Tr (r_\lambda(g) \mathsf{a}_\lambda) \\&=-\frac{1}{|G|} \sum_{\lambda \in \hat{G}} d_\lambda \Tr(r_{\lambda^\theta}(g)^\dag\mathsf{a}_{\lambda^\theta}) \\ & = - \frac{1}{|G|} \sum_{\lambda \in \hat{G}} d_\lambda \Tr( r_\lambda(g)^\dagger \mathsf{a}_\lambda)=-a(g)\ ,
  \end{split}
\end{equation}
where to get the third line we have used \cref{klkl}) and in the last line we resummed and used the fact that the dimension of $\lambda$ is equal to the dimension of $\lambda^\theta$. The last equality follows from definition of function $a(g)$ in \cref{def-func}). 

If, instead of a finite group, $G$ was a general compact topological group (such as a compact Lie group), then the proof goes through the same as before, replacing the sum over group elements with the integral over the Haar measure, e.g., \cref{eq:Khaar}. Since the representation is finite-dimensional, there are only finitely many irreps that show up in the decomposition, i.e., $\Lambda$ is finite; the function $a$ defined in \cref{def-func} can be shown to be $L^2(G)$. This completes the proof of \cref{thm:Z2}.

\newpage 

\section{Time evolution in the fermionic subspace}\label{App:Fermi}

In this section we further expand the ideas presented in \cref{Sec:ferm} on the qudit-fermion correspondence and the group generated by exponentials of fermion swaps.

\subsection{Representation of the symmetric group \texorpdfstring{$\mathbb{S}_n$}{Sn} on a fermionic system}\label{app:fermi-Pf}

In \cref{fermionic group}, we studied a representation of the group $\mathbb{S}_n$ on a fermionic system with $n$ sites. To specify this representation, it suffices to determine the representation of transpositions (swaps). Specifically, for $a,b \in\{1,\cdots, n\}$, we assumed transposition $\sigma_{ab}\in\mathbb{S}_n$ is represented by the operator $ \P^f_{ab} $ that satisfies the two equations 
\begin{align}
  \P^f_{ab} c^\dag_i \P^{f\dagger}_{ab} = c^\dag_{\sigma_{ab}(i)},\quad \text{for}\quad i=1, \cdots, n, \tag{\ref{perm0}}
\end{align}
and
\begin{align}
  \P^f_{ab} |\mathrm{vac}\rangle= |\mathrm{vac}\> \tag{\ref{perm5}}\ .
\end{align}
It can be easily seen that there is a unique operator satisfying both of these constraints. This operator, which turns out to be both Hermitian and unitary, is given by 
\begin{align}
  \P^f_{ab} = \mathbb{I}^f - (c_a^\dagger - c_b^\dagger) (c_a - c_b)\ . \tag{\ref{def:P}} 
\end{align}
The fact that this operator satisfies \cref{perm0,perm5} follows immediately from the fermionic anti-commutation relations
\be
\{c^\dag_i, c^\dag_j\}= 0\ , \ \ \{c^\dag_i, c_j\}= \delta_{ij} \ .
\ee
Now, suppose there is another operator $ \tilde\P^f_{ab}$ that satisfies these equations. Then, \cref{perm0} implies
\begin{align}
  \tilde\P^f_{ab}\P^f_{ab} c^\dag_i \P^{f\dagger}_{ab} \tilde\P^{f\dagger}_{ab} = c^\dag_{\sigma^2_{ab}(i)} = c^\dagger_i.
\end{align}
Taking the adjoint of this equation, we find 
\begin{align}
  \tilde\P^f_{ab}\P^f_{ab} c_i \P^{f\dagger}_{ab} \tilde\P^{f\dagger}_{ab} = c_{\sigma^2_{ab}(i)} = c_i\ .
\end{align}
We conclude that the operator $\tilde\P^f_{ab}\P^f_{ab}$ preserves all creation and annihilation operators and, therefore, preserves the entire algebra generated by them. But this algebra contains all operators acting on the fermionic (Fock) Hilbert space. It follows that $\tilde\P^f_{ab}\P^f_{ab}$ should be equal to the identity operator, up to a phase $\e^{i\phi}$, i.e., $\tilde\P^f_{ab}\P^f_{ab}=\e^{i\phi} \mathbb{I}^f$. Using the fact that operator $\P^f_{ab} = \mathbb{I}^f - (c_a^\dagger - c_b^\dagger) (c_a - c_b)$ satisfies $(\P^f_{ab})^2=\mathbb{I}^f$, we find
\be
 \tilde\P^f_{ab}=\e^{i\phi} \P^f_{ab}\ . 
\ee
Finally, applying both sides of this equation on the vacuum state and using the assumption in \cref{perm5} we find $\e^{i\phi}=1$. We conclude that \cref{perm0} and \cref{perm5} have a unique solution, namely, $ \P^f_{ab}= \mathbb{I}^f - (c_a^\dagger - c_b^\dagger) (c_a - c_b)$.\\

\begin{remark}[Hard-core bosons]
  One can define a similar representation of the permutation group in the case of hard-core bosons (such bosons appear, e.g., in the infinite repulsion limit of the bose-Hubbard model \cite{PhysRevB.40.546}). However, unlike the case of fermions, in this case the swap operators correspond to interacting Hamiltonians: Let $\mathcal{H}^b$ be the $2^n$-dimensional Hilbert space of a system of hard-core bosons living on $n$ sites. For hard-core bosons creation and annihilation operators satisfy the commutation relations 
\be
[b_i, b^\dagger_j] = (1-2n_i)\delta_{ij}\ ,\ [b^\dagger_i, b^\dagger_j] = 0 ,\ b^{\dagger2}_i = 0\ ,
\ee
where $n_i=b^\dag_i b_i$ is the number operator \cite{10.1143/PTP.16.569,Batyev84}. 
Again, to define a representation of $\mathbb{S}_n$, it suffices to define the representation of transpositions (swaps). For a pair of sites $i,j\in\{1,\cdots n\}$, consider the operator $\P^b_{ij}$ satisfying equations 
\begin{align}
  \P^b_{ij} b^\dag_k \P^{b\dagger}_{ij} = b^\dag_{\sigma_{ij}(k)},\quad \text{for}\quad k=1, \cdots, n,
\end{align}
and
\begin{align}
  \P^b_{ij} |\mathrm{vac}\rangle= |\mathrm{vac}\rangle\ ,
\end{align}
where $|\mathrm{vac}\rangle$ is the normalized vacuum state satisfying $b_i |\mathrm{vac}\rangle=0$ for all $i\in\{1,\cdots n\}$. Then, following a similar argument used above, we can see that these equations have a unique solution, namely 
\begin{align}
  \P^b_{ij} = b_i^\dagger b_j + b_j^\dagger b_i + 2(n_i - \frac{\mathbb{I}^b}{2})(n_j - \frac{\mathbb{I}^b}{2}) + \frac{\mathbb{I}^b}{2} ,
\end{align}
where $\mathbb{I}^b$ is the identity operator. The presence of the term $2(n_i - \frac{\mathbb{I}^b}{2})(n_j - \frac{\mathbb{I}^b}{2})$ means that, in contrast to the femionic case, the transposition (swap) operators $\P^b_{ij}$ is no longer of the form of a free Hamiltonian. 
\end{remark}

\begin{remark}[chain of harmonic oscillators]
One can also define a similar representation of the permutation group on a chain of harmonic oscillators. Let $\mathcal{H}^a$ be the Hilbert space of $n$ harmonic oscillators, and $a^\dag_i$ is the raising operator on site $i$, which satisfies the commutation relations
\begin{align}
  [a_i^\dagger, a_j] = \delta_{ij}, \quad [a_i^\dagger, a_j^\dagger] = 0\ .
\end{align}
For any transposition (swap) $\sigma_{ij} \in \mathbb{S}_n$, the corresponding unitary operator $\P^a_{ij}$ on $\mathcal{H}^a$ is defined to have the following properties: for any raising operator $a_k^\dagger$,
\begin{align}
  \P^a_{ij} a^\dag_k \P^{a\dagger}_{ij} = a^\dag_{\sigma_{ij}(k)},\quad \text{for}\quad i=1, \cdots, n\ ,
\end{align}
and
\begin{align}
  \P^a_{ij} \bigotimes_{k=1}^n|0\>_k= \bigotimes_{k=1}^n|0\>_k\ ,
\end{align}
where $|0\>_k$ is the normalized ground state satisfying $a_k |0\>_k=0$. Then, one can see that operator
\begin{align}
  \P^a_{ij} = \sum_{l,m=0}^\infty |m\>_i \<l|_i \otimes |l\>_j \<m|_j,
\end{align}
satisfies both equations, where $|m\>_i = \frac{1}{\sqrt{m!}} (a_i^\dagger)^m |0\>_i$, and we have suppressed the tensor product with the identities on all the other sites.
\end{remark}

\subsection{Unitary time evolution of fermion creation operators under the swap Hamiltonian}\label{app:fermi-conjugate-c}
Recall that fermionic swap operator is equal to $
\P^f_{ab} = \mathbb{I}^f - (c_a^\dagger - c_b^\dagger) (c_a - c_b)\ $. Using the fact that $\P^{f2}_{ab} = \mathbb{I}^f$, we have $\e^{\i\theta\P^f_{ab}} = \cos\theta \mathbb{I}^f + \i\sin\theta \P^f_{ab}$. Therefore
\begin{align}
  \e^{-\i\theta\P^f_{ab}} c_i^\dagger \e^{\i\theta\P^f_{ab}} &= (\cos\theta \mathbb{I}^f - \i\sin\theta \P^f_{ab}) c_i^\dagger (\cos\theta \mathbb{I}^f + \i\sin\theta \P^f_{ab}) \nonumber\\
                                          &= \cos^2\theta c_i^\dagger + \i\cos\theta\sin\theta [c_i^\dagger, \P^f_{ab}] + \sin^2\theta \P^f_{ab} c_i^\dagger \P^f_{ab} \nonumber \\&= \cos^2\theta c_i^\dagger + \i\cos\theta\sin\theta [c_i^\dagger, \P^f_{ab}] + \sin^2\theta c_{\sigma_{ab}(i)}^\dagger \ .
\end{align}
Applying the fermionic anti-commutation relations, we have
\begin{align}
  [c_i^\dagger, \P^f_{ab}] = c_i^\dagger - c_{\sigma_{ab(i)}}^\dagger.
\end{align}
Therefore
\begin{align}
  \e^{-\i\theta\P^f_{ab}} c_i^\dagger \e^{\i\theta\P^f_{ab}} &= \cos^2\theta c_i^\dagger + \i\cos\theta\sin\theta (c_i^\dagger - c_{\sigma_{ab(i)}}^\dagger) + \sin^2\theta c_{\sigma_{ab(i)}}^\dagger \nonumber\\
                                          &= \e^{\i\theta}(\cos\theta c_i^\dagger - \i\sin\theta c_{\sigma_{ab(i)}}^\dagger).
\end{align}

\subsection{In the single-particle sector, the group generated by exponentials of fermionic swaps is isomorphic to \texorpdfstring{$\U(n-1)$}{U(n-1)}}\label{app:fermi-single-fermion}
In \cref{fermionic group}, we argued that the group generated by the exponentials of fermionic swaps $\{\P_{ab}^f\}$ is isomorphic to $\U(n-1)$. The proof relies on the following lemma, which characterizes the action of this group in the single-particle sector. 
\begin{lemma}\label{lem2022}
  Let $\mathfrak{g}_n$ be the Lie algebra generated by skew-Hermitian operators $\{\i (|a\rangle-|b\rangle)(\langle a|-\langle b|): 1\le a<b\le n \}$. Then, 
\be
\mathfrak{g}_n\equiv \big\langle\i (|a\rangle-|b\rangle)(\langle a|-\langle b|): 1\le a<b\le n \big\rangle= \{A\in \mathrm{L}(\mathbb{C}^n): A+A^\dag=0, A\sum_{j=1}^n|j\rangle=0\} \cong \mathfrak{u}(n-1)\ .
\ee
\end{lemma}

To prove this result, we use the following  lemma which is proven in \cite{LHM_2022}.  Suppose a finite-dimensional Hilbert space is decomposed into two orthogonal complementary subspaces. Then, this lemma implies that the family of all block-diagonal unitaries with respect to this decomposition, together with the family of unitaries generated by a Hamiltonian that is not block-diagonal with respect to this decomposition, generate all unitaries. 

\begin{lemma}\label{lemma2}
  Let $\mathcal{H}_\alpha$ and $\mathcal{H}_\beta$ be two complementary orthogonal subspaces of a $D$-dimensional Hilbert space $\mathcal{H}=\mathcal{H}_\alpha\oplus \mathcal{H}_\beta$ with dimensions $D_\alpha$ and $D_\beta$, respectively. Suppose that the dimensions are not $D_\alpha = D_\beta = 1$ or $D_\alpha = D_\beta = 2$. Consider the group of all unitaries that are block-diagonal with respect to this decomposition, and inside each block have determinant 1, i.e., the group 
\be\label{group}
G=\big\{W_\alpha\oplus W_\beta: \det(W_\alpha)=\det(W_\beta)=1 \big\}\ .
\ee
Let $A$ be a non-zero Hermitian operator that is not block-diagonal with respect to the decomposition $\mathcal{H}=\mathcal{H}_\alpha\oplus \mathcal{H}_\beta$. The group generated by unitaries $F=\{\e^{\i A t}: t\in\mathbb{R}\}$ together with unitaries in \cref{group} include all unitaries on $\mathcal{H}$, up to a global phase, i.e., $\SU(D) \subseteq \langle F, G\rangle\ .$ Equivalently, in terms of Lie algebras, this means that $\mathfrak{su}(D)\subseteq \langle \mathfrak{g}, \{ \i A \} \rangle$, where 
 \be
 \mathfrak{g}=\big\{B_\alpha\oplus B_\beta: \ \Tr(B_\alpha)=\Tr(B_\beta)=0\ , B_\alpha+B_\alpha^\dag=0\ , B_\beta+B_\beta^\dag=0 \big\}\ .
 \ee
\end{lemma}

\begin{remark}\label{remark} In the case of $D_\alpha = D_\beta = 1$, the group $G$ as defined above contains just the identity. But if we relax the condition $\det W_\alpha = 1$, then for the one-dimensional group $G' = \cset{\e^{\i \theta} \oplus \identity}{\theta \in \real}$, the theorem holds, i.e., $G'$ together with any group $F = \cset{\e^{\i A t}}{t \in \real}$ (where $A$ satisfies the assumption above, so, in particular, is not block-diagonal) generate a group that contains all unitaries with determinant $1$, $\SU(2) \subseteq \langle F, G' \rangle$ (Equivalently, we could allow the one-parameter group with elements of the form $\e^{\i \theta} \oplus \e^{-\i \theta}$). In terms of the Lie algebra, we have $\mathfrak{g}' = \i \mathbb{R} \qproj{0}$, the one-dimensional Lie algebra generated by the projector $\qproj{0}$ to the space $\hilbert_\alpha$. This is equivalent to the statement that two rotations in nonparallel axes together generate the whole rotation group (up to global phase).
\end{remark}

\color{black}

\begin{proof}(lemma \ref{lem2022})
  To prove this lemma we use induction over $n$. It can be easily seen that the result holds for $n=2$. For $n\ge 2$, we show that the hypothesis
\be
\mathfrak{g}_n=\{A\in \mathrm{L}(\mathbb{C}^n): A+A^\dag=0, A\sum_{j=1}^n|j\rangle=0\} \ ,
\ee
implies
\be\label{trp0}
\mathfrak{g}_{n+1}=\{A\in \mathrm{L}(\mathbb{C}^{n+1}): A+A^\dag=0, A\sum_{j=1}^{n+1}|j\rangle=0\} \ .
\ee
To show this, first decompose 
\be
\mathbb{C}^n\cong \mathcal{H}_\perp^{(n)} \oplus \mathbb{C} |\eta^{(n)}\rangle\ ,
\ee
where 
\be
|\eta^{(n)}\rangle=\frac{1}{\sqrt{n}}\sum_{j=1}^{n} |j\rangle\ ,
\ee
and 
\be
\mathcal{H}_\perp^{(n)}=\big\{\sum_{j=1}^n \alpha_j |j\rangle: \sum_j \alpha_j=0\big\}\ ,
\ee
is the subspace orthogonal to $|\eta^{(n)}\rangle$. By the induction hypthothesis any skew-Hermitian operator $\tilde{A}$ acting on $\mathcal{H}_\perp^{(n)} $ is in $\mathfrak{g}_{n}$.

Next, decompose $\mathbb{C}^{n+1}$ to 3 orthogonal subspaces as 
\be
\mathbb{C}^{n+1}\cong \mathbb{C}^{n}\oplus \mathbb{C}|n+1\rangle \cong
\big(\mathcal{H}_\perp^{(n)} \oplus \mathbb{C} |\eta^{(n)}\rangle\big)\oplus \mathbb{C} |n+1\rangle
 \cong \mathcal{H}_\perp^{(n)} \oplus \mathbb{C} |\beta^{(n+1)}\rangle \oplus\mathbb{C} |\eta^{(n+1)}\rangle\ ,
\ee
where
\be
|\eta^{(n+1)}\rangle=\frac{1}{\sqrt{n+1}}\sum_{j=1}^{n+1} |j\rangle\ =\sqrt{\frac{n}{n+1}}|\eta^{(n)}\rangle+\sqrt{\frac{1}{n+1}}|n+1\rangle \ .
\ee
and
\be
|\beta^{(n+1)}\rangle=\sqrt{\frac{1}{n+1}}|\eta^{(n)}\rangle-\sqrt{\frac{n}{n+1}}|n+1\rangle\ .
\ee
 By assumption, the Lie algebra $\mathfrak{g}_{n+1}$ contains
\be
\i(\mathbb{I}-E_{n,n+1})=\i(|n\rangle-|n+1\rangle)(\langle n|- \langle n+1|)\ ,
\ee
where we have used the notation $
E_{ab}\equiv\mathbb{I}- (|a\rangle-|b\rangle)(\langle a|-\langle b|)$ introduced before. 

Since $|n\rangle-|n+1\rangle$ is orthogonal to $|\eta^{(n+1)}\rangle$, it can be written as 
\be
|n\rangle-|n+1\rangle=\sqrt{\frac{n+1}{n}} |\beta^{(n+1)}\rangle+ \Big(\frac{n-1}{n}|n\rangle-\frac{1}{n}\sum_{j=1}^{n-1}|j\rangle\Big)=\sqrt{\frac{n+1}{n}} |\beta^{(n+1)}\rangle+|\gamma_\perp\rangle \ ,
\ee
where 
\be
|\gamma_\perp\rangle=\frac{n-1}{n}|n\rangle-\frac{1}{n}\sum_{j=1}^{n-1}|j\rangle\in \mathcal{H}_\perp^{(n)} \ ,
\ee
is a non-zero vector. This means
\be
\i(\mathbb{I}-E_{n,n+1})=\i\Big[\sqrt{\frac{n+1}{n}} |\beta^{(n)}\rangle+ |\gamma_\perp\rangle\Big]\Big[\sqrt{\frac{n+1}{n}} \langle\beta^{(n)}|+ \langle\gamma_\perp|\Big]\ .
\ee
Clearly, $\i(\mathbb{I}-E_{n,n+1}) |\eta^{(n+1)}\rangle=0$. Therefore, the support of $\i(\mathbb{I}-E_{n,n+1}) $ is restricted to the $n$-dimensional subspace $\mathcal{H}_\perp^{(n)} \oplus \mathbb{C} |\beta^{(n+1)}\rangle $. Furthermore, because $|\beta^{(n)}\rangle\langle\gamma_\perp| \neq 0 $, the operator $\i(\mathbb{I}-E_{n,n+1})$ is not block-diagonal with respect to the decomposition $\mathcal{H}_\perp^{(n)} \oplus \mathbb{C} |\beta^{(n+1)}\rangle $. Recall that by induction hypothesis $\mathfrak{g}_{n+1}$ also contains all skew-Hermitian operators with support restricted to $\mathcal{H}_\perp^{(n)}$. Therefore, applying \cref{lemma2} for decomposition $\mathcal{H}_\perp^{(n)} \oplus \mathbb{C} |\beta^{(n+1)}\rangle $ we conclude that $\mathfrak{g}_{n+1}$ contains all traceless skew Hermitian operators with support restricted to $\mathcal{H}_\perp^{(n)} \oplus \mathbb{C} |\beta^{(n+1)}\rangle $, which is isomorphic to $\mathfrak{su}(n)$ (in the case of going from $n = 2$ to $n + 1 = 3$, we use the modified version of the lemma as in \cref{remark}). Furthermore, $\mathfrak{g}_{n+1}$ also contains operators with support restricted to $\mathcal{H}_\perp^{(n)} \oplus \mathbb{C} |\beta^{(n+1)}\rangle $ with non-zero trace (for instance, $\i(\mathbb{I}-E_{n,n+1})\in\mathfrak{g}_{n+1}$ has trace $2i$). The linear combination of traceless skew-Hermitian operators with a skew-Hermitian operator with non-zero trace yields all skew-Hermitian operators. This implies \cref{trp0} and completes the proof of the lemma via induction. 
\end{proof}

\subsection{Equivalence of wedge representations in Eqs. (\ref{eq:equiv}) and (\ref{eq:equivwedge})}\label{app:wedge}

Here, we further explain the equivalences in \cref{eq:equiv,eq:equivwedge}. Recall that $G_{\mathrm{single}} = \cset{\identity \oplus T}{T \in \U(n - 1)}$ with respect to a particular decomposition of $\complex^n$. Call the state on which $G_{\mathrm{single}}$ acts as identity $\nket{0}$, and let the orthogonal subspace be $\hilbert_\perp$, with basis $\sets{\nket{e}}_{e = 1}^{n - 1}$. Then the equivalences say that $\bigwedge^L \p{\complex \nket{0} \oplus \hilbert_\perp} \cong \bigwedge^L \hilbert_\perp \oplus \bigwedge^{L - 1} \hilbert_\perp$.

For any $L$ basis states $\sets{\nket{e_i}}_{i = 1}^L$, $0 \leq e_i \leq n-1$, consider their wedge product
\begin{equation}\label{eq:wedges}
  \bigwedge_{j = 1}^L \nket{e_j} = \frac{1}{\sqrt{L!}} \sum_{\sigma\in \mathbb{S}_L} \sgn(\sigma) \nket{e_{\sigma(1)}} \otimes \cdots \otimes \nket{e_{\sigma(L)}}.
\end{equation}
This vanishes if and only if for some $i \neq j$ we have $e_i = e_j$. Note that, up to a sign, this does not depend on the ordering of the basis states $\sets{e_j}$. In fact, because any reordering is just a permutation $\sigma \in \mathbb{S}_L$, the wedge products defined by these different orderings are the same up to $\sgn(\sigma)$. Since the nonzero (distinct) states of the form \cref{eq:wedges} form a basis for the wedge product space $\bigwedge^L (\complex \nket{0} \oplus \hilbert_\perp)$, we see that this space can be decomposed into two subspaces: one spanned by states of the form \cref{eq:wedges} where none of the $e_i = 0$, and another where exactly one $e_i = 0$. Note that the first is exactly $\bigwedge^L \hilbert_\perp$ and the second is isomorphic to $\bigwedge^{L - 1} \hilbert_\perp$. (One can also verify that dimensions match---this is just Pascal's rule.) Further, $G_{\mathrm{single}}$ acts on these as $(\identity \oplus T)^{\otimes L} \bigwedge_{i = 1}^L \nket{e_i}$, from which it is easy to verify \cref{eq:equiv,eq:equivwedge}.

\subsection{Linearity and unitarity of \texorpdfstring{$U^f$}{Uf}}\label{app:fermi-Uf}
In \cref{eq:perm} we claim that there exists a linear map $U^f$ that preserves the inner products, i.e., is an isometry, and satisfies 
\begin{align}
  \forall\sigma \in\mathbb{S}_n:\ U^f \P(\sigma) \Big[\big(\bigwedge_{m=1}^{L}|m\>\big) \otimes |0\>^{\otimes (n-L)}\Big] = \prod_{i=1}^{L} c_{\sigma(i)}^\dag |\mathrm{vac}\>\ .\tag{\ref{eq:perm}}
\end{align}
To prove this claim it suffices to show that the pairwise inner products of the two sides are equal, i.e.,
\be\label{Gram}
\forall \sigma_{1,2}\in\mathbb{S}_n:\ \ \ \<\Phi| \P(\sigma_2)^\dagger \P(\sigma_1) |\Phi\> = \<\mathrm{vac}| \Bigl(\prod_{i=1}^{L} c_{\sigma_2(i)}^\dagger\Bigr)^\dagger \prod_{i=1}^{L} c_{\sigma_1(i)}^\dag |\mathrm{vac}\>\ ,
\ee
where $ |\Phi\rangle = \big[\big(\bigwedge_{m=1}^{L}|m\>\big) \otimes |0\>^{\otimes (n-L)}\big]$. 
To show this first note that, using the anti-commutation relation of the fermion creation operators, we have
\be
\<\mathrm{vac}| (\prod_{i=1}^{L} c_{\sigma_2(i)}^\dagger)^\dagger \prod_{i=1}^{L} c_{\sigma_1(i)}^\dag |\mathrm{vac}\>\neq 0\ ,
\ee
 if and only if, for each $i\in\{1,\cdots, L\}$, there exists $j\in\{1,\cdots, L\}$ such that 
 \be
 \sigma_1(i)=\sigma_2(j) \ ,
 \ee
 or, equivalently, 
 \be
\forall i=1,\cdots, L :\ \ \sigma^{-1}_2\sigma_1(i)\in \{1,\cdots L\} \ .
 \ee
This means that $ \tau=\sigma^{-1}_2\sigma_1$ permutes the first $L$ elements $\{1,\cdots, L\}$ among themselves. Therefore, $\tau$ has a unique decomposition as 
\be
\tau=\tau_L\ \tau_{n-L}\ ,
\ee
where $\tau_L: \{1,\cdots, L\}\rightarrow \{1,\cdots, L\}$ is a permutation defined on the first $L$ elements and $\tau_{n-L}: \{L+1,\cdots, n\}\rightarrow \{L+1,\cdots, n\}$ is permutation defined on the rest of elements. Then, using the anti-commutation relations of fermions we can easily see that
\begin{align}
  \<\mathrm{vac}| \Big(\prod_{i=1}^{L} c_{\sigma_2(i)}^\dagger\Big)^\dagger \prod_{i=1}^{L} c_{\sigma_1(i)}^\dag |\mathrm{vac}\> = \sgn(\tau_L)\ .
\end{align}

Next, we consider the inner products on the left-hand side of \cref{Gram}. For arbitrary $\sigma_1, \sigma_2 \in \mathbb{S}_n$, we have
\begin{align*}
  \<\Phi| \P(\sigma_2)^\dagger \P(\sigma_1) |\Phi\> = \<\Phi| \P(\sigma_2^{-1}\sigma_1) |\Phi\>= \Big[\big(\bigwedge_{m=1}^{L}\langle m|\big) \otimes \langle0|^{\otimes (n-L)}\Big]\P(\sigma_2^{-1}\sigma_1)\Big[\big(\bigwedge_{m=1}^{L}|m\>\big) \otimes |0\>^{\otimes (n-L)}\Big]
\end{align*}
Again, this inner product is non-zero if and only if $\sigma_2^{-1}\sigma_1$ permutes the first $L$ elements among themselves. Therefore $\tau=\sigma_2^{-1}\sigma_1 $ has a unique decomposition as $\tau_L \tau_{n-L}$, where $\tau_L$ is a permutation on the first $L$ elements, and $\tau_{n-L}$ permutes among the rest. With this notation, when the inner product is non-zero, then
\begin{align}
  \<\Phi| \P(\sigma_2)^\dagger \P(\sigma_1) |\Phi\> = \sgn(\tau_L)\ .
\end{align}
We conclude that \cref{Gram} holds for all $\sigma_{1,2}\in\mathbb{S}_n$ and, therefore, there exists an isometry $U^f$ satisfying \cref{eq:perm}.

\subsection{An example of the one-particle reduced state}\label{app:fermi-1p-example}
Next, we present an example of the single-particle reduced state, obtained by applying map $\Omega$, defined in \cref{matrix}. Consider the $n$-qudit state
\be
|\psi\>=\sum_{j=1}^n \psi_j |0\>^{\otimes (j-1)} |1\> |0\>^{\otimes (n-j)}\ .
\ee
For $i\neq j$ we have
\begin{align}
  \Omega_{ij}[\psi] &= \<\psi |\Pi_\comp \P_{ij} Q_{ij} \Pi_\comp |\psi\> = \psi_i \<\psi |\P_{ij} \big(|0\>^{\otimes (i-1)} |1\> |0\>^{\otimes (n-i)}\big) = \psi_i \<\psi |0\>^{\otimes (j-1)} |1\> |0\>^{\otimes (n-j)} = \psi_i\psi_j^*\ ,
\end{align}
where the second equality follows from the fact that $|\psi\>$ is already in the fermionic subspace and $Q_{ij}$ is a projector to the subspace where $i$ is non-zero and $j$ is zero. Similarly for $i = j$ we have
\begin{align}
  \Omega_{ii}[\psi] &= \<\psi |\Pi_\comp (\mathbb{I}_i-|0\>\< 0|_i) \Pi_\comp |\psi\> = \psi_i\psi_i^*,
\end{align}
Therefore we have $\Omega[\psi] = \sum_{ij} \Omega_{ij}[\psi] |i\>\< j| = \sum_{ij} \psi_i \psi_j^* |i\>\< j|$.

\subsection{Fermion hopping in the qudit language (the equivalence of the two representations of the map \texorpdfstring{$\Omega$}{Omega})}\label{app:fermi-equiv}

Next, we show the equivalence of the definitions of map $\Omega$ in \cref{eq:Omega_fermion,single-particle}, that is
\begin{align}\tag{\ref{eq:Omega_fermion}}
  \Omega_{ij}[\rho] = \Tr(c^\dag_j c_i U^f\Pi_\comp\rho\Pi_\comp U^{f\dagger}),
\end{align}
and
\begin{align}\tag{\ref{single-particle}}
  \Omega_{ij}[\rho] =
  \begin{cases}
    \Tr\big(\Pi_{\comp} \rho\Pi_{\comp}\ [\P_{ij} Q_{ij}]\ \big) &: i\neq j, \\
    \Tr\big(\Pi_{\comp} \rho\Pi_{\comp}\ [\mathbb{I}_i-|0\rangle\langle 0|_i]\big) &: i=j,
  \end{cases}
\end{align}
where $\rho\in \End((\mathbb{C}^d)^{\otimes n})$ and $Q_{ij} = (\mathbb{I}_i - |0\>\<0|_i) |0\>\<0|_j$.
In order to show the equivalence of these two definitions, we simply need to verify that for all $i\neq j\in \{1,\cdots, n\}$ and all $k\in \{1,\cdots, n\}$, it holds that
\begin{align}
  \Pi_\comp \P_{ij} Q_{ij} \Pi_\comp &= \Pi_\comp U^{f\dagger} c^\dagger_jc_i U^f\Pi_\comp, \label{eq:fermion_qudit_hop} \\
  \Pi_\comp (\mathbb{I}_k - |0\>\<0|_k) \Pi_\comp &= \Pi_\comp U^{f\dagger} c_k^\dagger c_k U^f \Pi_\comp\ . \label{eq:fermion_qudit_n}
\end{align} 
Recall that $\mathcal{H}_\comp$ is spanned by vectors $ \P(\sigma) |\Phi\rangle: \sigma\in\mathbb{S}_n$, where 
\be
|\Phi\rangle\equiv\big(\bigwedge_{m=1}^{L}|m\>\big) \otimes |0\>^{\otimes (n-L)}\ .
\ee
Therefore, to show that $\Pi_\comp \P_{ij} Q_{ij} \Pi_\comp = \Pi_\comp U^{f\dagger} c^\dagger_jc_i U^f\Pi_\comp$ it suffices to show the equality of the matrix elements of both sides in this basis, i.e., 
\be
\forall \sigma_1,\sigma_2\in\mathbb{S}_n:\ \ \ \langle\Phi|\P^\dag(\sigma_2) \Pi_\comp \P_{ij} Q_{ij} \Pi_\comp \P(\sigma_1) |\Phi\rangle=\langle\Phi|
 \P^\dag(\sigma_2) \Pi_\comp U^{f\dagger} c^\dagger_jc_i U^f\Pi_\comp \P(\sigma_1) |\Phi\rangle\ ,
\ee
or, equivalently
\be
 \forall \sigma_1,\sigma_2\in\mathbb{S}_n:\ \ \ \langle\Phi| \P^\dagger(\sigma_2\sigma^{-1}_1) \P_{\sigma_1^{-1}(i)\sigma_1^{-1}(j)} Q_{\sigma_1^{-1}(i)\sigma_1^{-1}(j)} |\Phi\rangle=\langle\Phi^f| \P^{f\dagger}(\sigma_2\sigma^{-1}_1) c^\dagger_{\sigma_1^{-1}(j)} c_{\sigma_1^{-1}(i)} |\Phi^f\rangle\ ,
\ee
where 
\be
|\Phi^f\rangle=U^f|\Phi\rangle=\prod_{j=1}^L c_j^\dag|\mathrm{vac}\rangle\ .
\ee
To show this it suffices to show that for all $i\neq j\in\{1,\cdots, n\}$ and all $\sigma\in \mathbb{S}_n$, it holds that
\begin{align}\label{eq:fermion_qudit_hop_exp}
  \<\Phi| \P^\dagger(\sigma) \P_{ij} Q_{ij} |\Phi\>=\<\Phi^f| \P^{f\dagger}(\sigma) c^\dagger_{j} c_{i} |\Phi^f\> \ .
\end{align}
Similarly, to show \cref{eq:fermion_qudit_n} we need to show 
\begin{align}\label{jkhf3}
  \<\Phi| \P^\dagger(\sigma) [\mathbb{I}_i-|0\rangle\langle 0|_i] |\Phi\>=\<\Phi^f| \P^{f\dagger}(\sigma) c^\dagger_{i} c_{i} |\Phi^f\> \ .
\end{align}

Recall that $Q_{ij} = (\mathbb{I}_i - |0\>\<0|_i) |0\>\<0|_j$ is a projector which projects the state on site $i$ to one orthogonal to $|0\>_i$, while projects the state on site $j$ to $|0\>_j$. This means 
\be
Q_{ij}|\Phi\rangle= \big[(\mathbb{I}_i - |0\>\<0|_i) |0\>\<0|_j\big]\big(\bigwedge_{m=1}^{L}|m\>\big) \otimes |0\>^{\otimes (n-L)}=
\begin{cases}
  |\Phi\rangle &: i\le L <j \\
  0 &: \text{otherwise}\ ,
\end{cases} 
\ee
which means the left-hand side of \cref{eq:fermion_qudit_hop_exp} is equal to
\be\label{wlkef}
\<\Phi| \P^\dagger(\sigma) \P_{ij} Q_{ij} |\Phi\>=\begin{cases}
  \langle\Phi|\P(\sigma^{-1} \sigma_{ij})|\Phi\rangle &: i\le L <j \\
  0 &: \text{otherwise}\ .
\end{cases} 
\ee
Similarly, the left-hand side of \cref{jkhf3} can be written as 
\be\label{wlkef1}
\<\Phi| \P^\dagger(\sigma) [\mathbb{I}_i-|0\rangle\langle 0|_i] |\Phi\>=\begin{cases}
  \langle\Phi|\P(\sigma^{-1})|\Phi\rangle &: i\le L \\
  0 &: \text{otherwise}\ .
\end{cases} 
\ee
On the other hand, using $\P^f_{ij} = \mathbb{I}^f - (c_i^\dagger - c_j^\dagger) (c_i - c_j) $, we find 
\be\label{wlkef2}
c_j^\dag c_i|\Phi^f\rangle=c_j^\dag c_i\prod_{l=1}^L c_l^\dag|\mathrm{vac}\rangle =\begin{cases}
  \P_{ij}^f |\Phi^f\rangle &: i\le L < j \\ |\Phi^f\rangle &: i=j\le L\\ 
  0 &: \text{otherwise}\ ,
\end{cases}
\ee
which means the right-hand side of \cref{eq:fermion_qudit_hop_exp} and \cref{jkhf3} can be combined together as 
\be
\langle\Phi^f|\P^{f\dagger}(\sigma) c^\dagger_{j} c_{i} |\Phi^f\rangle=\begin{cases}
  \langle\Phi^f|\P^f(\sigma^{-1} \sigma_{ij})|\Phi^f\rangle &: i\le L <j \\ \langle\Phi^f|\P^f(\sigma^{-1})|\Phi^f\rangle &: i=j\le L \\
  0 &: \text{otherwise}\ .
\end{cases} 
\ee
Finally, recall that
$\P^f(\sigma)|\Phi^f\rangle=U^f \P(\sigma)|\Phi\rangle$, and $U^f$ preserves inner products. This means \be
\forall \sigma\in\mathbb{S}_n:\ \ \langle\Phi|\P(\sigma)|\Phi\rangle=\langle\Phi^f|\P^f(\sigma)|\Phi^f\rangle\ .
\ee
Combining this with \cref{wlkef,wlkef1,wlkef2}, we conclude that the two representations of map $\Omega$ in \cref{eq:Omega_fermion,single-particle} are indeed equivalent.

\subsection{Complete-positivity and covariance of the map \texorpdfstring{$\Omega$}{Omega} (proof of Lemma \ref{lem3})}\label{app:fermi-1p-reduced}
Next, we prove \cref{lem3}. That is we show that map $\Omega$ is completely positive and satisfies the covariance relation $\Omega[\e^{\i\theta\P_{ab}}\rho\e^{-\i\theta\P_{ab}}]=\e^{\i \theta E_{ab}} \Omega[\rho] \e^{-\i \theta E_{ab}}$ for all $\theta\in[0,2\pi)$. Both of these properties follow from the observation that
\begin{align}\label{partial65}
  \Omega[\rho] &= \Tr_f (S U^f \Pi_\comp \rho \Pi_\comp U^{f\dagger} S^\dagger) \ ,
\end{align}
where $S: \H^f \rightarrow \H^f \otimes \mathbb{C}^n$ is defined by
\begin{align}
  S=\sum_{i=1}^n c_i \otimes |i\> \ , 
\end{align}
and the partial trace in \cref{partial65} is over the fermionic subsystem $\H^f$. 

Recall that the partial trace is a completely-positive map. Furthermore, for any operator $A$, the map $\rho\rightarrow A\rho A^\dag$ is also completely positive. It follows that $\Omega$ is the concatenation of two completely-positive maps and hence is completely positive \cite{wilde2013quantum}. 

Next, we note that \cref{rep:fermi2} can be rewritten as 
\begin{align}\label{inv98}
 S\e^{\i\theta \P^f_{ab}} = (\e^{\i\theta \P^f_{ab}}\otimes \e^{\i \theta (E_{ab}-\mathbb{I})})S\ .
\end{align}
This implies that for all $\theta\in[0,2\pi)$, it holds that 
\begin{align}
  \Omega[\e^{\i\theta\P_{ab}}\rho\e^{-\i\theta\P_{ab}}] &= \Tr_f (S U^f \Pi_\comp \e^{\i\theta\P_{ab}}\rho\e^{-\i\theta\P_{ab}} \Pi_\comp U^{f\dagger} S^\dagger) \nonumber\\
                                      &= \Tr_f (S \e^{\i\theta\P^f_{ab}} U^f \Pi_\comp\rho \Pi_\comp U^{f\dagger}\e^{-\i\theta\P^f_{ab}} S^\dagger) \nonumber\\ 
                                      &= \Tr_f (\e^{\i\theta \P^f_{ab}}\otimes \e^{\i \theta (E_{ab}-\mathbb{I})} S U^f \Pi_\comp\rho \Pi_\comp U^{f\dagger} S^\dagger \e^{-\i\theta \P^f_{ab}}\otimes \e^{-\i \theta (E_{ab}-\mathbb{I})}) \nonumber\\ 
                                      &= \e^{\i \theta (E_{ab}-\mathbb{I})} \Tr_f (S U^f \Pi_\comp\rho \Pi_\comp U^{f\dagger} S^\dagger) \e^{-\i \theta (E_{ab}-\mathbb{I})} \nonumber\\ 
                                      &= \e^{\i \theta E_{ab}} \Omega[\rho] \e^{-\i \theta E_{ab}}\ ,
\end{align}
where to get the second line we have used the fact that $\e^{\i\theta \P^f_{ab}}$ commutes with $\Pi_\comp$ and $U^f$ and to get the third line we have used \cref{inv98}. This proves the covariance of $\Omega$ and completes the proof of \cref{lem3}.

\newpage
\section{Constraints on relative phases between sectors with different charges} \label{app:relative-phases}

Recall that $\mathfrak{z}_{k}$ denotes the center of $\mathfrak{v}_{k}$, the Lie algebra generated by $k$-local $\SU(d)$-invariant anti-Hermitian operators on $(\mathbb{C}^d)^{\otimes n}$. As we discuss more in \cite{HLM_2021}, the general results of \cite{marvian2022restrictions} implies that the dimension of $\mathfrak{z}_{k}$ is equal to the number of inequivalent irreps of $\SU(d)$ that appear on $k$ qudits. Equivalently, it is the number of Young diagrams with $k$ boxes and with, at most, $d$ rows. In the special case of $k=2$, this means that the dimension of the center is 2D. In the following we present a characterization of this 2D space. We note that essentially an equivalent characterization in terms of characters can also be found in the work of Marin in \cite{marin2007algebre}. 

First, note that since $\mathfrak{v}_{2}$ contains all swaps up to a phase, then any operator in the center should be permutationally invariant and should be in the linear span of permutations operators $\{\mathbf{P}(\sigma): \sigma\in\mathbb{S}_n\}$ . It follows that any element of $\mathfrak{z}_{2}$ can be written as a linear combination of projectors to sectors with different charges, that is 
\begin{equation}
  \mathfrak{z}_{2} = \cset{\i \sum_\lambda (\alpha + \beta b_\lambda) \Pi_\lambda}{\alpha, \beta \in \real}\ ,
\end{equation}
where $\Pi_\lambda$ is the projector to sector with irrep $\lambda$ of $(\mathbb{C}^d)^{\otimes n}$. In the following, we find several formulas for coefficients $\{b_\lambda\}$. 

First, note that using the charge vector terminology in \cite{marvian2022restrictions}, the charge vectors are
\begin{equation}
  \{\sum_\lambda [\alpha+\beta b_\lambda] \Tr(\Pi_\lambda) |\lambda\rangle: \alpha,\beta\in\mathbb{R} \}\ .
\end{equation}
where 
\be
b_\lambda=\frac{\Tr(\P_{rs}\Pi_\lambda)}{\Tr(\Pi_\lambda) } \ ,
\ee
for any pair of $r\neq s\in\{1,\cdots , n\}$ (here $\sets{\nket{\lambda}}$ denotes an orthonormal basis for an abstract vector space). The fact that $\Tr(\P_{rs}\Pi_\lambda)$ is independent of $r$ and $s$, follows immediately from the permutational symmetry of $\Pi_\lambda$. This, in particular, means that
\be
\Tr(\P_{rs}\Pi_\lambda)=\frac{1}{n(n-1)}\sum_{r\neq s} \Tr(\P_{rs}\Pi_\lambda)=\Tr(Z \Pi_\lambda) \ ,
\ee
where 
\be
 Z\equiv\sum_{r\neq s} \P_{rs}\ . 
 \ee
Note that $Z$ is (by definition) permutationally-invariant and since it is in the linear span of $\{\P_{rs}\}$, we have $iZ\in\mathfrak{z}_{2}$, and $Z=\sum_\lambda z_\lambda \Pi_\lambda$, where $z_\lambda$ is the eigenvalue of $Z$ in sector $\Pi_\lambda$. To summarize, we have
\begin{equation}\label{lmg01}
  \mathfrak{z}_{2} = \cset{\i \sum_\lambda (\alpha + \beta b_\lambda) \Pi_\lambda}{\alpha, \beta \in \real}= \mathrm{span}_\mathbb{C}\{i\mathbb{I}, i Z\}\ ,
\end{equation}
and 
\be\label{lmg02}
b_\lambda= \frac{\Tr(\P_{rs}\Pi_\lambda)}{\Tr(\Pi_\lambda) }=\frac{1}{n(n-1)} \frac{\Tr(Z\Pi_\lambda)}{\Tr(\Pi_\lambda) }=\frac{z_\lambda}{n(n-1)} \ .
\ee
It turns out that operator $Z$ is closely related to the quadratic Casimir operator. To explain this connection, we present the following lemma, which is of independent interest (see the end of this section for the proof).
\begin{lemma}\label{lem:swap}
  Let $\sets{T^a: a=1, \cdots, d^2-1}$ be a set of $d\times d$ traceless hermitian operators satisfying the orthogonality relation $\Tr(T^{a \dagger} T^b) = \Tr(T^a T^b) = 2\delta^{ab}$. Then, for a pair of qudits, labeled as 1 and 2, the swap operator $\mathbf{P}_{12}$ can be written as 
\begin{align}
  \mathbf{P}_{12} = \frac{1}{2} \sum_{a = 1}^{d^2 - 1} T^a \otimes T^a + \frac{1}{d}\mathbb{I}\ .
\end{align}
where $\mathbb{I}$ is the identity operator.
\end{lemma}
We also remark that using a similar argument one can show the identity 
\be\label{kjhefowr}
\frac{1}{2} \sum_{a = 1}^{d^2 - 1} (T^a)^2 = \frac{d^2 - 1}{d} \identity \ ,
\ee
where the left-hand side is the quadratic Casimir operator associated to the fundamental representation of $\SU(d)$. This Casimir operator can be seen to be $\SU(d)$-invariant by a similar argument to \cref{eq:C-inv} (in the proof of the above lemma), and is therefore proportional to the identity by Schur's Lemma. The proportionality constant is determined by taking traces and using $\Tr((T^a)^2) = 2$.

Applying \cref{lem:swap} for a system with $n$ qudits we can write the swap operator $\P_{ij}$ as
\be
\P_{ij}= \frac{1}{2} \sum_{a = 1}^{d^2 - 1} T_i^a T_j^a + \frac{1}{d}\mathbb{I}\ ,
\ee
 where $T_i^a $ is operator $T^a$ on qudit $i$, tensor product with the identity operator on the rest of qudits. 
It is also useful to define the quadratic Casimir operator
\be
C_2 \equiv \frac{1}{2}\sum_{a=1}^{d^2-1} \big(\sum_{i=1}^n T^a_i\big)^2\ .
\ee
 Then, using this definition together with the identity in \cref{kjhefowr}, we find 
 \be
\sum_{i\neq j} \sum_{a = 1}^{d^2 - 1} T_i^a T_j^a= \sum_{i, j} \sum_{a = 1}^{d^2 - 1} T_i^a T_j^a-\sum_i \sum_{a = 1}^{d^2 - 1} (T_i^a)^2=2 C_2- 2n \frac{d^2-1}{d}\mathbb{I}\ .
\ee
Therefore, we conclude that 
\begin{equation}
  \begin{split}
    Z &= \sum_{i \neq j} \P_{ij} = \sum_{i \neq j} (\frac{1}{2} \sum_{a = 1}^{d^2 - 1} T_i^a T_j^a + \frac{1}{d}\mathbb{I})= C_2- 2n \frac{d^2-1}{d}\mathbb{I} + \frac{n(n-1)}{d}\mathbb{I}
    =C_2+ \frac{n(n - d^2)}{d}\mathbb{I}\ .
  \end{split}
\end{equation}
Hence, 
\be
z_\lambda=c_\lambda+\frac{n(n - d^2)}{d}\ ,
\ee
where $c_\lambda$ is the eigenvalues of $C_2$ in sector $\Pi_\lambda$. 

Combining this result with \cref{lmg01,lmg02}, we conclude that
\begin{equation}
  \mathfrak{z}_{2} = \cset{\i \sum_\lambda (\alpha + \beta b_\lambda) \Pi_\lambda}{\alpha, \beta \in \real}= \mathrm{span}_\mathbb{C}\{i\mathbb{I}, i Z\}=\mathrm{span}_\mathbb{C}\{i\mathbb{I}, i C\}\ ,
\end{equation}
and 
\be
b_\lambda= \frac{\Tr(\P_{rs}\Pi_\lambda)}{\Tr(\Pi_\lambda) }=\frac{1}{n(n-1)} \frac{\Tr(Z\Pi_\lambda)}{\Tr(\Pi_\lambda) }=\frac{z_\lambda}{n(n-1)}=\frac{d c_\lambda+n(n - d^2) }{d n(n-1)}\ . 
\ee
In particular, in the case of $d = 2$, the quadratic Casimir is $C_2 = 2 J^2$, where $J^2 = J^2_x+J^2_y+J_z^2$ is the total squared angular momentum operator. Thus the eigenvalue associated to the irrep with definite angular momentum $j$ is $c_j = 2 j (j + 1)$. In this case, the eigenvalue $c_j$ completely specifies the irrep of $\SU(2)$.

Finally, it is worth noting that $b_\lambda$ can also be expressed in terms of the character of the irrep $\lambda$ of $\mathbb{S}_n$. In particular, using Schur-Weyl duality, we know that sector $\mathcal{H}_\lambda$ decomposed as $\mathcal{H}_\lambda=\mathcal{Q}_\lambda\otimes \mathcal{M}_\lambda$, where $\SU(d)$ acts irreducibly on $\mathcal{Q}_\lambda$ and trivially on $\mathcal{M}_\lambda$. Since $\Tr(\Pi_\lambda) $ is equal to the dimension of $\mathcal{H}_\lambda$, we find that $\Tr(\Pi_\lambda)=d_\lambda\times m_\lambda$, where $d_\lambda$ is the dimension of irrep $\lambda$ of $\SU(d)$ and $m_\lambda$ is the multiplicity of this irrep, or equivalently, the dimension of the corresponding irrep of $\mathbb{S}_n$. Similarly, $\Tr(\Pi_\lambda \P_{rs})=d_\lambda\times \Tr(\mathbf{p}_\lambda(\sigma_{rs}))$, where $\Tr(\mathbf{p}_\lambda(\sigma_{ij}))$ is known as the irreducible character of $\sigma_{ij} \in \mathbb{S}_n$ associated to the irrep $\lambda$. We conclude that $b_\lambda$ can also be written as
\be
b_\lambda=\frac{\Tr(\mathbf{p}_\lambda(\sigma_{rs}))}{m_\lambda}=\frac{\Tr(\mathbf{p}_\lambda(\sigma_{12}))}{m_\lambda}\ .
\ee 
An immediate corollary of these results is that 
\begin{corollary}\label{cor:center}
  For Hamiltonian $H=\sum_\lambda h_\lambda \Pi_\lambda$ the family of unitaries $\cset{\e^{-\i H t}}{t \in \real}$ is in the group generated by 2-local $\SU(d)$-invariant Hamiltonians if, and only if there exists $\alpha, \beta\in \mathbb{R}$ such that
\be
h_\lambda= \alpha +\beta b_\lambda\ ,
\ee
where
\be
b_\lambda=\frac{z_\lambda}{n(n-1)}= \frac{d c_\lambda+n(n - d^2) }{d n(n-1) } =\frac{\Tr(\P_{12} \Pi_\lambda)}{\Tr(\Pi_\lambda)}=\frac{\Tr(\mathbf{p}_\lambda(\sigma_{12}))}{m_\lambda} \ ,
\ee
where $c_\lambda$ and $z_\lambda$ are, respectively, the eigenvalues of the quadratic Casimir operator $C_2$ and operator $Z=\sum_{r\neq s} \P_{rs}$ in the sector $\lambda$. 
\end{corollary}

Finally, we present the proof of \cref{lem:swap}.
\begin{proof} (\cref{lem:swap}) 
  From Schur-Weyl duality we know that the space of $\SU(d)$-invariant operators on $\complex^d \otimes \complex^d$ is two dimensional, spanned by the identity $\mathbb{I}$ and the swap $\P_{12}$. On the other hand, the operator $F=\sum_{a = 1}^{d^2 - 1} T_1^a \otimes T_2^a$ is also $\SU(d)$-invariant and not proportional to the identity. This can be seen in the following way. Since, for any $U \in \SU(d)$, the map $T^a \mapsto U T^a U^\dagger$ preserves the Hilbert-Schmidt inner product $\Tr (T^{a \dagger} T^b)$ on the space of operators, we must have
\begin{align}
  U T^a U^\dagger = \sum_{b=1}^{d^2-1} R(U)_{ab} T^b,
\end{align}
where $R(U)$ is a unitary matrix. To verify that $R(U)$ is unitary, take the hermitian conjugate and use the assumption that $T^a$ is hermitian to obtain
\begin{equation}
  U T^a U^\dagger = \sum_{b = 1}^{d^2 - 1} R(U)_{ab}^\ast T^b.
\end{equation}
Then,
\begin{equation}
  \sum_{c = 1}^{d^2 - 1} R(U)_{ac} R(U)_{cb}^\dagger = \sum_{c, d = 1}^{d^2 - 1} R(U)_{ac} R(U)_{bd}^\ast \frac{1}{2} \Tr(T^c T^d) = \frac{1}{2} \Tr(U T^a U^\dagger U T^b U^\dagger) = \delta_{ab}.
\end{equation}
Therefore,
\begin{equation}\label{eq:C-inv}
  \sum_{a = 1}^{d^2 - 1} U T^a U^\dagger \otimes U T^{a} U^\dagger = \sum_{a, b, c = 1}^{d^2 - 1} R(U)_{ab} R(U)^\ast_{ac} T^b \otimes T^{c} = \sum_{b, c} \p[\big]{R(U)^\dagger R(U)}_{cb} T^b \otimes T^c = \sum_{b} T^b \otimes T^b,
\end{equation}
and so $(U \otimes U) F (U^\dagger \otimes U^\dagger) = F$. Since, by the assumption that $T^a$ is traceless for all $a$, $\Tr(F) = \sum_a \Tr(T^a)^2 = 0$, we know that $F$ and $\identity$ form an orthogonal basis for the space of $\SU(d)$-invariant operators. Therefore, we must have
\begin{align}\label{eq:Pij-undetermined}
  \P_{12} = \alpha F + \beta\mathbb{I},
\end{align}
where $\alpha$ and $\beta$ are real numbers. Taking the trace of both sides of \cref{eq:Pij-undetermined}, we have
\begin{align}
  d = \beta d^2 \implies \beta = 1/d.
\end{align}
Furthermore, rearranging \cref{eq:Pij-undetermined} and multiplying by $\P_{12}$ gives $\alpha \Tr (F \P_{12}) = \Tr (\identity - \beta \P_{12}) = d^2 - 1$; together with
\begin{equation}
  \Tr (F \P_{12}) = \frac{1}{2} \sum_{a = 1}^{d^2 - 1} \Tr ((T^a \otimes T^a) \P_{12}) = \frac{1}{2} \sum_{a = 1}^{d^2 - 1} \Tr ((T^a)^2) = d^2 - 1,
\end{equation}
this implies that $\alpha = 1$. This completes the proof.
\end{proof}

\end{document}